\documentclass[a4paper,11pt]{article}
\usepackage{bm,mathrsfs,amsmath,amsthm,amssymb,amscd,longtable,array,graphicx}
\newcommand{\nc}{\newcommand}
\nc{\rnc}{\renewcommand}
\nc{\nn}{\nonumber}
\nc{\der}{{\partial}}
\rnc{\Im}{{\rm{Im}\,}}
\rnc{\Re}{{\rm{Re}\,}}
\nc{\db}{\displaybreak[0]\\}
\nc{\bra}{\langle}
\nc{\ket}{\rangle}
\nc{\bs}{\boldsymbol}

\nc{\tcr}{\textcolor{red}}

\newtheorem{theorem}{Theorem}[section]
\newtheorem{lemma}[theorem]{Lemma}

\newtheorem{proposition}[theorem]{Proposition}

\theoremstyle{definition}
\newtheorem{definition}[theorem]{Definition}

\numberwithin{equation}{section}

\numberwithin{equation}{section}

\begin{document}%
%%%%%%%%%%%%%%%%%%%%%%%%%%%%%%%%%%%%%%%%%%%%%%%%%%%%%%%%%
%TITLE
%%%%%%%%%%%%%%%%%%%%%%%%%%%%%%%%%%%%%%%%%%%%%%%%%%%%%%%%%
%
\title{
Quantum inverse scattering method and
generalizations of \\
symplectic Schur functions and Whittaker functions
}

\author{
Kohei Motegi$^1$\thanks{E-mail: kmoteg0@kaiyodai.ac.jp} \,
\ \
Kazumitsu Sakai$^2$\thanks{E-mail: k.sakai@rs.tus.ac.jp} \,
and
Satoshi Watanabe$^3$\thanks{E-mail: watanabe@gokutan.c.u-tokyo.ac.jp}
\\\\
$^1${\it Faculty of Marine Technology,}\\ 
{\it Tokyo University of Marine Science and Technology,}\\
 {\it Etchujima 2-1-6, Koto-Ku, Tokyo, 135-8533, Japan} \\
\\
$^2${\it Department of Physics,
Tokyo University of Science,}\\
 {\it Kagurazaka 1-3, Shinjuku-ku, Tokyo, 162-8601, Japan} \\
\\
$^3${\it Institute of physics, University of Tokyo,} \\ 
{\it Komaba 3-8-1, Meguro-ku, Tokyo 153-8902, Japan}
\\\\
\\
}

%\date{\today}

\maketitle

\begin{abstract}
%% Text of abstract
We introduce generalizations of type $C$ and $B$ ice models
which were recently introduced by Ivanov and Brubaker-Bump-Chinta-Gunnells,
and study in detail the partition functions of the models
by using the quantum inverse scattering method.
We compute the explicit forms of the wavefunctions and their duals
by using the Izergin-Korepin technique, which can be applied to both models.
For type $C$ ice, we show the wavefunctions are expressed using
generalizations of the symplectic Schur functions. This gives
a generalization of the correspondence by Ivanov.
For type $B$ ice, we prove that the exact expressions of the
wavefunctions are given by generalizations of the Whittaker functions
introduced by Bump-Friedberg-Hoffstein. The special case
is the correspondence conjectured by Brubaker-Bump-Chinta-Gunnells.
We also show the factorized forms for
the domain wall boundary partition functions for both models.
As a consequence of the studies of the partition functions,
we obtain dual Cauchy formulas for the generalized symplectic Schur functions
and the generalized Whittaker functions.
\end{abstract}

\section{Introduction}
Integrable lattice models, which are a
special class of mathematical models in statistical physics,
have been playing important roles in the 
developments of combinatorics and representation theory in modern mathematics.
For example, investigating the mathematical structure of the $R$-matrices
and $L$-operators
lead to the discovery of quantum groups \cite{Dr,J}
and the developments of the quantum inverse scattering method
\cite{FST,Baxter,KBI}.

From the point of view of statistical physics,
the most important quantities are partition functions.
Partition functions of integrable lattice models are global quantities
consisting of $L$-operators and under specified boundary conditions.
A particular class called the domain wall boundary partition functions
\cite{Ko,Iz} have found applications to
the enumeration of alternating sign matrices
(see \cite{Bre,Ku} for examples) in the 1990s.

Partition functions are now gaining attention again
in combinatorial representation theory.
It has been realized that analyzing the details of partition functions
lead us to the discovery and refinements of algebraic identities
such as the Cauchy and Littlewood identities
for various types of symmetric functions
like the Schur, Grothendieck, Hall-Littlewood polynomials,
noncommutative versions and so on.
Recently, there are various studies on this topic,
including
\cite{Bogo,BBF,BMN,Lascoux,Mcnamara,KS,Korff,GK2,MSvertex,MSboson,MSW,BW,BWZ,WZ,Borodin,BP1,TakeyamaHecke,Takeyama,vDE,Motegi}.

The most fundamental and important thing to accomplish the detailed study is to
find the exact correspondence between the wavefunctions 
and the symmetric functions,
which depends on the type of the local $L$-operators
and the global boundary conditions.
In the correspondence,
the spectral parameters of the integrable lattice models
play the role of symmetric variables of the corresponding
symmetric functions.

One of the seminal works are by Brubaker-Bump-Friedberg \cite{BBF},
which they showed that the wavefunctions of some
free-fermionic six-vertex model are given as the product
of some factors and the Schur functions.
A free parameter lives in the free-fermionic six-vertex model, which plays
the role of refining the combinatorial formula for the
Schur polynomials. The Tokuyama formula \cite{To,OkTo,HK}
is a combinatorial formula which gives a deformation
of the Weyl character formula, and Brubaker-Bump-Friedberg
found that the wavefunctions of the free-fermionic six-vertex model
gives an integrable model realization of the Tokuyama formula.
The gauge transformed version of the $L$-operator
are related with the Felderhof model 
\cite{Felderhof,Mu,DA}
or the Perk-Schultz model \cite{PS}
which the underlying quantum group structure
are the colored representation of $U_q(sl_2)$ \cite{Mu,DA}
or superalgebra representation \cite{Yamane}.

Today, there are developments on studying the variations and generalizations
of the correspondence. See
\cite{OkTo,HK1,HK2,ZZ,FCWZ,ZYZ,Iv,Ivthesis,Tab,BBCG,BS,BBB,BBBG,LMP,MoIK,Moelliptic} for examples on this topic as well as former related works on this model.
For example, the correspondence between the wavefunctions of the
free-fermionic six-vertex models and the
Schur functions was generalized to the factorial Schur functions
by Bump-McNamara-Nakasuji \cite{BMN}
(including \cite{Lascoux,Mcnamara} as special cases)
by introducing inhomogeneous parameters in the quantum spaces.
They also showed that comparing two expressions
for the domain wall boundary partition functions,
one can derive a dual Cauchy identity for the factorial Schur functions.

One of the most recent developments are to twist the
higher rank Perk-Schultz model \cite{PS}
to construct the metaplectic ice \cite{BBB},
as done by Brubaker-Buciumas-Bump,
which gives combinatorial descriptions
of metaplectic Whittaker functions.
The dual version of the metaplectic ice giving the
same partition functions as the original one was recently constructed by
Brubaker-Buciumas-Bump-Gray \cite{BBBGduality}.
There are also related works on Hecke modules \cite{BBBF},
and the notion of metaplectic ice
was extended to type $C$ by Gray \cite{Gray}
by introducing the metaplectic ice under reflecting boundary.
There is also a recent work on
understanding the partition functions of metaplectic objects
in terms of vertex operators by
Brubaker-Buciumas-Bump-Gustafsson \cite{BBBG}.
Another development is the extension of the work by Brubaker-Bump-Friedberg,
Bump-McNamara-Nakasuji to the elliptic case \cite{Motegi,Moelliptic} by 
introducing and computing the wavefunctions of the
elliptic generalization of the Felderhof model and the Perk-Schultz model,
introduced by Foda-Wheeler-Zuparic \cite{FWZ}
(see also Deguchi-Akutsu \cite{DA})
and Okado \cite{Okado}, Deguchi-Fujii \cite{DF} and Deguchi-Martin \cite{DM}.

As for variations of the correspondence,
the seminal works are by Ivanov \cite{Iv,Ivthesis} and
Brubaker-Bump-Chinta-Gunnells \cite{BBCG} (see also Hamel-King \cite{HK1,HK2} 
for former related works),
in which they studied the wavefunctions under reflecting boundary conditions.
Ivanov \cite{Iv,Ivthesis} introduced and studied the wavefunctions of the six-vertex model
under reflecting boundary (type $C$ ice), which a diagonal boundary $K$-matrix 
is used at the boundary,
and showed that the explicit forms of the wavefunctions
are expressed as a product of factors and symplectic Schur functions.
By using a different boundary $K$-matrix (type $B$ ice),
Brubaker-Bump-Chinta-Gunnells \cite{BBCG}
conjectured the explicit forms of the wavefunctions.
They conjectured that the wavefunctions are given
by the product of factors and
the Whittaker functions introduced by Bump-Friedberg-Hoffstein \cite{BFH}.

In this paper, we introduce and study
in detail generalizations of the partition functions
of the free-fermionic six-vertex model under reflecting boundary,
and derive algebraic identities for
symmetric functions as a consequence of the detailed study.
We use an $L$-operator \cite{MSW,MoIK}
which is a specialization of the one in
Brubaker-Bump-Friedberg \cite{BBF} and which is a one-parameter
deformation of the one used in Bump-McNamara-Nakasuji \cite{BMN}
as the bulk weights, and deal with the partition functions under
reflecting boundary, using two types of $K$-matrices as the boundary weights
which were used in  Brubaker-Bump-Chinta-Gunnells \cite{BBCG} and Ivanov
\cite{Iv,Ivthesis}.
See also \cite{OkTo,Ts,FK,Filali,RK,CMRV,MoRMP,Mouqreflecting,Chinesegroup,Chinesegroup2,Galleasone,GL,Lamers}
for examples on works
on the domain wall boundary partition functions under reflecting boundary
of the $U_q(sl_2)$ six-vertex model and the free-fermionic six-vertex model,
and extensions to the elliptic model and to the wavefunctions.
The elliptic models were first analyzed for the
basic domain wall boundary partition functions without reflecting boundary conditions in \cite{PRS,Ros,YZ}.

We analyze the wavefunctions, the dual wavefunctions
and the domain wall boundary
partition functions in detail using the quantum inverse scattering method.
We use the Izergin-Korepin technique, which was originally initiated by
Korepin \cite{Ko} as a trick to characterize the properties of the polynomials
representing the domain wall boundary partition functions of the $U_q(sl_2)$ six-vertex model.
Izergin later found the explicit determinant representation (Izergin-Korepin determinant) in \cite{Iz}.
The Izergin-Korepin technique was applied to analyze variations
of the domain wall boundary partition functions (see Tsuchiya and Kuperberg \cite{Ku,Ts} for examples on seminal works of variations),
and the method was recently extended to the scalar products by Wheeler
\cite{Wheeler}
and to the wavefunctions by one of the authors \cite{Motegi,MoIK,Moelliptic}.

We use the Izergin-Korepin technique developed recently 
for the wavefunctions to prove the
exact correspondence between the wavefunctions and the symmetric functions.
As for the generalization of the type $C$ ice by Ivanov \cite{Iv,Ivthesis}, the wavefunctions can be also analyzed using
the argument of Ivanov which extends the one by
Brubaker-Bump-Friedberg \cite{BBF}.
However, the argument does not lead to the final answer
for type $B$ ice, as mentioned in the paper by Brubaker-Bump-Chinta-Gunnells
\cite{BBCG}.
We use the Izergin-Korepin technique for the analysis
as this method is applicable to the type $B$ ice as well as the type $C$ ice.
For the Izergin-Korepin technique to work,
one needs to generalize the $L$-operator, which is the one
used in this paper.
We show the exact correspondences between the wavefunctions,
the dual wavefunctions
and the symmetric functions by using this technique,
and the symmetric functions which appear are
generalizations of the (factorial) symplectic Schur functions
and the Bump-Friedberg-Hoffstein
Whittaker functions \cite{BFH}
for the type $C$ ice and type $B$ ice, respectively.
We also show the factorized expressions
for the domain wall boundary partition functions
by using the Izergin-Korepin technique.
Combining the above correspondences, we derive dual Cauchy formulas
for the generalized symplectic Schur functions
and the generalized Bump-Friedberg-Hoffstein Whittaker functions,
using the idea of Bump-McNamara-Nakasuji \cite{BMN}.

This paper is organized as follows.
In the next section, we introduce a generalized $L$-operator
of the free-fermionic six-vertex model.
In section 3, we introduce two types of the wavefunctions under
reflecting boundary, based on two different types 
of diagonal $K$-matrices and
call them as the type I wavefunctions and the type II wavefunctions.

Sections 4 and 5 are devoted to the detailed study
of the type I wavefunctions.
In section 4, we make the Izergin-Korepin analysis
on the wavefunctions and their duals, and derive their explicit forms
and show the correspondence with the generalized symplectic Schur functions.
In section 5, we prove the complete factorized form for
the domain wall boundary partition functions by the Izergin-Korepin analysis,
and derive the dual Cauchy formula
for the generalized symplectic Schur functions as a consequence
of evaluating the domain wall boundary partition functions
in two ways.

In sections 6 and 7, we perform a similar analysis on the
type II partition functions.
We derive the exact correspondence between the wavefunctions
, the dual wavefunctions of type II
and the generalized Bump-Friedberg-Hoffstein Whittaker functions.
We also derive the factorized form of 
the type II domain wall boundary partition functions,
and derive the dual Cauchy formulas for
the generalized Bump-Friedberg-Hoffstein Whittaker functions.
Section 8 is devoted to conclusion. Some calculations are deferred to
Appendix.

\section{The free-fermionic six-vertex model}

In this section, we introduce the $R$-matrix and the $L$-operator
of the free-fermionic six-vertex model \cite{MSW,MoIK},
and explain that the $L$-operator
$L_{aj}(z,t,\alpha_j,\gamma_j)$ \eqref{generalizedloperator}
which we use as the bulk elements of the wavefunctions in this paper
is a specialization of the one in
Brubaker-Bump-Friedberg \cite{BBF} and is a one-parameter
deformation of the one used in Bump-McNamara-Nakasuji \cite{BMN}.

The $R$-matrix of the free-fermionic six-vertex model is given as
\begin{eqnarray}
R(z,p,q)=\left( 
\begin{array}{cccc}
1-pqz & 0 & 0 & 0 \\
0 & -p^2(1-p^{-1}qz) & 1-q^2 & 0 \\
0 & (1-p^2)z & z-p^{-1}q & 0 \\
0 & 0 & 0 & z-pq
\end{array}
\right), \label{generalizedrmatrix}
\end{eqnarray}
which acts on the tensor product $W \otimes W$ of the
complex two-dimensional vector space $W=\mathbb{C}^2$
spanned by the 
``empty state"  $|0\ket={1 \choose 0}$ and the 
``particle occupied state" $|1\ket={0 \choose 1}$. 
The generalized $R$-matrix \eqref{generalizedrmatrix}
can be shown to satisfy the Yang-Baxter relation
\begin{align}
R_{12}(z_1/z_2,p_1,p_1)&R_{13}(z_1,p_1,p_2)R_{23}(z_2,p_1,p_2) \nn \\
&=R_{23}(z_2,p_1,p_2)R_{13}(z_1,p_1,p_2)R_{12}(z_1/z_2,p_1,p_1),
\label{generalizedyangbaxter}
\end{align}
holding in $\mathrm{End}(W_1 \otimes W_2 \otimes W_3)$.

There is one explanation
why one can introduce at least two parameters $p$ and $q$
besides the spectral parameter $z$ in the generalized $R$-matrix
\eqref{generalizedrmatrix} from the point of view of quantum groups.
There is a class of an exotic quantum group called the
colored representation or the nilpotent representation \cite{Mu,DA}.
The colored representation is a finite-dimensional
highest weight representation which exists when the
parameter of the quantum group is fixed at roots of unity.
Each colored representation space is allowed to have
a free parameter, and since the $R$-matrix is understood as
an intertwiner acting on the tensor product of two representation spaces,
one can include two free parameters.

The two types of the Boltzmann weights in \cite{BBF}
can be given as reductions from the
generalized $R$-matrix \eqref{generalizedrmatrix}.
The first type of the Boltzmann weight in \cite{BBF} comes by specializing $q$
to $q=p$:
\begin{eqnarray}
R(z,t)=R(z,p,p)=\left( 
\begin{array}{cccc}
1+tz & 0 & 0 & 0 \\
0 & t(1-z) & t+1 & 0 \\
0 & (t+1)z & z-1 & 0 \\
0 & 0 & 0 & z+t
\end{array}
\right), \label{rmatrix}
\end{eqnarray}
where we set $t=-p^2$.

The second type of Boltzmann weight is given by
taking $q=0$:
\begin{eqnarray}
L(z,t)=R(z,p,0)=\left( 
\begin{array}{cccc}
1 & 0 & 0 & 0 \\
0 & t & 1 & 0 \\
0 & (t+1)z & z & 0 \\
0 & 0 & 0 & z
\end{array}
\right). \label{loperator}
\end{eqnarray}

Under this specialization, the generalized Yang-Baxter relation 
\eqref{generalizedyangbaxter} is rewritten as
\begin{align}
R_{ab}(z_1/z_2,t)L_{aj}(z_1,t)L_{bj}(z_2,t)
=L_{bj}(z_2,t)L_{aj}(z_1,t)R_{ab}(z_1/z_2,t), \label{RLL}
\end{align}
acting on $W_a \otimes W_b \otimes V_j$.
Here, $V$ and $W$ are both complex two-dimensional
vector spaces $V=W=\mathbb{C}^2$.
By convention we call $W$ and $V$ the auxiliary space and
quantum space, respectively.
The operator \eqref{rmatrix} acting on the tensor product
of two auxiliary spaces $W \otimes W$ is called the $R$-matrix,
and the operator \eqref{loperator} acting on the tensor product
of one auxiliary and one quantum spaces $W \otimes V$
is called the $L$-operator.
A class of the Yang-Baxter relation \eqref{RLL}
is usually referred to as the $RLL$ relation.

A certain class of partition functions
which we call wavefunctions, which we describe in more detail
in the next section, is constructed from the $L$-operator.
It was shown by Brubaker-Bump-Friedberg \cite{BBF} that the wavefunctions for the basic case
(without reflecting boundary)
constructed from the $L$-operator \eqref{loperator}
are expressed in terms of Schur functions.
Later, the construction was extended to the factorial Schur functions
by Bump-McNamara-Nakasuji \cite{BMN}.
From the quantum integrable point of view, the crucial point
to construct the factorial Schur functions
was to generalize the $L$-operator \eqref{loperator}
by keeping the $RLL$ relation \eqref{RLL}.

We find a further generalization of the $L$-operator \cite{MSW,MoIK}
satisfying the $RLL$ relation which is given by
\begin{eqnarray}
L_{aj}(z,t,\alpha_j,\gamma_j)=\left( 
\begin{array}{cccc}
1-\gamma_j z & 0 & 0 & 0 \\
0 & t+\gamma_j z & 1 & 0 \\
0 & (t+1)z & \alpha_j+(1-\alpha_j \gamma_j)z & 0 \\
0 & 0 & 0 & -t \alpha_j+(1-\alpha_j \gamma_j)z
\end{array}
\right), \label{generalizedloperator}
\end{eqnarray}
acting on $W_a \otimes V_j$. The parameters $\alpha_j$ and $\gamma_j$ can be
regarded as parameters associated with the quantum space $V_j$.
See Figure \ref{pictureloperator} for a graphical representation
of \eqref{generalizedloperator}. 
Hereafter we call this $L$-operator 
type $\Gamma$ $L$-operator. The $L$-operator whose basic wavefunctions
give the factorial Schur functions
\cite{BMN}
is a special limit of the generalized $L$-operator
\eqref{generalizedloperator}
\begin{eqnarray}
L_{aj}(z,t,\alpha_j)=L_{aj}(z,t,\alpha_j,0)=\left( 
\begin{array}{cccc}
1 & 0 & 0 & 0 \\
0 & t & 1 & 0 \\
0 & (t+1)z & \alpha_j+z & 0 \\
0 & 0 & 0 & -t \alpha_j+z
\end{array}
\right)
\label{factorialloperator}. 
\end{eqnarray}

Let us make some comments on the $L$-operator
\eqref{generalizedloperator}.
The $L$-operator
\eqref{generalizedloperator} is in fact a special case of
a class of $L$-operators satisfying the free-fermion condition
and satisfying the $RST$ relation of Theorems 3 and 4 in
Brubaker-Bump-Friedberg \cite{BBF}
(see also Proposition 1 in Bump-McNamara-Nakasuji \cite{BMN}).
What we call as the $L$-operator or the vertex which the $L$-operator
is associated with corresponds to
the operators $S$, $T$ in Theorem 3 in \cite{BBF},
and the vertices with free-fermionic Boltzmann weights $v$ and $w$
in Proposition 1 in \cite{BMN}.
The $R$-matrix we call in this paper
or the vertex which the $R$-matrix is associated with
essentially corresponds to the operator $R$ in Theorem 3 in \cite{BBF},
and the vertex with the Boltzmann weights $u$
in Proposition 1 in \cite{BMN} (note that there is
a freedom to multiply the $R$-matrix by overall factors
which keeps the $RLL$ relation).
For example, it is explained in the Proof of Lemma 1
in \cite{BMN} that to check the
$L$-operator \eqref{factorialloperator} and the $R$-matrix \eqref{rmatrix}
(multiplied by an overall factor) satisfy the relation
\begin{eqnarray}
a_1(u)=a_1(v)a_2(w)+b_2(v)b_1(w),
\end{eqnarray}
which is one of the relations first stated in \cite{BBF}
($a_1(u), a_1(v), a_2(w),b_2(v), b_1(w)$ are notations used in
\cite{BMN})
is nothing but to check the relation
\begin{eqnarray}
tz_i+z_k=1\cdot(z_k-t \alpha_j)
+(z_i+\alpha_j)t.
\end{eqnarray}
For the case of the $L$-operator \eqref{generalizedloperator}
which is a one-parameter deformation of the one \eqref{factorialloperator},
checking the relation now becomes to check the following relation holds:
\begin{eqnarray}
tz_i+z_k=(1-\gamma_j z_i)(-t \alpha_j+(1-\alpha_j \gamma_j)z_k)
+(\alpha_j+(1-\alpha_j \gamma_j)z_i)(t+\gamma_j z_k).
\end{eqnarray}
One can easily see that this identity still holds.
One can also check the other five relations in \cite{BBF}, \cite{BMN}.
It is also easy to check that the $L$-operator \eqref{generalizedloperator}
satisfy the free-fermion condition.
We use the $L$-operator \eqref{generalizedloperator}
since we investigate the partition functions in this paper
by the Izergin-Korepin method which we view as functions
in the variables $\gamma_j$s.

\begin{figure}[hh]
\centering
\includegraphics[width=0.95\textwidth]{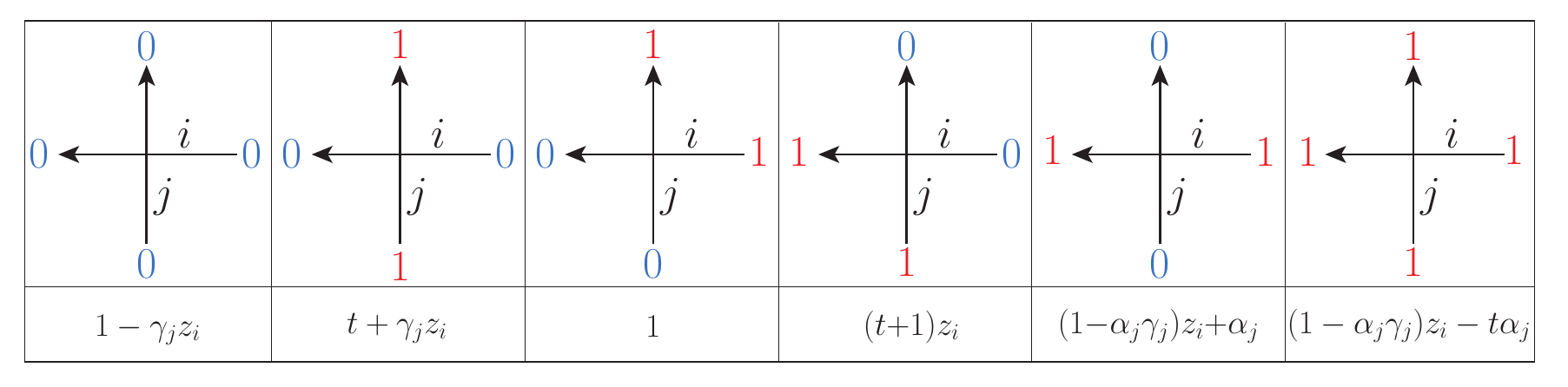}
\caption{Type $\Gamma$ $L$-operator 
\eqref{generalizedloperator}.}
\label{pictureloperator}
\end{figure}

\section{Type I and type II
wavefunctions under reflecting boundary}
In this section, we introduce two types of the wavefunctions
under reflecting boundary
by using $K$-matrices of
Ivanov \cite{Iv,Ivthesis} and Brubaker-Bump-Chinta-Gunnells
\cite{BBCG}.
The wavefunctions are constructed from
the double monodromy matrices.
To introduce double monodromy matrices,
we introduce the following another $L$-operator
(we call it type $\Delta$ $L$-operator 
$\widetilde{L}_{aj}(z,t,\alpha_j,\gamma_j)$
at the $j$th site
in the quantum space (Figure \ref{pictureanotherloperator})
\begin{eqnarray}
\widetilde{L}_{aj}(z,t,\alpha_j,\gamma_j)=\left( 
\begin{array}{cccc}
\alpha_j+(1-\alpha_j \gamma_j)z & 0 & 0 & 0 \\
0 & t(1-\alpha_j \gamma_j)z-\alpha_j & 1 & 0 \\
0 & (t+1)z & 1-\gamma_j z & 0 \\
0 & 0 & 0 & t \gamma_j z+1
\end{array}
\right). \label{anothergeneralizedloperator}
\end{eqnarray}
Note that the type $\Gamma$ and $\Delta$ operators satisfy the
following Yang-Baxter relations
\begin{equation}
\includegraphics[width=0.75\textwidth]{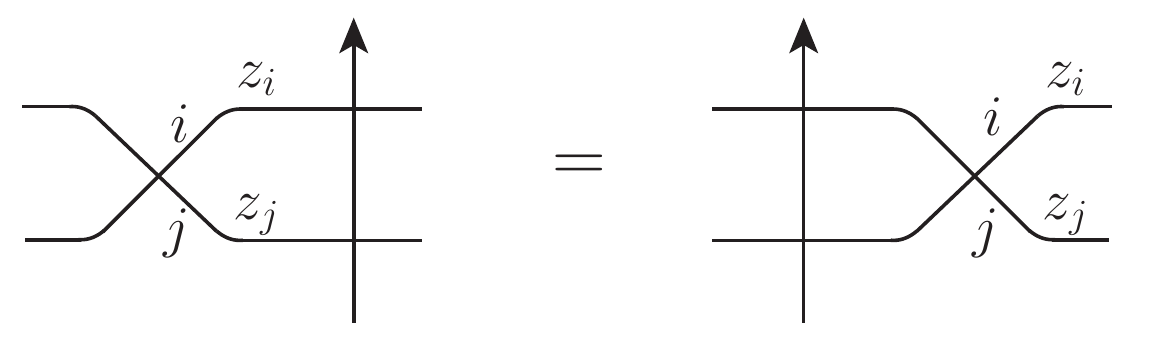}
\label{YBR}
\end{equation}
where the $R$-matrices are defined in Figure \ref{pic-vertex}.
At the boundary, we use boundary weights
called as the $K$-matrix. We use two types of $K$-matrices
in this paper.
The first $K$-matrix $K^{\rm I}_{a}(z,t,\alpha_0,\gamma_0)$ is given by
(Figure \ref{picturekmatrixone})
\begin{eqnarray}
K^{\rm I}_{a}(z,t,\alpha_0,\gamma_0)=\left( 
\begin{array}{cc}
(1-\alpha_0 \gamma_0)tz-\alpha_0 & 0 \\
0 & (1-\alpha_0 \gamma_0)z^{-1}+\alpha_0 \\
\end{array}
\right), \label{generalizedkmatrix}
\end{eqnarray}
where $\alpha_0$, $\gamma_0$ is a free parameter.
\eqref{generalizedkmatrix} is a generalization 
of the one given by Ivanov.
Setting $\alpha_0=\gamma_0=0$, \eqref{generalizedkmatrix}
reduces to the $K$-matrix used in \cite{Iv,Ivthesis}. The $K$-matrix \eqref{generalizedkmatrix}
satisfies the reflection equation \cite{Sklyanin}
\begin{equation}
\includegraphics[width=0.75\textwidth]{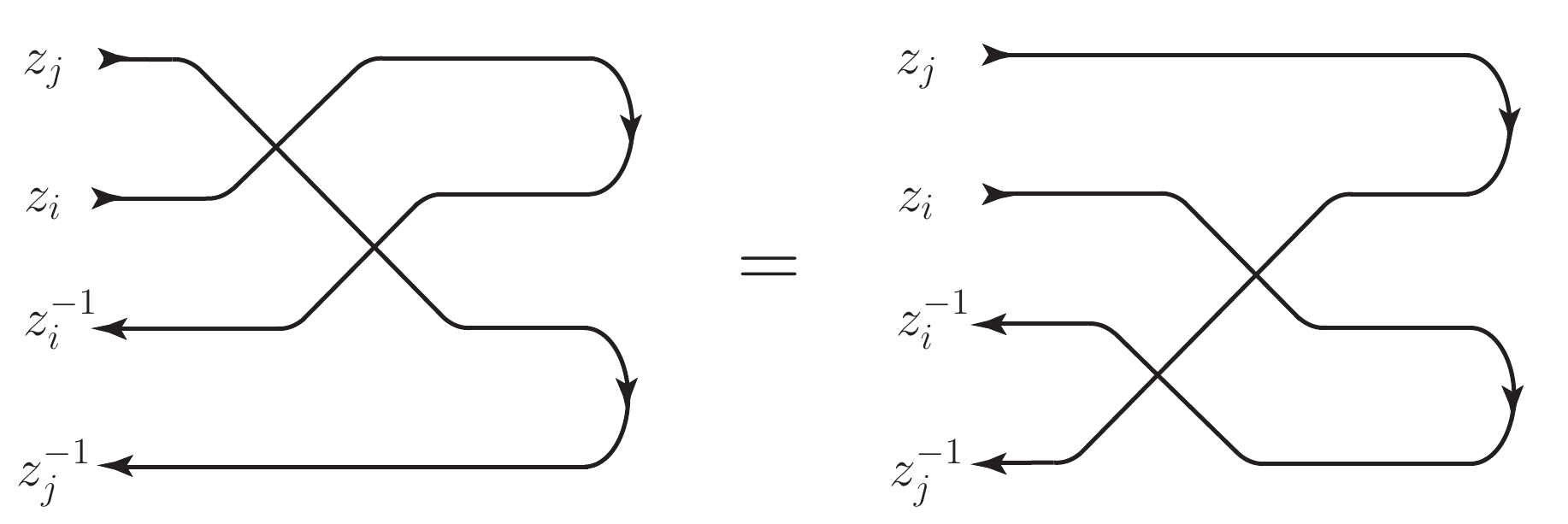}
\label{RE}
\end{equation}

\begin{figure}[b]
\centering
\includegraphics[width=0.95\textwidth]{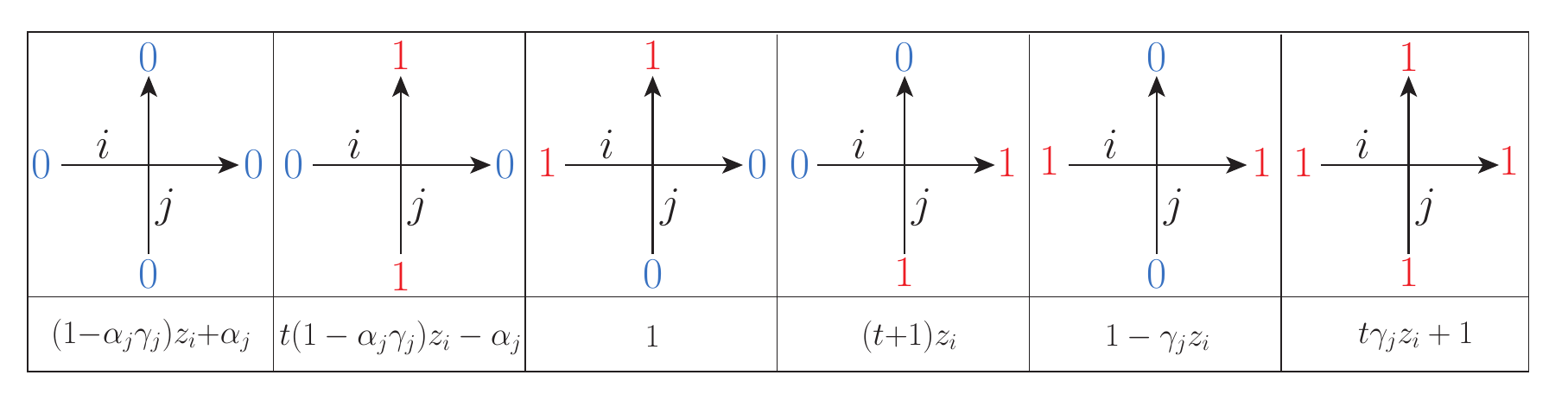}
\caption{Type $\Delta$ $L$-operator
\eqref{anothergeneralizedloperator}.}
\label{pictureanotherloperator}
\end{figure}

\begin{figure}[ttt]
\centering
\includegraphics[width=0.95\textwidth]{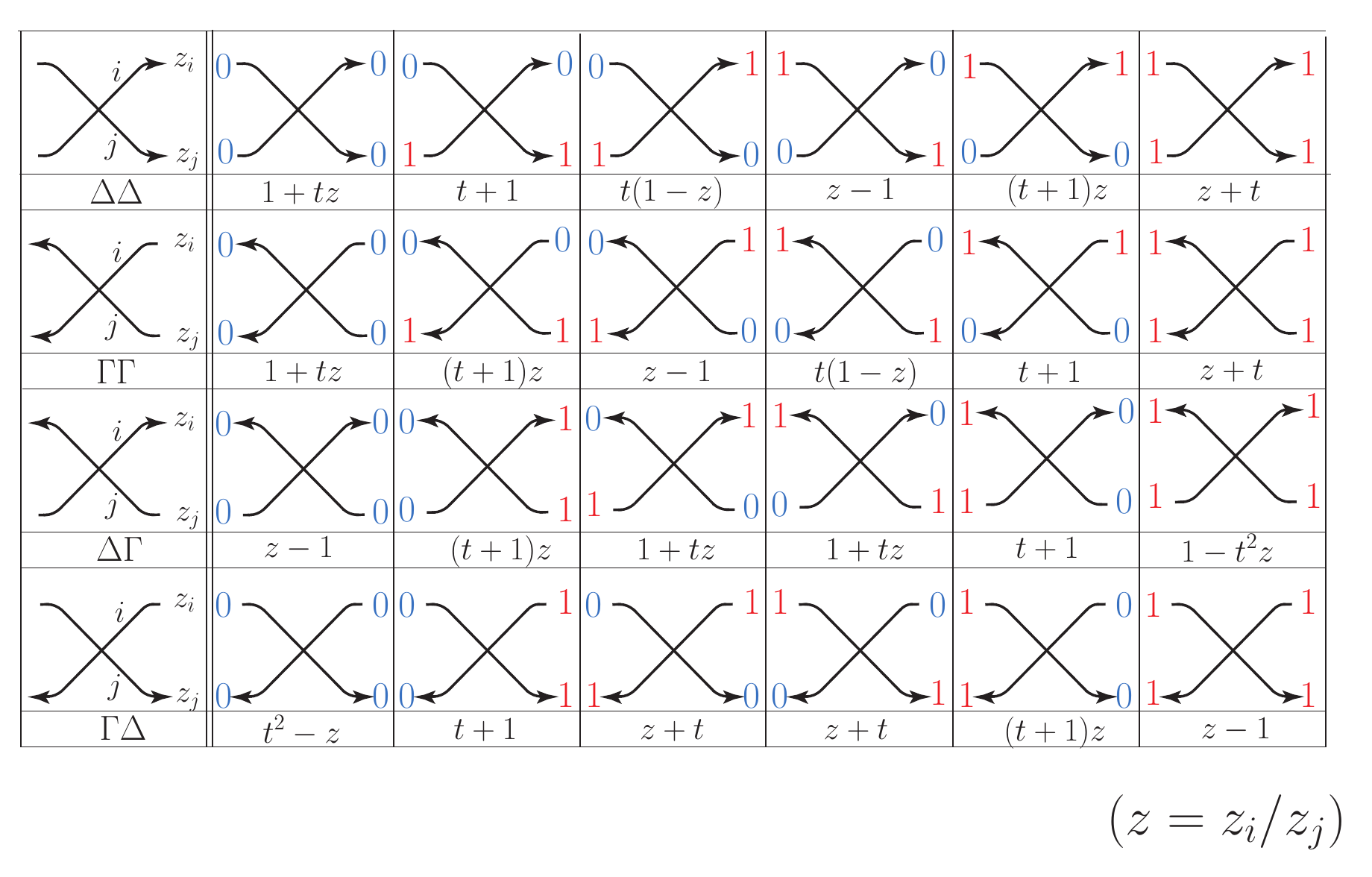}
\caption{$R$-matrices.}
\label{pic-vertex}
\end{figure}

\begin{figure}[b]
\centering
\includegraphics[width=0.3\textwidth]{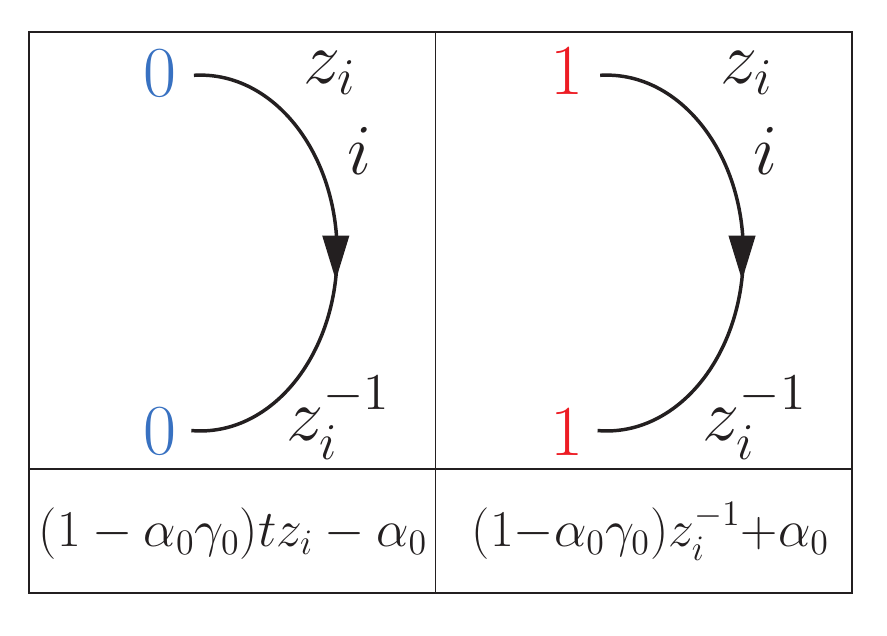}
\caption{The type I $K$-matrix
\eqref{generalizedkmatrix}.}
\label{picturekmatrixone}
\end{figure}

\begin{figure}[b]
\centering
\includegraphics[width=0.3\textwidth]{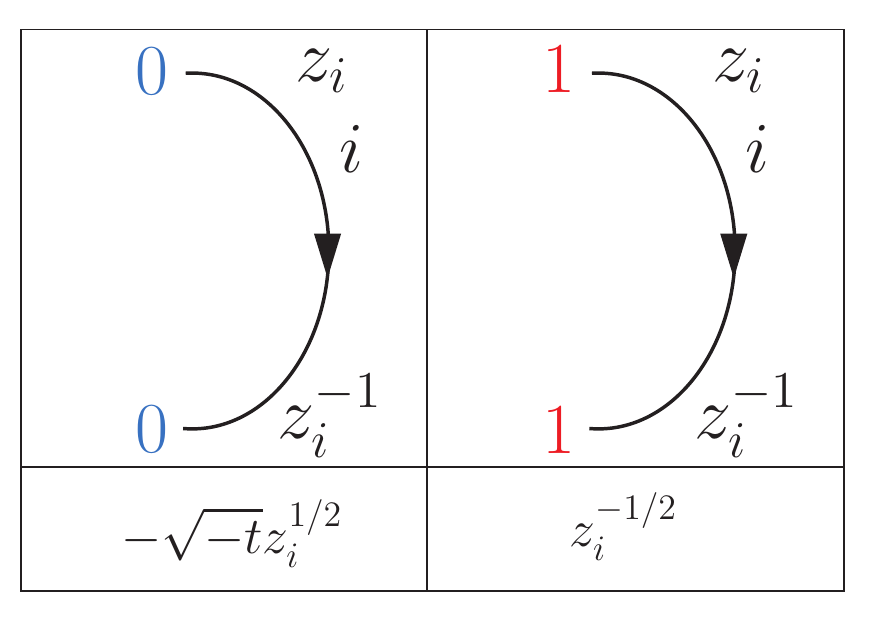}
\caption{The type II $K$-matrix
\eqref{generalizedkmatrixsecond}.}
\label{picturekmatrixtwo}
\end{figure}

\begin{figure}[ht]
\centering
\includegraphics[width=0.7\textwidth]{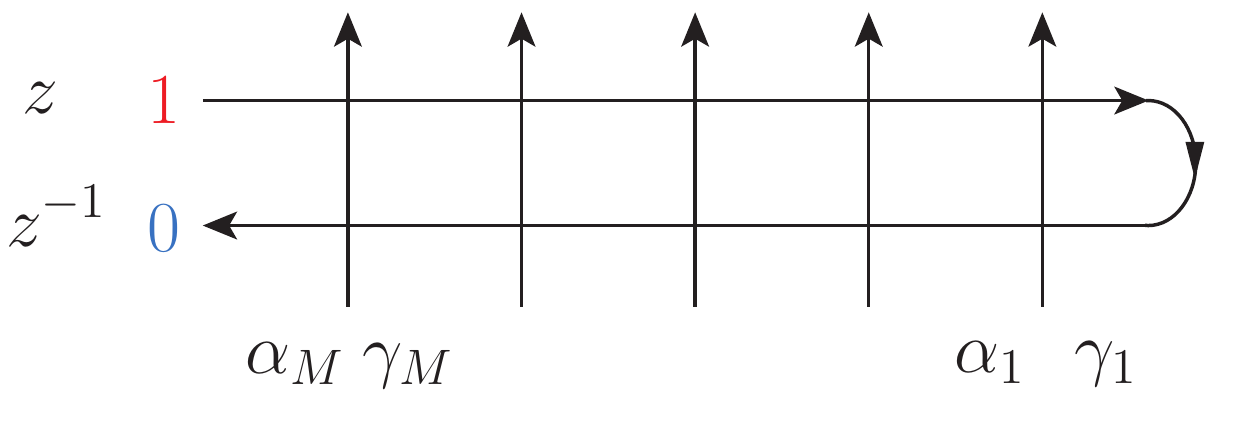}
\caption{The double row monodromy matrices $\mathcal{B}^{\rm I}(z,\{ \overline{\alpha} \},\{ \overline{\gamma} \})$
\eqref{generalizeddoublerow},
$\mathcal{B}^{\rm II}(z,\{ \alpha \},\{ \gamma \})$
\eqref{generalizeddoublerowsecond}.}
\label{picturedoublerow}

\end{figure}

We use another diagonal $K$-matrix in this paper.
The second $K$-matrix $K^{\rm II}_{a}(z,t)$ is the one introduced by
Brubaker-Bump-Chinta-Gunnells \cite{BBCG}
(Figure \ref{picturekmatrixtwo})
\begin{eqnarray}
K^{\rm II}_{a}(z,t)=\left( 
\begin{array}{cc}
-\sqrt{-t}z^{1/2} & 0 \\
0 & z^{-1/2} \\
\end{array}
\right), \label{generalizedkmatrixsecond}
\end{eqnarray}
which also satisfies the reflection equation \eqref{RE}.

From the generalized $L$-operators
\eqref{generalizedloperator} and \eqref{anothergeneralizedloperator},
one constructs two types of monodromy matrices
\begin{align}
T_{a}(z,\{ \alpha \},\{ \gamma \})=L_{a M}(z^{-1},t,\alpha_M,\gamma_M) \cdots L_{a 1}(z^{-1},t,\alpha_1,\gamma_1)
,
\label{generalizedmonodromy1}
\end{align}
and
\begin{align}
\widetilde{T}_{a}(z,\{ \alpha \},\{\gamma \})=\widetilde{L}_{a 1}(z,t,\alpha_1,\gamma_1) \cdots \widetilde{L}_{a M}(z,t,\alpha_M,\gamma_M)
,
\label{generalizedmonodromy2}
\end{align}
which act on $W_a \otimes (V_1\otimes\dots\otimes 
V_{M})$,
We denote the matrix elements of the two monodromy matrices as
\begin{align}
A(z,\{ \alpha \},\{\gamma \})={}_a \langle 0|T_{a}(z, \{ \alpha \},\{\gamma \})|0 \rangle_{a}, \nn \\
B(z,\{ \alpha \},\{\gamma \})={}_a \langle 0|T_{a}(z, \{ \alpha \},\{\gamma \})|1 \rangle_{a},
\end{align}
and
\begin{align}
\widetilde{A}(z,\{ \alpha \},\{\gamma \})={}_a \langle 1|\widetilde{T}_{a}(z,\{ \alpha \},\{\gamma \})|1 \rangle_{a}, \nn \\
\widetilde{B}(z,\{ \alpha \},\{\gamma \})={}_a \langle 0|\widetilde{T}_{a}(z,\{ \alpha \},\{\gamma \})|1 \rangle_{a},
\end{align}
where, $\{ \alpha \}=\{\alpha_1,\dots,\alpha_M \}$
and $\{\gamma \}=\{\gamma_1,\dots,\gamma_M \}$.

Next, we introduce the double row $B$-operators.
We define two types depending on which $K$-matrix we use
at the boundary.
For type I $K$-matrix \eqref{generalizedkmatrix},
we define the type I double row $B$-operator
$\mathcal{B}^{\rm I}(z,\{ \overline{\alpha} \},\{ \overline{\gamma} \})$ as
\begin{align}
\mathcal{B}^{\rm I} (z,\{ \overline{\alpha} \},\{ \overline{\gamma} \}) 
&=\widetilde{B}(z,\{\alpha \},\{\gamma \}) {}_a \langle 0| K_a^{\rm I}
(z,t,\alpha_0,\gamma_0)|0 \rangle_a A(z,\{\alpha \},\{\gamma \}) \nonumber \\
&\qquad +\widetilde{A}(z,\{ \alpha \},\{ \gamma \}) {}_a \langle 1| 
K_a^{\rm I}(z,t,\alpha_0,\gamma_0)|1 \rangle_a B(z,\{ \alpha \},\{ \gamma \}) \nn \\ 
&=
\left\{(1-\alpha_0 \gamma_0)tz-\alpha_0\right\}
\widetilde{B}(z,\{ \alpha \},\{ \gamma \})A(z,\{ \alpha \},\{ \gamma \})
\nonumber \\
&\qquad +
\left\{(1-\alpha_0 \gamma_0)z^{-1}+\alpha_0\right\}
\widetilde{A}(z,\{ \alpha \},\{ \gamma \}) B(z,\{ \alpha \},\{ \gamma \}).
\label{generalizeddoublerow}
\end{align}
Here, the set of parameters $\{ \overline{\alpha} \}$
and $\{ \overline{\gamma} \}$ means
$\{ \overline{\alpha} \}=\{\alpha_0\}\sqcup\{\alpha\}$
and
$\{ \overline{\gamma} \}=\{\gamma_0\}\sqcup\{\gamma\}$,
respectively.

For type II $K$-matrix \eqref{generalizedkmatrixsecond},
we similarly define $\mathcal{B}^{\rm II}(z,\{ \alpha \},\{ \gamma \})$ as
\begin{align}
\mathcal{B}^{\rm II}(z,\{ \alpha \},\{ \gamma \}) 
=&\widetilde{B}(z,\{\alpha \},\{\gamma \}) {}_a \langle 0| K_a^{\rm II}(z,t)|0 \rangle_a A(z,\{\alpha \},\{\gamma \}) \nonumber \\
&\qquad +\widetilde{A}(z,\{ \alpha \},\{ \gamma \}) {}_a \langle 1| K_a^{\rm II}(z,t)|1 
\rangle_a B(z,\{ \alpha \},\{ \gamma \}) \nn \\
=&
-\sqrt{-t}z^{1/2}
\widetilde{B}(z,\{ \alpha \},\{ \gamma \})A(z,\{ \alpha \},\{ \gamma \})
\nonumber \\
&\qquad +
z^{-1/2}
\widetilde{A}(z,\{ \alpha \},\{ \gamma \}) B(z,\{ \alpha \},\{ \gamma \}).
\label{generalizeddoublerowsecond}
\end{align}
See Figure \ref{picturedoublerow} for graphical representations
of the double row $B$-operators
\eqref{generalizeddoublerow}, \eqref{generalizeddoublerowsecond}.

\begin{figure}[ttt]
\centering
\includegraphics[width=0.6\textwidth]{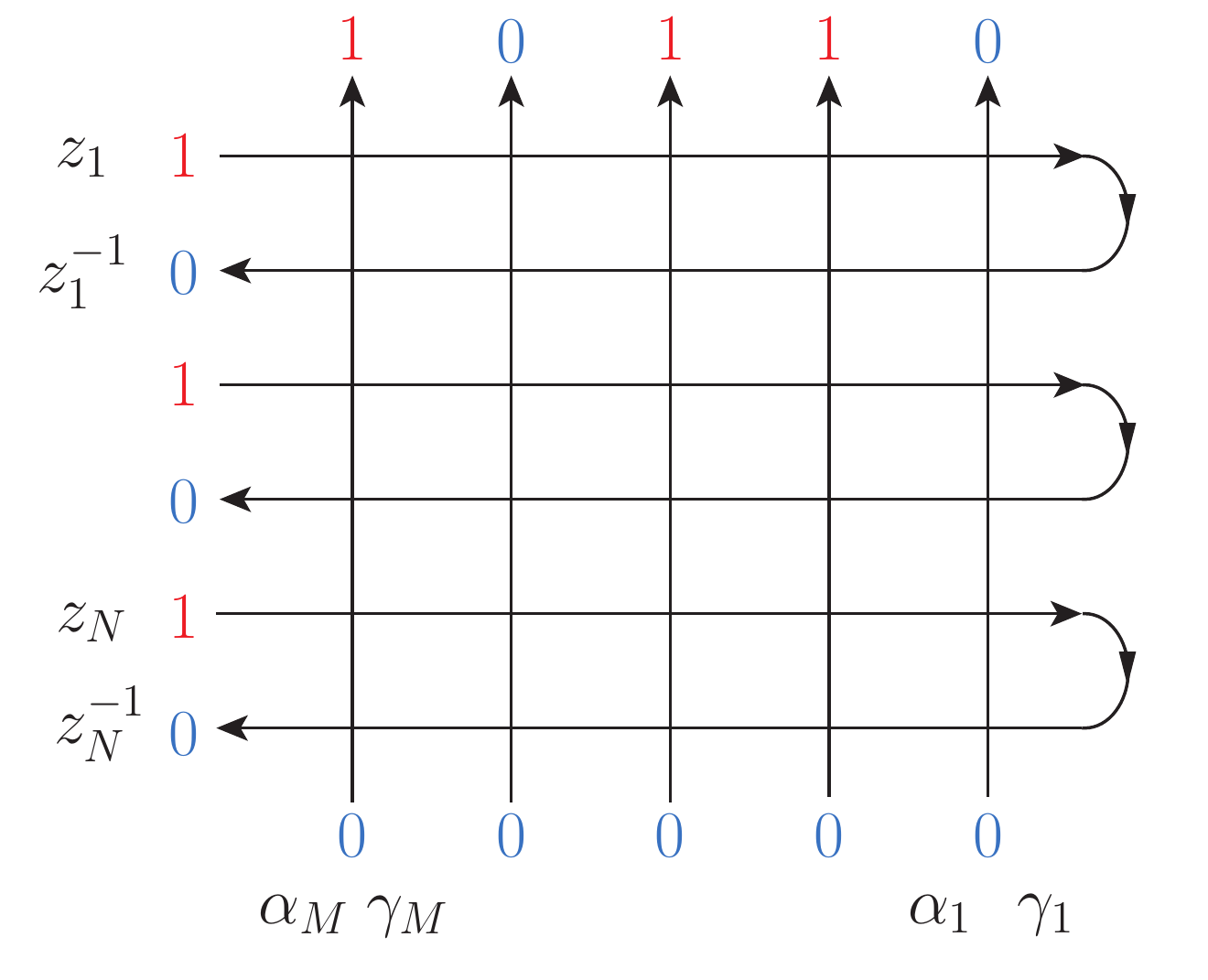}
\caption{The wavefunctions $\Phi^{\rm I}_{M,N}$ \eqref{wavefunctions},
$\Phi^{\rm II}_{M,N}$
\eqref{typetwowavefunctions}
under reflecting boundary.
The figure illustrates the case $M=5$, $N=3$, $x_1=2$, $x_2=3$, $x_3=5$.
}
\label{picturewavefunction}

\end{figure}

\begin{figure}[ht]
\centering
\includegraphics[width=0.6\textwidth]{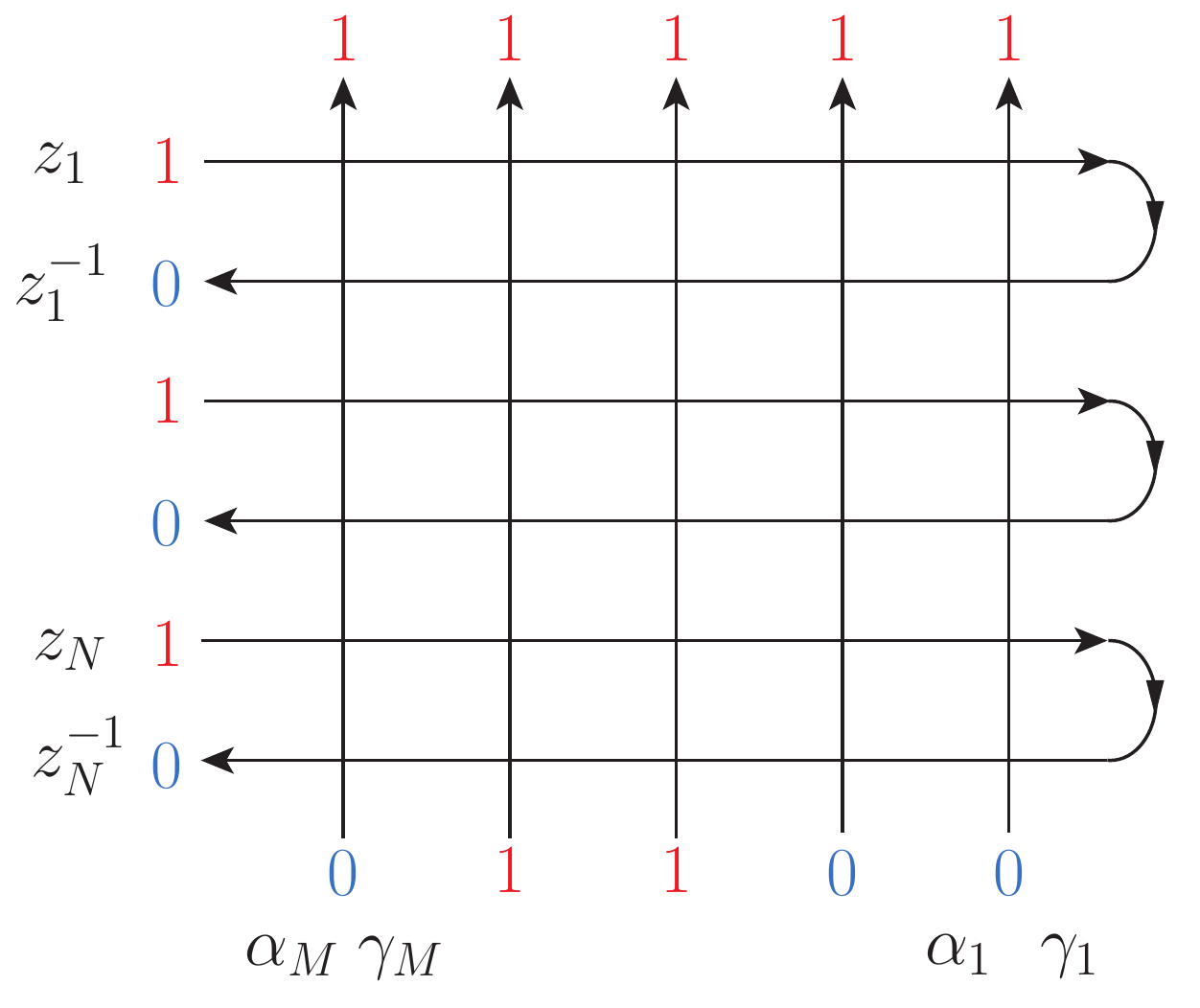}
\caption{The dual wavefunctions $\overline{\Phi}^{\rm I}_{M,N}$ 
\eqref{dualwavefunctions},
$\overline{\Phi}^{\rm II}_{M,N}$
\eqref{typetwodualwavefunctions}
under reflecting boundary.
The figure illustrates the case $M=5$, $N=3$, $\overline{x_1}=1$,
$\overline{x_2}=2$, $\overline{x_3}=5$.
}
\label{picturedualwavefunction}

\end{figure}

We also introduce notations for special states
in the tensor product of the Fock spaces
$V_1 \otimes \cdots \otimes V_M$
and its dual $(V_1 \otimes \cdots \otimes V_M)^*$
as
\begin{alignat}{2}
| 0^{M} \rangle&:=|0\ket_1\
\otimes \dots \otimes |0\ket_M, \qquad
&& | 1^{M} \rangle:=|1\ket_1\
\otimes \dots \otimes |1\ket_M, \\
\langle 0^{M}|&:=
{}_1\bra 0|\otimes\dots \otimes{}_M\bra 0|, \qquad
&& \langle 1^{M}|:=
{}_1\bra 1|\otimes\dots \otimes{}_M\bra 1|.
\end{alignat}
By acting the operators $\sigma_j^+$ and $\sigma_j^-$ defined by
\begin{alignat}{4}
&\sigma^+_j|1 \rangle_j= |0 \rangle_j, \quad
&&\sigma^+_j|0 \rangle_j= 0, \quad
&&{}_j \langle 1|\sigma^+_j=0, \quad
&&{}_j \langle 0|\sigma^+_j={}_j \langle 1|, \nn \\
&\sigma^-_j|0 \rangle_j=|1 \rangle_j, \quad
&&\sigma^-_j|1 \rangle_j =0, \quad
&&{}_j \langle 0|\sigma^-_j=0, \quad
&&{}_j \langle 1|\sigma^-_j={}_j \langle 0|,
\end{alignat}
on $|0^M \rangle$, $\langle 0^M|$,  $| 1^M\ket$,
we introduce states
\begin{align}
\langle x_1 \cdots x_N|
=
\langle 0^M|
\prod_{j=1}^N \sigma^+_{x_j},\quad 
|x_1 \cdots x_N \rangle
=
\prod_{j=1}^N \sigma^-_{x_j}|0^M \rangle,
\label{particleconfiguration} 
\end{align}
for integers $x_1,\dots,x_N$ satisfying
$1 \le x_1 < x_2 < \cdots < x_N \le M$,
and
\begin{align}
|\overline{x_1} \cdots \overline{x_N} \rangle
&=
\prod_{j=1}^N \sigma^+_{\overline{x_j}}|1^M \rangle,
\label{dualholeconfiguration}
\end{align}
for integers $\overline{x_1},\dots,\overline{x_N}$
satisfying
$1 \le \overline{x_1}< \overline{x_2} < \cdots < \overline{x_N} \le M$.

We are now in a position to introduce the wavefunctions.
We introduce the type I wavefunctions
$\Phi^{\rm I}_{M,N}(z_1,\dots,z_N|\gamma_1,\dots,\gamma_M|x_1,\dots,x_N)$
by acting the type I double-row $B$-operators
$\mathcal{B}^{\rm I}(z_j,\{ \overline{\alpha} \},\{ \overline{\gamma} \})$
($j=1,\dots,N$) \eqref{generalizeddoublerow}
on the state $|0^M \rangle$
and taking the inner product with
$\langle x_1 \cdots x_N|$ (Figure \ref{picturewavefunction}):
\begin{align}
\Phi^{\rm I}_{M,N}&(z_1,\dots,z_N|\gamma_1,\dots,\gamma_M|x_1,\dots,x_N) \nonumber \\
&=\langle x_1 \cdots x_N|\mathcal{B}^{\rm I}(z_1,\{ \overline{\alpha} \},\{ \overline{\gamma} \}) \cdots
\mathcal{B}^{\rm I}(z_N,\{ \overline{\alpha} \},\{ \overline{\gamma} \})|0^M \rangle.
\label{wavefunctions}
\end{align}
Likewise, we define the dual wavefunctions of type I
$\overline{\Phi}^{\rm I}_{M,N}(z_1,\dots,z_N|\gamma_1,\dots,\gamma_M|x_1,\dots,x_N)$ as
(Figure \ref{picturedualwavefunction})
\begin{align}
\overline{\Phi}^{\rm I}_{M,N}&(z_1,\dots,z_N|\gamma_1,\dots,
\gamma_M|\overline{x_1},\dots,\overline{x_N}) \nonumber \\
&=\langle 1^M|\mathcal{B}^{\rm I}(z_1,\{ \overline{\alpha} \},\{ \overline{\gamma} \}) 
\cdots \mathcal{B}^{\rm I}(z_N,\{ \overline{\alpha} \},\{ \overline{\gamma} \})|
\overline{x_1} \cdots \overline{x_N} \rangle.
\label{dualwavefunctions}
\end{align}

Similarly, 
we define the type II wavefunctions
$\Phi^{\rm II}_{M,N}(z_1,\dots,z_N|\gamma_1,\dots,\gamma_M|x_1,\dots,x_N)$
(Figure \ref{picturewavefunction})
and the dual wavefunctions
$\overline{\Phi}^{\rm II}_{M,N}
(z_1,\dots,z_N|\gamma_1,\dots,\gamma_M|\overline{x_1},\dots,\overline{x_N})$
(Figure \ref{picturedualwavefunction})
by using the type II double-row $B$-operators 
$\mathcal{B}^{\rm II}(z_N,\{ \alpha \},\{ \gamma \})$ 
\eqref{generalizeddoublerowsecond} as
\begin{align}
\Phi^{\rm II}_{M,N}&(z_1,\dots,z_N|\gamma_1,\dots,\gamma_M|x_1,\dots,x_N) \nonumber \\
&=\langle x_1 \cdots x_N|\mathcal{B}^{\rm II}(z_1,\{ \alpha \},\{ \gamma \}) \cdots
\mathcal{B}^{\rm II}(z_N,\{ \alpha \},\{ \gamma \})|0^M \rangle,
\label{typetwowavefunctions}
\end{align}
and
\begin{align}
\overline{\Phi}^{\rm II}_{M,N}&(z_1,\dots,z_N|\gamma_1,\dots,\gamma_M|\overline{x_1},
\dots,\overline{x_N}) \nonumber \\
&=\langle 1^M|\mathcal{B}^{\rm II}(z_1,\{ \alpha \},\{ \gamma \}) \cdots
\mathcal{B}^{\rm II}(z_N,\{ \alpha \},\{ \gamma \})|\overline{x_1} \cdots \overline{x_N} \rangle.
\label{typetwodualwavefunctions}
\end{align}
The wavefunctions of type I
$\Phi^{\rm I}_{M,N}(z_1,\dots,z_N|\gamma_1,\dots,\gamma_M|x_1,\dots,x_N)$
and
type II \\
$\Phi^{\rm II}_{M,N}(z_1,\dots,z_N|\gamma_1,\dots,\gamma_M|x_1,\dots,x_N)$
are the partition functions which are generalizations
of the ones introduced and studied by Ivanov and Brubaker-Bump-Chinta-Gunnells.
The special case $\alpha_j=0$, $\gamma_j=0$ ($j=0,\dots,M$)
of the type I wavefunctions was treated by Ivanov \cite{Iv},
in which the correspondence with the symplectic Schur functions were proven
(the dual wavefunctions was treated recently in \cite{MoRMP}).
The special case
$\alpha_j=0$, $\gamma_j=0$ ($j=1,\dots,M$)
of the type II wavefunctions were introduced by Brubaker-Bump-Chinta-Gunnells
\cite{BBCG}, and they conjectured the correspondence
with the Bump-Friedberg-Hoffstein Whittaker functions \cite{BFH}.

In this paper, we compute the general case where there are no restrictions on
the parameters $\alpha_j$ and $\gamma_j$ ($j=0,\dots,M$).
In the next two sections, we treat the wavefunctions,
the dual wavefunctions and the domain wall boundary partition functions
of type I.
In sections 6 and 7, we present the results for the type II wavefunctions.

\section{Type I wavefunctions and
generalized symplectic Schur functions}
We derive the exact correspondence between the type I wavefunctions
and the generalization of the symplectic Schur functions
in this section.
First, we compute the simplest case $N=1$.
Next, we perform the Izergin-Korepin analysis,
which is a trick to extract sufficiently many conditions to
uniquely determine the wavefunctions (see \cite{Motegi,MoIK} for simpler examples
without boundaries).
Finally, we introduce the generalization of the symplectic Schur functions
and show that the symmetric functions multiplied by some factors
satisfy all the required properties extracted from
the Izergin-Korepin analysis.

\subsection{One particle case}
The simplest case $N=1$ of the type I wavefunctions
can be calculated with the help of
the following identity.
\begin{lemma} \label{identitylemma}
The following identity holds:
\begin{align}
(z-z^{-1})&\sum_{j=1}^{x_1-1}
\prod_{k=1}^{j-1}\left\{
\alpha_k+(1-\alpha_k \gamma_k)z^{-1}\right\}
(1-\gamma_k z)
\prod_{k=j+1}^{x_1-1}
\left\{\alpha_k+(1-\alpha_k \gamma_k)z\right\}
(1-\gamma_k z^{-1})
\nonumber
\\
=&\prod_{k=1}^{x_1-1}\left\{
\alpha_k+(1-\alpha_k \gamma_k)z\right\}(1-\gamma_k z^{-1})
-\prod_{k=1}^{x_1-1}
\left\{\alpha_k+(1-\alpha_k \gamma_k)z^{-1}\right\}(1-\gamma_k z). 
\label{identity}
\end{align}
\end{lemma}
\begin{proof}
This can be proved by induction on $x_1$.
\end{proof}

\begin{proposition} \label{simplestproposition}
The type I wavefunction $\Phi^{\rm I}_{M,1}(z|\gamma_1,\dots,\gamma_M|x_1)$
is explicitly expressed as
\begin{align}
\Phi^{\rm I}_{M,1}&(z|\gamma_1,\dots,\gamma_M|x_1) \nonumber \\
=&\frac{1+tz^2}{z^2-1}
\sum_{\tau=\pm 1} \tau
\prod_{j=0}^{x_1-1}\left\{\alpha_j+(1-\alpha_j \gamma_j)z^\tau\right\}
\prod_{j=x_1+1}^M (1-\gamma_j z^\tau)
\prod_{j=1}^M (1-\gamma_j z^{-\tau}). \label{simplestmatrixelement}
\end{align}.
\end{proposition}

\begin{proof}
We can explicitly calculate $\Phi^{\rm I}_{M,1}(z|\gamma_1,\dots,\gamma_M|x_1)$
by decomposing as
\begin{align}
\Phi^{\rm I}_{M,1}&(z|\gamma_1,\dots,\gamma_M|x_1) \nonumber \\
=&\left\{(1-\alpha_0 \gamma_0)tz-\alpha_0\right\}
\langle x_1 |\widetilde{B}(z,\{ \alpha \},\{ \gamma \})| 0^M \rangle
\langle 0^M |A(z,\{ \alpha \},\{ \gamma \})| 0^M \rangle \nonumber \\
&+\left\{(1-\alpha_0 \gamma_0)z^{-1}+\alpha_0\right\}
\langle x_1 |\widetilde{A}(z,\{ \alpha \},\{ \gamma \})| x_1 \rangle
\langle x_1 |B(z,\{ \alpha \},\{ \gamma \})| 0^M \rangle \nonumber \\
&+\left\{(1-\alpha_0 \gamma_0)z^{-1}+\alpha_0\right\}
\sum_{j=1}^{x_1-1}
\langle x_1 |\widetilde{A}(z,\{ \alpha \},\{ \gamma \})| j \rangle
\langle j |B(z,\{ \alpha \},\{ \gamma \})| 0^M \rangle,
\end{align}
and computing the matrix elements of the monodromy matrices
$\langle x_1| \widetilde{B}(z,\{ \alpha \},\{ \gamma \})| 0^M \rangle$,
$\langle 0^M |A(z,\{ \alpha \},\{ \gamma \})| 0^M \rangle$,
$\langle x_1 |\widetilde{A}(z,\{ \alpha \},\{ \gamma \})| x_1 \rangle$,
$\langle x_1 |B(z,\{ \alpha \},\{ \gamma \})| 0^M \rangle$,
$\langle x_1 |\widetilde{A}(z,\{ \alpha \},\{ \gamma \})| j \rangle$,
$\langle j |B(z,\{ \alpha \},\{ \gamma \})| 0^M \rangle$
to get
\begin{align}
\Phi^{\rm I}_{M,1}&(z|\gamma_1,\dots,\gamma_M|x_1)=
\prod_{j=x_1+1}^M (1-\gamma_j z)(1-\gamma_j z^{-1}) \nonumber \\
&\times \biggl[ \left\{(1-\alpha_0 \gamma_0)tz-\alpha_0\right\}
\prod_{k=1}^{x_1-1}\left\{(1-\alpha_k \gamma_k)z+\alpha_k\right\} 
\prod_{k=1}^{x_1}(1-\gamma_k z^{-1}) \nonumber \\
&\qquad +\left\{(1-\alpha_0 \gamma_0)z^{-1}+\alpha_0\right\}(t\gamma_{x_1}z+1)
\prod_{k=1}^{x_1-1}\left\{(1-\alpha_k \gamma_k)z^{-1}+\alpha_k\right\}(1-\gamma_k z)
\nonumber \\
&\qquad +(t+1)z\left\{(1-\alpha_0 \gamma_0)z^{-1}+\alpha_0\right\}\sum_{j=1}^{x_1-1}
\prod_{k=1}^{j-1}\left\{(1-\alpha_k \gamma_k)z^{-1}+\alpha_k\right\}(1-\gamma_k z)
\nonumber \\
&\qquad \times \prod_{k=j+1}^{x_1-1} \left\{(1-\alpha_k \gamma_k)z+\alpha_k\right\}
\prod_{k=j+1}^{x_1} (1-\gamma_k z^{-1})\biggr]. \label{intoneparticle}
\end{align}
Using the identity \eqref{identity}
in Lemma \ref{identitylemma}, one can show that the right hand side of
\eqref{intoneparticle} can be expressed as \eqref{simplestmatrixelement}.
\end{proof}
We can similarly calculate the type I dual wavefunction $\overline{\Phi}^{\rm I}_{M,1}(z|\gamma_1,\dots,\gamma_M|\overline{x_1})$.
We state the result below.
\begin{proposition}
The type I dual wavefunction $\overline{\Phi}^{\rm I}_{M,1}(z|\gamma_1,\dots,\gamma_M|\overline{x_1})$
is explicitly expressed as
\begin{align}
&\overline{\Phi}^{\rm I}_{M,1}(z|\gamma_1,\dots,\gamma_M|\overline{x_1})
=\frac{t^{M}(1+tz^2)}{t^2z^2-1}\nn \\
&\quad \times \sum_{\tau=\pm 1} \tau
\prod_{j=0}^{\overline{x_1}-1}\left\{
-\alpha_j+(1-\alpha_j \gamma_j)(tz)^\tau\right\}
\prod_{j=\overline{x_1}+1}^M (1+\gamma_j (tz)^\tau)
\prod_{j=1}^M \left\{1+\gamma_j (tz)^{-\tau}\right\}. \label{dualsimplestmatrixelement}
\end{align}
\end{proposition}

\subsection{Izergin-Korepin analysis}
In this subsection, we use the Izergin-Korepin technique \cite{Ko,Iz,Motegi}
and extract the properties of the type I wavefunctions
which uniquely define them.

\begin{proposition} 
\label{ordinarypropertiesfordomainwallboundarypartitionfunction}
The type I wavefunctions
$\Phi^{\rm I}_{M,N}(z_1,\dots,z_N|\gamma_1,\dots,\gamma_M|x_1,\dots,x_N)$
satisfies the following properties. \\
\\
 (1) When $x_N=M$, the type I wavefunctions
$\Phi^{\rm I}_{M,N}(z_1,\dots,z_N|\gamma_1,\dots,\gamma_M|x_1,\dots,x_N)$
is a polynomial of degree $2N-1$ in $\gamma_M$.
\\
 (2) The following form
\begin{align}
\frac{\Phi^{\rm I}_{M,N}(z_1,\dots,z_N|\gamma_1,\dots,\gamma_M|x_1,\dots,x_N)}
{
\prod_{j=1}^N z_j^{j-1-N} (1+tz_j^2) \prod_{1 \le j < k \le N}
(1+tz_j z_k) (1+tz_j z_k^{-1})
},
\label{symmetrywavefunction}
\end{align}
is symmetric with respect to $z_1,\dots,z_N$,
and is invariant under the exchange $z_i \longleftrightarrow z_i^{-1}$
for $i=1,\dots,N$.
\\
(3) The following recursive relations between the
type I wavefunctions hold if $x_N=M$
(Figure \ref{izerginkorepinone}):
\begin{align}
\Phi^{\rm I}_{M,N}&(z_1,\dots,z_N|\gamma_1,\dots,\gamma_M|x_1,\dots,x_N)
|_{\gamma_M=z_N}
\nonumber \\
&=\prod_{j=1}^{N}(t z_N z_j+1)
\prod_{j=1}^{N-1} (t+z_N z_j^{-1})
\prod_{j=0}^{M-1} \left\{(1-\alpha_j \gamma_j)z_N^{-1}+\alpha_j\right\}
\prod_{j=1}^{M-1} (1-\gamma_j z_N)
\nonumber \\
&\quad \times 
\Phi^{\rm I}_{M-1,N-1}(z_1,\dots,z_{N-1}|\gamma_1,\dots,\gamma_{M-1}|x_1,\dots,x_{N-1}). 
\label{ordinaryrecursionwavefunction}
\end{align}

If $x_N \neq M$, the following factorizations hold for the type I
wavefunctions
(Figure \ref{izerginkorepintwo}):
\begin{align}
\Phi_{M,N}^{\rm I}&(z_1,\dots,z_N|\gamma_1,\dots,\gamma_M|x_1,\dots,x_N)
 \nonumber \\
&=\prod_{j=1}^N (1-\gamma_M z_j)(1-\gamma_M z_j^{-1})
\Phi^{\rm I}_{M-1,N}(z_1,\dots,z_N|\gamma_1,\dots,\gamma_{M-1}|x_1,\dots,x_N).
\label{ordinaryrecursionwavefunction2}
\end{align}
\\
(4) The following holds for the case $N=1$, $x_1=M$:
\begin{align}
\Phi^{\rm I}_{M,1}&(z|\gamma_1,\dots,\gamma_M|M) \nonumber \\
&=\frac{1+tz^2}{z^2-1}
\sum_{\tau=\pm 1} \tau
\prod_{j=0}^{M-1}\left\{\alpha_j+(1-\alpha_j \gamma_j)z^\tau\right\}
\prod_{j=1}^M (1-\gamma_j z^{-\tau})
.
\label{ordinaryinitialrecursion}
\end{align}
\end{proposition}

\begin{figure}[ht]
\centering
\includegraphics[width=0.6\textwidth]{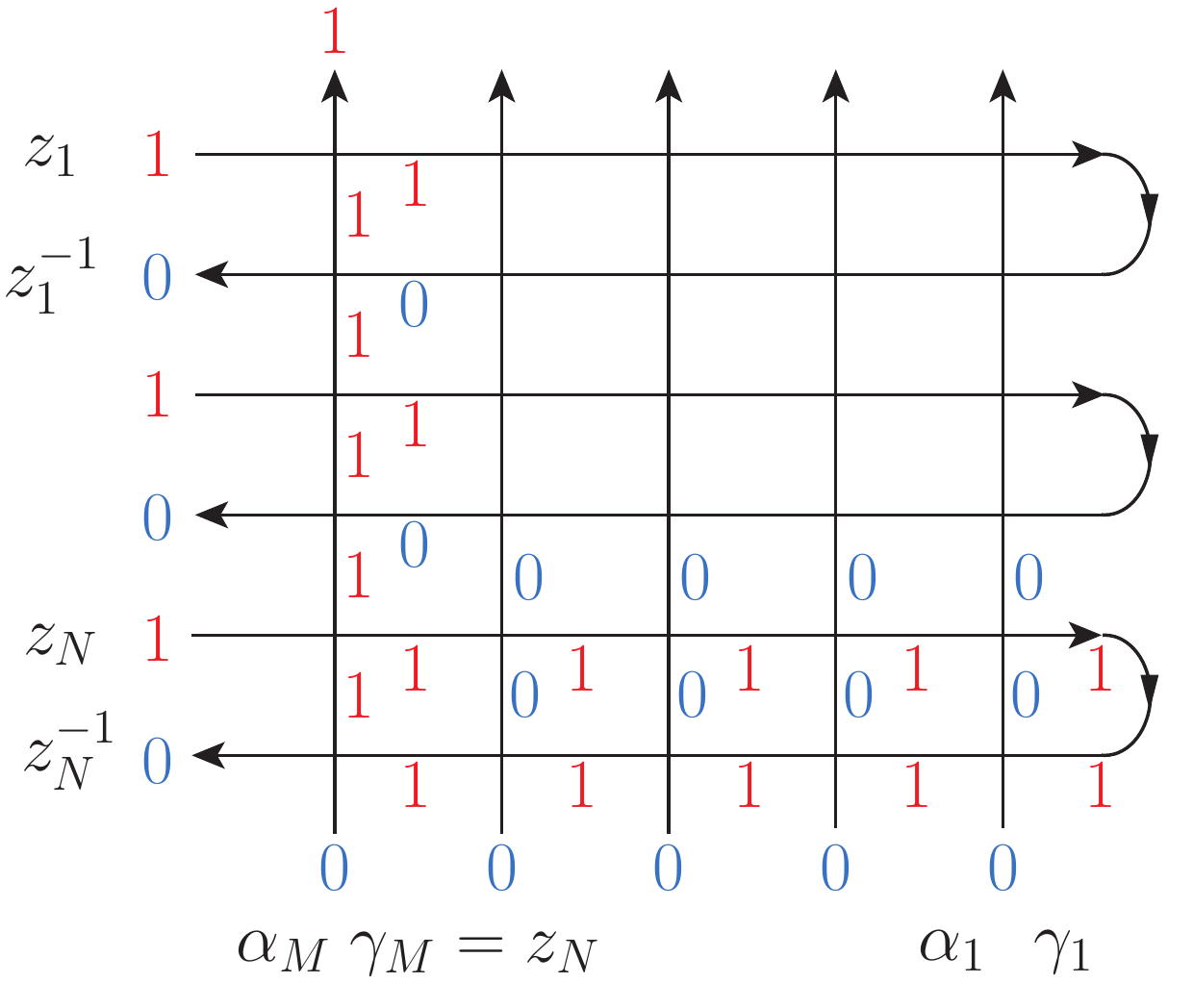}
\caption{A graphical description of the relation
\eqref{ordinaryrecursionwavefunction}.
one can see that if one sets $\gamma_M$ to $\gamma_M=z_N$,
all the $L$-operators at the
leftmost column and the bottom double-row
get frozen.
}
\label{izerginkorepinone}

\end{figure}

\begin{proof}
Property (1) can be shown in a standard way
using a graphical representation of the wavefunctions.

Property (2) can be proved in the same way with Ivanov \cite{Iv,Ivthesis},
since the generalized $L$-operators \eqref{generalizedloperator},
\eqref{anothergeneralizedloperator}
and the generalized $K$-matrix \eqref{generalizedkmatrix} satisfy
certain lemmas called the ``caduceus relations", ``fish relations",
and the arguments built out of the lemmas work as well (see Appendix \ref{app1}
for details).
There is another argument to prove the invariance under the exchange
$z_i \longleftrightarrow z_i^{-1}$ for $i=1,\dots,N$
by using the argument applied to the dual wavefunctions in \cite{MoRMP}.

Property (4) is a special case $x_1=M$ of \eqref{simplestmatrixelement} in
Proposition \ref{simplestproposition} which is already proven
in the previous subsection.

The way to prove the two relations in
Property (3) is also standard in the Izergin-Korepin analysis.
We use the power of the graphical
representation of the wavefunctions
(Figures \ref{izerginkorepinone} and \ref{izerginkorepintwo}).
If $x_N=M$, one sees that after the specialization $\gamma_M=z_N$,
the leftmost column and the bottom row gets frozen.
The product of the matrix elements of the $L$-operators
coming from the frozen part
\begin{align}
\prod_{j=1}^{N}(t z_N z_j+1)
\prod_{j=1}^{N-1} (t+z_N z_j^{-1})
\prod_{j=0}^{M-1} \left\{(1-\alpha_j \gamma_j)z_N^{-1}+\alpha_j\right\}
\prod_{j=1}^{M-1} (1-\gamma_j z_N),
\end{align}
which, multiplied by the remaining part
$\Phi^{\rm I}_{M-1,N-1}
(z_1,\dots,z_{N-1}|\gamma_1,\dots,\gamma_{M-1}|x_1,\dots,x_{N-1})$,
gives the specialization of
$\Phi^{\rm I}_{M,N}(z_1,\dots,z_N|\gamma_1,\dots,\gamma_M|x_1,\dots,x_N)$
at $\gamma_M=z_N$, i.e., \eqref{ordinaryrecursionwavefunction} follows.

If $x_N \neq M$,
one sees that the leftmost column is already frozen without
imposing any specialization on $\gamma_M$, which gives the factor
$\prod_{j=1}^N (1-\gamma_M z_j)(1-\gamma_M z_j^{-1})$.
The remaining part is
$\Phi^{\rm I}_{M-1,N}(z_1,\dots,z_N|\gamma_1,\dots,\gamma_{M-1}|x_1,\dots,x_N)$
and one concludes that the wavefunctions
$\Phi^{\rm I}_{M,N}(z_1,\dots,z_N|\gamma_1,\dots,\gamma_{M}|x_1,\dots,x_N)$
is the product of
$\prod_{j=1}^N (1-\gamma_M z_j)(1-\gamma_M z_j^{-1})$
and
$\Phi^{\rm I}_{M-1,N}(z_1,\dots,z_N|\gamma_1,\dots,\gamma_{M-1}|x_1,\dots,x_N)$,
hence one gets \eqref{ordinaryrecursionwavefunction2}.
\end{proof}

\begin{figure}[htt]
\centering
\includegraphics[width=0.6\textwidth]{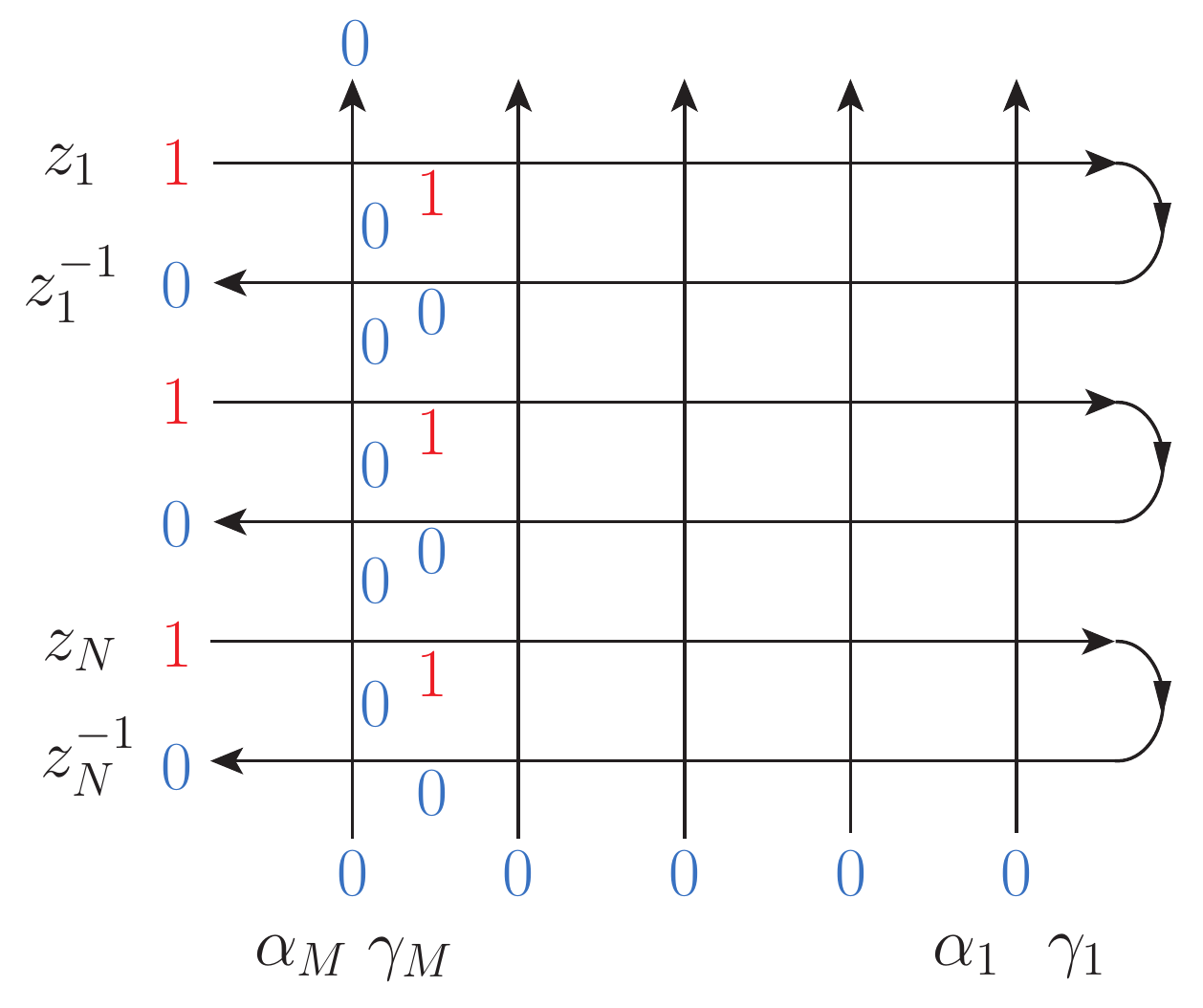}
\caption{A graphical description of the factorization relation
\eqref{ordinaryrecursionwavefunction2}.
One can see that all the $L$-operators at the
leftmost column are frozen.
}
\label{izerginkorepintwo}

\end{figure}

We also list the property
for the type I dual wavefunctions.

\begin{proposition} \label{dualproposition}
The type I dual wavefunctions
$\overline{\Phi}^{\rm I}_{M,N}(z_1,\dots,z_N|\gamma_1,\dots,\gamma_M|\overline{x_1},\dots,\overline{x_N})$
satisfies the following properties. \\
\\
 (1) When $\overline{x_N}=M$, the type I dual wavefunctions
$\overline{\Phi}^{\rm I}_{M,N}
(z_1,\dots,z_N|\gamma_1,\dots,\gamma_M|\overline{x_1},\dots,\overline{x_N})$
is a polynomial of degree $2N-1$ in $\gamma_M$.
\\
 (2) The following form
\begin{align}
\frac{\overline{\Phi}^{\rm I}_{M,N}(z_1,\dots,z_N|\gamma_1,\dots,\gamma_M|x_1,\dots,x_N)}
{
\prod_{j=1}^N z_j^{j-1-N} (1+tz_j^2) \prod_{1 \le j < k \le N}
(1+tz_j z_k) (1+tz_j z_k^{-1})
},
\label{dualsymmetrywavefunction}
\end{align}
is symmetric with respect to $z_1,\dots,z_N$,
and is invariant under the exchange $z_i \longleftrightarrow z_i^{-1}$
for $i=1,\dots,N$.
\\
(3) The following recursive relations between the
type I dual wavefunctions hold if $\overline{x_N}=M$:
\begin{align}
\overline{\Phi}^{\rm I}_{M,N}&
(z_N,\dots,z_1|\gamma_1,\dots,\gamma_M|\overline{x_1},\dots,\overline{x_N})
|_{\gamma_M=-t^{-1}z_N^{-1}}
\nonumber \\
&=\prod_{j=1}^{N} \Bigg(1+ \frac{1}{t z_N z_j} \Bigg)
\prod_{j=1}^{N-1} \Bigg(1+\frac{z_j}{tz_N} \Bigg)
\prod_{j=0}^{M-1} \left\{t(1-\alpha_j \gamma_j)z_N-\alpha_j\right\}
\prod_{j=1}^{M-1} (t+\gamma_j z_N^{-1})
\nonumber \\
&\quad
\times \overline{\Phi}^{\rm I}_{M-1,N-1}
(z_{N-1},\dots,z_{1}|\gamma_1,\dots,\gamma_{M-1}|\overline{x_1},\dots,\overline{x_{N-1}}). \label{remark}
\end{align}

If $\overline{x_N} \neq M$, the following factorizations hold for the
type I dual wavefunctions:
\begin{align}
\overline{\Phi}_{M,N}^{\rm I}
&(z_1,\dots,z_N|\gamma_1,\dots,\gamma_M|\overline{x_1},\dots,\overline{x_N})
 \nonumber \\
&=\prod_{j=1}^N (1+t\gamma_M z_j)(t+\gamma_M z_j^{-1})
\overline{\Phi}^{\rm I}_{M-1,N}(z_1,\dots,z_N|\gamma_1,\dots,\gamma_{M-1}|\overline{x_1},\dots,\overline{x_N}).
\end{align}
\\
(4) The following holds for the case $N=1$, $\overline{x_1}=M$:
\begin{align}
\overline{\Phi}^{\rm I}_{M,1}&(z|\gamma_1,\dots,\gamma_M|M)
\nonumber \\
&=\frac{t^{M}(1+tz^2)}{t^2z^2-1}
\sum_{\tau=\pm 1} \tau
\prod_{j=0}^{M-1}\left\{-\alpha_j+(1-\alpha_j \gamma_j)(tz)^\tau\right\}
\prod_{j=1}^M \left\{1+\gamma_j (tz)^{-\tau}\right\}.
\end{align}
\end{proposition}
We can prove Proposition \ref{dualproposition}
in a similar way with Proposition
\ref{ordinarypropertiesfordomainwallboundarypartitionfunction}.
Note that the ordering of the spectral parameters
in the dual wavefunctions in
\eqref{remark} are $z_N,\dots,z_1$
(the ordering is
$z_1,\dots,z_N$ in the wavefunctions in \eqref{ordinaryrecursionwavefunction}).

\subsection{Generalized symplectic Schur functions}
In this subsection,
we first introduce the following symmetric functions
which generalize the symplectic Schur functions.

\begin{definition}
We define the generalized symplectic Schur functions to be the
following determinant:
\begin{align}
sp_\lambda(\{ z \}_N|\{ \overline{\alpha} \}|\{ \overline{\gamma} \})
=\frac{G_{\lambda+\delta}(\{ z \}_N|\{ \overline{\alpha} \}|\{ 
\overline{\gamma} \})}
{\mathrm{det}_N(z_k^{N-j+1}-z_k^{-N+j-1})}. \label{symplecticschur}
\end{align}
Here, $\{ z \}_N=\{z_1,\dots,z_N \}$ is a set of symmetric variables,
$\{ \overline{\alpha} \}$ and $\{ \overline{\gamma} \}$ are
sets of variables
$\{ \overline{\alpha} \}=\{\alpha_0,\alpha_1,\dots,\alpha_M \}$
and
$\{ \overline{\gamma} \}=\{\gamma_0,\gamma_1,\dots,\gamma_M \}$,
$\lambda$ denotes a Young diagram
$\lambda=(\lambda_1,\lambda_2,\dots,\lambda_N)$
with weakly decreasing non-negative integers
$\lambda_1 \ge \lambda_2 \ge \cdots \ge \lambda_N \ge 0$,
and $\delta=(N-1,N-2,\dots,0)$.
$G_{\mu}(\{ z \}_N|\{ \overline{\alpha} \}|\{ \overline{\gamma} \})$
is an $N \times N$ determinant
\begin{align}
G_{\mu}(\{ z \}_N|\{ \overline{\alpha} \}|\{ \overline{\gamma} \})
=\mathrm{det}_N
(
g_{\mu_j}(z_k|\{ \overline{\alpha} \}|\{ \overline{\gamma} \})
-
g_{\mu_j}(z_k^{-1}|\{ \overline{\alpha} \}|\{ \overline{\gamma} \})
),
\end{align}
where
\begin{align}
g_\mu(z|\{ \overline{\alpha} \}|\{ \overline{\gamma} \})
=\prod_{j=0}^\mu\left\{\alpha_j+(1-\alpha_j \gamma_j)z\right\}
\prod_{j=\mu+2}^M(1-\gamma_j z)
\prod_{j=1}^M(1-\gamma_jz^{-1}).
\label{det}
\end{align}
\end{definition}
One can see from the definition \eqref{symplecticschur} that
the generalized symplectic Schur functions are
symmetric functions with respect to the variables $z_1,\dots,z_N$,
and are invariant under the exchange $z_i \longleftrightarrow z_i^{-1}$
for $i=1,\dots,N$.
If one sets $\alpha_j=\gamma_j=0$ ($j=0,\dots,M$),
the generalized symplectic Schur functions
reduce to the ordinary symplectic Schur functions.

Now we state the correspondence between the type I wavefunctions
and the generalized symplectic Schur functions.

\begin{theorem} \label{maintheoremstatement}
The type I wavefunctions
$\Phi^{\rm I}_{M,N}(z_1,\dots,z_N|\gamma_1,\dots,\gamma_M|x_1,\dots,x_N)$
are explicitly expressed
using the generalized symplectic Schur functions as
\begin{align}
\Phi^{\rm I}_{M,N}&(z_1,\dots,z_N|\gamma_1,\dots,\gamma_M|x_1,\dots,x_N)
\nonumber \\
&=\prod_{j=1}^N z_j^{j-1-N}(1+tz_j^2)
\prod_{1 \le j<k \le N}(1+tz_j z_k)(1+tz_j z_k^{-1})
sp_\lambda(\{ z \}_N|\{ \overline{\alpha} \}|\{ \overline{\gamma} \}),
\label{maintheorem}
\end{align}
under the relation $\lambda_j=x_{N-j+1}-N+j-1$ ($j=1,\dots,N$).

The type I dual wavefunctions
$\overline{\Phi}^{\rm I}_{M,N}(z_1,\dots,z_N|\gamma_1,\dots,\gamma_M|\overline{x_1},\dots,\overline{x_N})$
are explicitly expressed
using the generalized symplectic Schur functions as
\begin{align}
\overline{\Phi}^{\rm I}_{M,N}&
(z_1,\dots,z_N|\gamma_1,\dots,\gamma_M|\overline{x_1},\dots,\overline{x_N})
\nonumber \\
=&t^{N(M-N)}
\prod_{j=1}^N z_j^{j-1-N}(1+tz_j^2)
\prod_{1 \le j<k \le N}(1+tz_j z_k)(1+tz_j z_k^{-1})
sp_{\overline{\lambda}} ( \{ tz \}_N |\{-\overline{\alpha} \}|\{-\overline{\gamma} \}),
\label{dualmaintheorem}
\end{align}
under the relation
$\overline{\lambda_j}=\overline{x_{N-j+1}}-N+j-1$ ($j=1,\dots,N$)
and the symmetric variables are
$\displaystyle \{ tz \}_N=
\{ tz_1,\dots,tz_N \}$.
Moreover, the signs of the parameters of the generalized symplectic
Schur functions in the right hand side of
\eqref{dualmaintheorem} are now
inverted simultaneously: $\{-\overline{\alpha} \}
=\{-\alpha_0,-\alpha_1,\dots,-\alpha_M \}$
and $\{-\overline{\gamma} \}
=\{-\gamma_0,-\gamma_1,\dots,-\gamma_M \}$.
\end{theorem}
The correspondences
\eqref{maintheorem} and \eqref{dualmaintheorem}
are generalizations of the one by Ivanov in \cite{Iv,Ivthesis} and
one of the authors in \cite{MoRMP}.
One way to prove these correspondences
is to adopt the argument by Brubaker-Bump-Friedberg \cite{BBF}
which they viewed the wavefunctions of the free-fermionic
six-vertex model without reflecting boundary
as polynomials in $t$ and studied its properties
to find the exact correspondence with the Schur functions.
The argument works for this case as well.
However, we use the Izergin-Korepin technique,
since this argument works for the type II wavefunctions as well,
in which case viewing as a function of $t$ can extract several properties
but does not lead to the final form,
as mentioned in Brubaker-Bump-Chinta-Gunnells \cite{BBCG}.
Note that the Izergin-Korepin technique views the wavefunctions
as functions of the parameter $\gamma_M$
associated with the quantum space $V_M$.

\begin{proof}
Let us show the correspondence \eqref{maintheorem}.
\eqref{dualmaintheorem} can be proved in the same way
(See Appendix for the detailed calculations).
We have to show that the right hand side of
\eqref{maintheorem} satisfies all the Properties
in Proposition \ref{ordinarypropertiesfordomainwallboundarypartitionfunction}.

It is first easy to see that Property (2) holds,
since the right hand side of \eqref{maintheorem}
divided by the factor $\prod_{j=1}^N z_j^{j-1-N}(1+tz_j^2)
\prod_{1 \le j<k \le N}(1+tz_j z_k)(1+tz_j z_k^{-1})$
is nothing but the generalized symplectic Schur functions
$sp_\lambda(\{ z \}_N|\{ \overline{\alpha} \}|\{ \overline{\gamma} \})$
which is symmetric with respect to the variables $z_1,\dots,z_N$
and is invariant under the exchange $z_i \longleftrightarrow z_i^{-1}$
for $i=1,\dots,N$.

Next, using the following factorization
\begin{align}
\mathrm{det}_N(z_k^{N-j+1}-z_k^{-N+j-1})
=(-1)^N \prod_{j=1}^N z_j^{j-1-N}(1-z_j^2)
\prod_{1 \le j<k \le N}(1-z_j z_k)(1-z_j z_k^{-1}),
\label{factorization}
\end{align}
and the definition of the determinant and the correspondence
between the position of particles $\{ x \}$ and
the Young diagrams $\{ \lambda \}$,
we rewrite the right hand side of \eqref{maintheorem} as
\begin{align}
F^{\rm I}_{M,N}&(z_1,\dots,z_N|\gamma_1,\dots,\gamma_M|x_1,\dots,x_N)
\nonumber \\
&:=
\frac{
\prod_{j=1}^N z_j^{j-1-N}(1+tz_j^2)
\prod_{1 \le j<k \le N}(1+tz_j z_k)(1+tz_j z_k^{-1})
}{
(-1)^N \prod_{j=1}^N z_j^{j-1-N}(1-z_j^2)
\prod_{1 \le j<k \le N}(1-z_j z_k)(-1+z_j z_k^{-1})
} \nonumber \\
&\quad \times\sum_{\sigma \in S_N}
\sum_{\tau_1=\pm 1,\dots,\tau_N=\pm 1} (-1)^\sigma (-1)^{|\tau|}
\prod_{j=1}^N \prod_{k=0}^{x_j-1}
\left\{\alpha_k+(1-\alpha_k \gamma_k)z_{\sigma(j)}^{\tau_{\sigma(j)}}
\right\} \nonumber \\
&\quad \times
\prod_{j=1}^N \prod_{k=x_j+1}^{M}
\left(1-\gamma_k z_{\sigma(j)}^{\tau_{\sigma(j)}}\right)
\prod_{j=1}^N \prod_{k=1}^{M}
\left(1-\gamma_k z_{\sigma(j)}^{-\tau_{\sigma(j)}}\right),
 \label{clearrepresentaitonforproof}
\end{align}
where $|\tau|$ denotes the number of $\tau_j$'s
satisfying $\tau_j=-1$. Note that the factor $(-1)^{N(N-1)/2}$ appears
when we write down the determinant in \eqref{det} in terms of the 
position of the particle $\{x\}$ instead of the Young diagrams $\{\lambda\}$.

It can be easily seen by rewriting the right hand side of \eqref{maintheorem}
as \eqref{clearrepresentaitonforproof} that it
satisfies Property (4).

Property (1) is also easy to see from the expression
\eqref{clearrepresentaitonforproof}.
If $x_N=M$, one can see that the factors
\begin{align}
\prod_{j=1}^N \prod_{k=x_j+1}^{M}
\left(1-\gamma_k z_{\sigma(j)}^{\tau_{\sigma(j)}}\right)
\prod_{j=1}^N \prod_{k=1}^{M}
\left(1-\gamma_k z_{\sigma(j)}^{-\tau_{\sigma(j)}}\right),
\end{align}
which contain $\gamma_M$ in each summands become
\begin{align}
\prod_{j=1}^{N-1} \prod_{k=x_j+1}^{M}
\left(1-\gamma_k z_{\sigma(j)}^{\tau_{\sigma(j)}}\right)
\prod_{j=1}^N \prod_{k=1}^{M}
\left(1-\gamma_k z_{\sigma(j)}^{-\tau_{\sigma(j)}}\right),
\end{align}
from which one concludes that the degree with respect to $\gamma_M$
is $2N-1$.

Let us prove that the functions
$F^{\rm I}_{M,N}(z_1,\dots,z_N|\gamma_1,\dots,\gamma_M|x_1,\dots,x_N)$
satisfy Property (3).
We first treat the case $x_N=M$.
Specializing $\gamma_M$ to $\gamma_M=z_N$,
we find that only the summands satisfying $\sigma(N)=N$, $\tau_N=-1$ in 
\eqref{clearrepresentaitonforproof} survive.
Making use of this observation, we rewrite
$F^{\rm I}_{M,N}(z_1,\dots,z_N|\gamma_1,\dots,\gamma_M|x_1,\dots,x_N)
|_{\gamma_M=z_N}$ as
\begin{align}
F^{\rm I}_{M,N}&(z_1,\dots,z_N|\gamma_1,\dots,\gamma_M|x_1,\dots,x_N)
|_{\gamma_M=z_N}
\nonumber \\
&=-\frac{1+tz_N^2}{1-z_N^2}
\frac{\prod_{j=1}^{N-1}(1+tz_j z_N)(1+tz_j z_N^{-1})}
{\prod_{j=1}^{N-1}(1-z_j z_N) (-1+z_j z_N^{-1})} \nonumber \\
&\quad \times
\frac{
\prod_{j=1}^{N-1} (1+tz_j^2)
\prod_{1 \le j<k \le N-1}(1+tz_j z_k)(1+tz_j z_k^{-1})
}{
(-1)^{N-1} \prod_{j=1}^{N-1} (1-z_j^2)
\prod_{1 \le j<k \le N-1}(1-z_j z_k)(-1+z_j z_k^{-1})
}
\nonumber \\
&\quad \times \sum_{\sigma \in S_{N-1}}
\sum_{\tau_1=\pm 1,\dots,\tau_{N-1}=\pm 1}
(-1)(-1)^\sigma (-1)^{|\tau|} \nonumber \\
&\quad \times \prod_{j=1}^{N-1} \prod_{k=0}^{x_j-1}
\left\{\alpha_k+(1-\alpha_k \gamma_k)z_{\sigma(j)}^{\tau_{\sigma(j)}}
\right\}
\prod_{k=0}^{M-1}
\left\{\alpha_k+(1-\alpha_k \gamma_k)z_N^{-1}
\right\}
\nonumber \\
&\quad \times
\prod_{j=1}^{N-1} \prod_{k=x_j+1}^{M-1}
\left(1-\gamma_k z_{\sigma(j)}^{\tau_{\sigma(j)}}\right)
\prod_{j=1}^{N-1}
\Bigg(1-z_N z_{\sigma(j)}^{\tau_{\sigma(j)}}\Bigg)
\nonumber \\
&\quad \times
(1-z_N^2)
\prod_{k=1}^{M-1}
(1-\gamma_k z_N )
\prod_{j=1}^{N-1}
\left(1-z_N z_{\sigma(j)}^{-\tau_{\sigma(j)}}\right)
\prod_{j=1}^{N-1} \prod_{k=1}^{M-1}
\left(1-\gamma_k z_{\sigma(j)}^{-\tau_{\sigma(j)}}\right). \label{intersymplectic}
\end{align}
Using the obvious identity
\begin{align}
\prod_{j=1}^{N-1}
\left(1-z_N z_{\sigma(j)}^{\tau_{\sigma(j)}}\right)
\left(1-z_N z_{\sigma(j)}^{-\tau_{\sigma(j)}}\right)
=
\prod_{j=1}^{N-1}
(1-z_N z_j)
(1-z_N z_{j}^{-1}),
\end{align}
and after some calculations, one sees that
\eqref{intersymplectic}
can be simplified as
\begin{align}
F^{\rm I}_{M,N}&(z_1,\dots,z_N|\gamma_1,\dots,\gamma_M|x_1,\dots,x_N)
|_{\gamma_M=z_N}
\nonumber \\
&=\prod_{j=1}^{N}(t z_N z_j+1)
\prod_{j=1}^{N-1} (t+z_N z_j^{-1})
\prod_{j=0}^{M-1} ((1-\alpha_j \gamma_j)z_N^{-1}+\alpha_j)
\prod_{j=1}^{M-1} (1-\gamma_j z_N)
\nonumber \\
&\quad \times 
\frac{
\prod_{j=1}^{N-1} z_j^{j-1-(N-1)}(1+tz_j^2)
\prod_{1 \le j<k \le N-1}(1+tz_j z_k)(1+tz_j z_k^{-1})
}{
(-1)^{N-1} \prod_{j=1}^{N-1} z_j^{j-1-(N-1)}(1-z_j^2)
\prod_{1 \le j<k \le N-1}(1-z_j z_k)(-1+z_j z_k^{-1})
} \nonumber \\
&\quad \times\sum_{\sigma \in S_{N-1}}
\sum_{\tau_1=\pm 1,\dots,\tau_{N-1}=\pm 1} (-1)^\sigma (-1)^{|\tau|}
\prod_{j=1}^{N-1} \prod_{k=0}^{x_j-1}
\left\{\alpha_k+(1-\alpha_k \gamma_k)z_{\sigma(j)}^{\tau_{\sigma(j)}}
\right\} \nonumber \\
&\quad \times
\prod_{j=1}^{N-1} \prod_{k=x_j+1}^{M-1}
\left(1-\gamma_k z_{\sigma(j)}^{\tau_{\sigma(j)}}\right)
\prod_{j=1}^{N-1} \prod_{k=1}^{M-1}
\left(1-\gamma_k z_{\sigma(j)}^{-\tau_{\sigma(j)}}\right) \nonumber \\
&=\prod_{j=1}^{N}(t z_N z_j+1)
\prod_{j=1}^{N-1} (t+z_N z_j^{-1})
\prod_{j=0}^{M-1} \left\{(1-\alpha_j \gamma_j)z_N^{-1}+\alpha_j\right\}
\prod_{j=1}^{M-1} (1-\gamma_j z_N)
\nonumber \\
&\quad \times F^{\rm I}_{M-1,N-1}(z_1,\dots,z_{N-1}|\gamma_1,\dots,\gamma_{M-1}|x_1,\dots,x_{N-1}),
\end{align}
hence it is shown that $F^{\rm I}_{M,N}(z_1,\dots,z_N|\gamma_1,\dots,\gamma_M|x_1,\dots,x_N)$
satisfies Property (3) for the case $x_N=M$.

Property (3) for the case $x_N \neq M$ is much easier to prove.
We just rewrite the functions
$F^{\rm I}_{M,N}(z_1,\dots,z_N|\gamma_1,\dots,\gamma_M|x_1,\dots,x_N)$ as
\begin{align}
F^{\rm I}_{M,N}&(z_1,\dots,z_N|\gamma_1,\dots,\gamma_M|x_1,\dots,x_N) \nonumber \\
&=
\frac{
\prod_{j=1}^N z_j^{j-1-N}(1+tz_j^2)
\prod_{1 \le j<k \le N}(1+tz_j z_k)(1+tz_j z_k^{-1})
}{
(-1)^N \prod_{j=1}^N z_j^{j-1-N}(1-z_j^2)
\prod_{1 \le j<k \le N}(1-z_j z_k)(-1+z_j z_k^{-1})
} \nonumber \\
&\quad \times\sum_{\sigma \in S_N}
\sum_{\tau_1=\pm 1,\dots,\tau_N=\pm 1} (-1)^\sigma (-1)^{|\tau|}
\prod_{j=1}^N \prod_{k=0}^{x_j-1}
\left\{\alpha_k+(1-\alpha_k \gamma_k)z_{\sigma(j)}^{\tau_{\sigma(j)}}
\right\} \nonumber \\
&\quad \times
\prod_{j=1}^N \prod_{k=x_j+1}^{M-1}
\left\{1-\gamma_k z_{\sigma(j)}^{\tau_{\sigma(j)}}\right\}
\prod_{j=1}^N \prod_{k=1}^{M-1}
\left\{1-\gamma_k z_{\sigma(j)}^{-\tau_{\sigma(j)}}\right\}
\nonumber \\
&\quad \times
\prod_{j=1}^N
\left(1-\gamma_M z_{\sigma(j)}^{\tau_{\sigma(j)}}\right)
\prod_{j=1}^N 
\left(1-\gamma_M z_{\sigma(j)}^{-\tau_{\sigma(j)}}\right),
\end{align}
and use the identity
\begin{align}
\prod_{j=1}^N
\left(1-\gamma_M z_{\sigma(j)}^{\tau_{\sigma(j)}}\right)
\prod_{j=1}^N 
\left(1-\gamma_M z_{\sigma(j)}^{-\tau_{\sigma(j)}}\right)
=\prod_{j=1}^N
(1-\gamma_M z_j)
\prod_{j=1}^N 
(1-\gamma_M z_j^{-1}), \nonumber
\end{align}
to get
\begin{align}
F^{\rm I}_{M,N}&(z_1,\dots,z_N|\gamma_1,\dots,\gamma_M|x_1,\dots,x_N) \nonumber \\
&=
\frac{
\prod_{j=1}^N z_j^{j-1-N}(1+tz_j^2)
\prod_{1 \le j<k \le N}(1+tz_j z_k)(1+tz_j z_k^{-1})
}{
(-1)^N \prod_{j=1}^N z_j^{j-1-N}(1-z_j^2)
\prod_{1 \le j<k \le N}(1-z_j z_k)(-1+z_j z_k^{-1})
} \nonumber \\
&\quad \times\sum_{\sigma \in S_N}
\sum_{\tau_1=\pm 1,\dots,\tau_N=\pm 1} (-1)^\sigma (-1)^{|\tau|}
\prod_{j=1}^N \prod_{k=0}^{x_j-1}
\left\{\alpha_k+(1-\alpha_k \gamma_k)z_{\sigma(j)}^{\tau_{\sigma(j)}}
\right\} \nonumber \\
&\quad \times
\prod_{j=1}^N \prod_{k=x_j+1}^{M-1}
\left(1-\gamma_k z_{\sigma(j)}^{\tau_{\sigma(j)}}\right)
\prod_{j=1}^N \prod_{k=1}^{M-1}
\left(1-\gamma_k z_{\sigma(j)}^{-\tau_{\sigma(j)}}\right)
%\nonumber \\
%&\quad 
%\times
\prod_{j=1}^N
(1-\gamma_M z_j)
\prod_{j=1}^N 
(1-\gamma_M z_j^{-1}) \nonumber \\
&=\prod_{j=1}^N
(1-\gamma_M z_j)
\prod_{j=1}^N 
(1-\gamma_M z_j^{-1})
F^{\rm I}_{M-1,N}(z_1,\dots,z_N|\gamma_1,\dots,\gamma_{M-1}|x_1,\dots,x_N),
\end{align}
which shows that $F_{M,N}(z_1,\dots,z_N|\gamma_1,\dots,\gamma_M|x_1,\dots,x_N)$
satisfies Property (3) for the case $x_N \neq M$.

\end{proof}

\section{Dual Cauchy formula for generalized symplectic Schur functions}
In this section, as an application of the correspondence between
the type I wavefunctions and the generalized symplectic Schur functions,
we derive the dual Cauchy formula for
the generalized symplectic Schur functions.
We apply the idea due to Bump-McNamara-Nakasuji \cite{BMN},
which they derived the dual Cauchy formula
for factorial Schur functions by
evaluating the domain wall boundary partition functions
in two ways and comparing the two evaluations.
The domain wall boundary partition functions
$Z^{\rm I}_M(z_1,\dots,z_M|\gamma_1,\dots,\gamma_M)$ are special cases
$M=N$, $x_j=j$, $j=1,\dots,M$ of the wavefunctions
\begin{align}
Z^{\rm I}_M(z_1,\dots,z_M|\gamma_1,\dots,\gamma_M)
:=\Phi^{\rm I}_{M,M}(z_1,\dots,z_M|\gamma_1,\dots,\gamma_M|1,\dots,M).
\label{specialcase}
\end{align}

\begin{figure}[ht]
\centering
\includegraphics[width=0.4\textwidth]{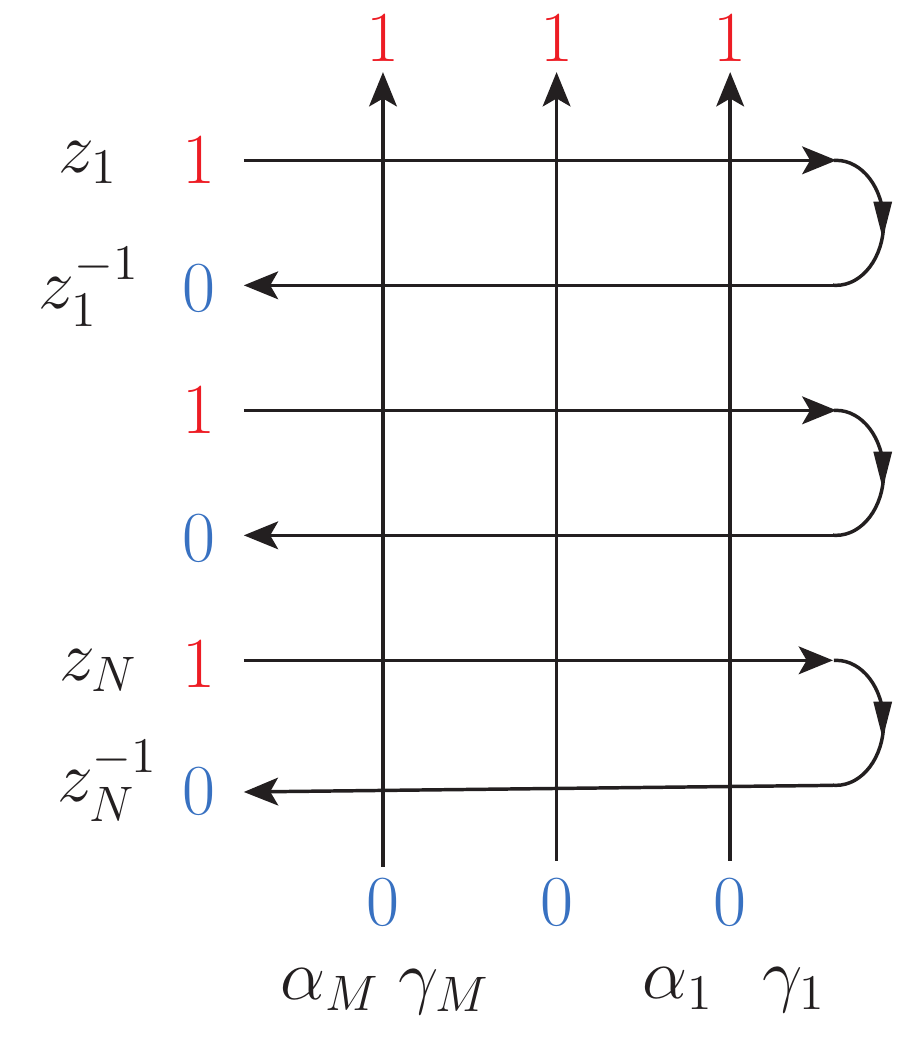}
\caption{The domain wall boundary partition functions
$Z^{\rm I}_M(z_1,\dots,z_M|\gamma_1,\dots,\gamma_M)$
\eqref{determinantinhomogeneheous},
$Z^{\rm II}_M(z_1,\dots,z_M|\gamma_1,\dots,\gamma_M)$
\eqref{typetwodeterminantinhomogeneheous}
under reflecting boundary.}

\end{figure}

First, one can show the following factorization formula
for the type I domain wall boundary partition functions.

\begin{theorem} \label{inhomogeneousexpression}
The type I domain wall boundary partition functions
$Z^{\rm I}_M(z_1,\dots,z_M|\gamma_1,\dots,\gamma_M)$ have
the following factorized form:

\begin{align}
Z^{\rm I}_M&(z_1,\dots,z_M|\gamma_1,\dots,\gamma_M)
=
\prod_{j=1}^M z_j^{j-1-M} (1+tz_j^2)
\prod_{1 \le j<k \le M}(1+tz_j z_k)(1+tz_j z_k^{-1}) \nonumber \\
&\quad \times \prod_{0 \le j<k \le M}
\left\{1+\alpha_j(\gamma_k-\gamma_j)\right\}
\prod_{1 \le j<k \le M} (1-\gamma_j \gamma_k).
\label{determinantinhomogeneheous}
\end{align}

\end{theorem}

\begin{proof}
Since the type I domain wall boundary partition functions
$Z^{\rm I}_M(z_1,\dots,z_M|\gamma_1,\dots,\gamma_M)$
are
special cases of the type I wavefunctions
$\Phi^{\rm I}_{M,M}(z_1,\dots,z_M|\gamma_1,\dots,\gamma_M|1,\dots,M)$
\eqref{specialcase}, what we need is to
prove the following special case of Proposition
\ref{ordinarypropertiesfordomainwallboundarypartitionfunction}.
\begin{proposition}
\label{specializerginkorepin}
The type I domain wall boundary partition functions
$Z^{\rm I}_M(z_1,\dots,z_M|\gamma_1,\dots,\gamma_M)$
satisfies the following properties. \\
\\
 (1) The type I domain wall boundary partition functions
$Z^{\rm I}_M(z_1,\dots,z_M|\gamma_1,\dots,\gamma_M)$
is a polynomial of degree $2M-1$ in $\gamma_M$.
\\
 (2) The following form
\begin{align}
\frac{Z^{\rm I}_{M}(z_1,\dots,z_M|\gamma_1,\dots,\gamma_M)}
{
\prod_{j=1}^M z_j^{j-1-M} (1+tz_j^2) \prod_{1 \le j < k \le M}
(1+tz_j z_k) (1+tz_j z_k^{-1})
},
\end{align}
is symmetric with respect to $z_1,\dots,z_M$,
and is invariant under the exchange $z_i \longleftrightarrow z_i^{-1}$
for $i=1,\dots,M$.
\\
(3) The following recursive relations between the
type I domain wall boundary partition functions hold:
\begin{align}
Z^{\rm I}_{M}&(z_1,\dots,z_M|\gamma_1,\dots,\gamma_M)
|_{\gamma_M=z_M}
\nonumber \\
&=\prod_{j=1}^{M}(t z_M z_j+1)
\prod_{j=1}^{M-1} (t+z_M z_j^{-1})
\prod_{j=0}^{M-1} \left\{(1-\alpha_j \gamma_j)z_M^{-1}+\alpha_j\right\}
\prod_{j=1}^{M-1} (1-\gamma_j z_M)
\nonumber \\
&\quad \times Z^{\rm I}_{M-1}(z_1,\dots,z_{M-1}|\gamma_1,\dots,\gamma_{M-1}).
\end{align}
\\
(4) The following holds for the case $M=1$:
\begin{align}
Z^{\rm I}_1(z|\gamma_1)
=&\frac{1+tz^2}{z^2-1}
\sum_{\tau=\pm 1} \tau
\left\{\alpha_0+(1-\alpha_0\gamma_0)z^{\tau}\right\}(1-\gamma_1 z^{-\tau}).
\end{align}
\end{proposition}
It is easy to see that the right hand side of
\eqref{determinantinhomogeneheous} satisfies
all the properties listed in
Proposition \ref{specializerginkorepin}.

\end{proof}

Now we derive the dual Cauchy formula for
the generalized symplectic Schur functions.

\begin{theorem}
The following dual Cauchy formula holds
for the generalized symplectic Schur functions
with sets variables $\{ x \}_N=\{x_1,\dots,x_N \}$, $\{ y \}_M=\{y_1,\dots,y_M \}$,
$\{ \overline{\alpha} \}=\{\alpha_0,\dots,\alpha_{N+M} \}$,
$\{ \overline{\gamma} \}=\{\gamma_0,\dots,\gamma_{N+M} \}$,
\begin{align}
&\sum_{\lambda \subseteq M^N}sp_\lambda(\{ x \}_N|\{ \overline{\alpha} \}|\{
\overline{\gamma} \})
sp_{\hat{\lambda}}(\{ y \}_M|\{- \overline{\alpha} \}|\{- \overline{\gamma} \}) \nonumber \\
=&\prod_{j=1}^M y_j^{-N}
\prod_{j=1}^N \prod_{k=1}^M (1+x_j y_k)(1+x_j^{-1} y_k)
\prod_{0 \le j<k \le N+M}(1+\alpha_j(\gamma_k-\gamma_j))
\prod_{1 \le j<k \le N+M}(1-\gamma_j \gamma_k),
\label{dualcauchy}
\end{align}
where
$\{ -\overline{\alpha} \}=\{-\alpha_0,\dots,-\alpha_{N+M} \}$,
$\{ -\overline{\gamma} \}=\{-\gamma_0,\dots,-\gamma_{N+M} \}$
and
$\hat{\lambda}
=(\hat{\lambda}_1,\dots,\hat{\lambda}_M)$ is the partition
of the Young diagram $\lambda=(\lambda_1,\dots,\lambda_N)$ given by
\begin{align}
\hat{\lambda}_i=|\{j \ | \ \lambda_j  \le M-i \}|.
\end{align}
\end{theorem}

\begin{proof}
First, we have shown Theorem \ref{inhomogeneousexpression}
which states that the type I domain wall boundary partition functions
$Z^{\rm I}_M(z_1,\dots,z_M|\gamma_1,\dots,\gamma_M)$
have the factorized form
\eqref{determinantinhomogeneheous}.

On the other hand, one can evaluate
the domain wall boundary partition functions
by inserting the completeness relation
\begin{align}
\sum_{\{ x \}}|x_1 \cdots x_N \rangle \langle x_1 \cdots x_N |=\mathrm{Id},
\end{align}
between the double-row $B$-operators and
using the correspondences between
the wavefunctions, the dual wavefunctions
and the generalized symplectic Schur functions
\eqref{maintheorem} and \eqref{dualmaintheorem}
to get
\begin{align}
Z^{\rm I}_M&(z_1,\dots,z_M|\gamma_1,\dots,\gamma_M) \nonumber \\
&=\langle 1^M|
\mathcal{B}^{\rm I}(z_1,\{ \overline{\alpha} \},\{ \overline{\gamma} \})
 \cdots 
\mathcal{B}^{\rm I}(z_M,\{ \overline{\alpha} \},\{ \overline{\gamma} \})
|0^M \rangle \nonumber \\
&=\sum_{\{x \}} \langle 1^M|
\mathcal{B}^{\rm I}(z_1,\{ \overline{\alpha} \},\{ \overline{\gamma} \})
\cdots
\mathcal{B}^{\rm I}(z_{M-N},\{ \overline{\alpha} \},\{ \overline{\gamma} \})
|x_1 \cdots x_N \rangle \nonumber \\
&\quad \times \langle x_1 \cdots x_N|
\mathcal{B}^{\rm I}(z_{M-N+1},\{ \overline{\alpha} \},\{ \overline{\gamma} \})
\cdots
\mathcal{B}^{\rm I}(z_M,\{ \overline{\alpha} \},\{ \overline{\gamma} \})
|0^M \rangle \nonumber \\
&=\sum_{x \sqcup \overline{x}=\{1,\dots,M \}} \langle 1^M|
\mathcal{B}^{\rm I}(z_1,\{ \overline{\alpha} \},\{ \overline{\gamma} \})
\cdots
\mathcal{B}^{\rm I}(z_{M-N},\{ \overline{\alpha} \},\{ \overline{\gamma} \})
|\overline{x_1} \cdots \overline{x_{M-N}} \rangle
\nonumber \\
&\quad \times \langle x_1 \cdots x_N|
\mathcal{B}^{\rm I}(z_{M-N+1},\{ \overline{\alpha} \},\{ \overline{\gamma} \})
\cdots
\mathcal{B}^{\rm I}(z_{M-N},\{ \overline{\alpha} \},\{ \overline{\gamma} \})
|0^M \rangle \nonumber \\
&=\sum_{x \sqcup \overline{x}=\{1,2,\dots,M \}} 
\overline{\Phi}^{\rm I}_{M,M-N}(z_1,\dots,z_{M-N}|\gamma_1,\dots,\gamma_M|\overline{x_1},
\dots,\overline{x_{M-N}}) \nonumber \\
&\quad \times \Phi^{\rm I}_{M,N}
(z_{M-N+1},\dots,z_M|\gamma_1,\dots,\gamma_M|x_1,\dots,x_N)
\nonumber \\
&=\sum_{\lambda \subseteq (M-N)^N}
t^{N(M-N)}
\prod_{j=1}^{M-N} z_j^{j-1-M+N}(1+tz_j^2)
\prod_{1 \le j<k \le M-N}(1+tz_j z_k)(1+tz_j z_k^{-1}) \nonumber \\
&\quad \times
sp_{\overline{\lambda}}(tz_1,\dots,tz_{M-N}|\{-\overline{\alpha} \}|\{-\overline{\gamma} \}) \nonumber \\
&\quad \times
\prod_{j=M-N+1}^M z_j^{j-1-M}(1+tz_j^2)
\prod_{M-N+1 \le j<k \le M}(1+tz_j z_k)(1+tz_j z_k^{-1})
\nonumber \\
&\quad \times
sp_\lambda(z_{M-N+1},\dots,z_M|
\{\overline{\alpha} \}|\{ \overline{\gamma} \}).
\label{comparisontwo}
\end{align}
Note that the sum over all particle configurations $\{x \}$
is translated to the sum over all Young diagrams $\lambda$
satisfying $\lambda \subseteq (M-N)^N$.
Comparing the two ways of evaluations
\eqref{comparisontwo} and
\eqref{determinantinhomogeneheous}
and cancelling common factors, we get
the following identity
\begin{align}
\prod_{j=1}^{M-N} &(tz_j)^{-N}
\prod_{\substack{1 \le j \le M-N \\ M-N+1 \le k \le M}}
(1+tz_j z_k)(1+tz_j z_k^{-1}) \nonumber \\
&\quad \times \prod_{0 \le j<k \le M}(1+\alpha_j(\gamma_k-\gamma_j))
\prod_{1 \le j<k \le M}(1-\gamma_j \gamma_k)
\nonumber \\
&=\sum_{\lambda \subseteq (M-N)^N}
sp_{\overline{\lambda}}(tz_1,\dots,tz_{M-N}|\{-\overline{\alpha} \}|\{-\overline{\gamma} \})
sp_\lambda(z_{M-N+1},\dots,z_M|
\{ \overline{\alpha} \}|\{ \overline{\gamma} \}),
\end{align}
which, after some reparametrization, can be rewritten in
the form \eqref{dualcauchy}.
\end{proof}

The dual Cauchy formula \eqref{dualcauchy} is a generalization
of the one for the ordinary symplectic Schur functions which were proven
by various ways \cite{Morris,King,JM,Ha,Te,BG,HK3}.

\section{Type II wavefunctions and
generalized Bump-Friedberg-Hoffstein Whittaker functions}
In this and the next sections,
we perform the analysis on the type II wavefunctions.
The arguments, statements and the proofs are similar to
those for the type I wavefunctions.

\subsection{One particle case}
The simplest case $N=1$
of type II wavefunctions can be calculated
using the identity \eqref{identity} in
Lemma \ref{identitylemma} as well.

\begin{proposition} \label{typetwosimplestproposition}
The type II wavefunction $\Phi^{\rm II}_{M,1}(z|\gamma_1,\dots,\gamma_M|x_1)$
is explicitly expressed as
\begin{align}
&\Phi^{\rm II}_{M,1}(z|\gamma_1,\dots,\gamma_M|x_1)
=\frac{z^{1/2}(1-\sqrt{-t} z)}{z^2-1} \nonumber \\
&\quad \times\sum_{\tau=\pm 1} \tau (z^\tau+\sqrt{-t})
\prod_{j=1}^{x_1-1}(\alpha_j+(1-\alpha_j \gamma_j)z^\tau)
\prod_{j=x_1+1}^M (1-\gamma_j z^\tau)
\prod_{j=1}^M (1-\gamma_j z^{-\tau}). \label{typetwosimplestmatrixelement}
\end{align}
\end{proposition}

\begin{proof}
We decompose
$\Phi^{\rm II}_{M,1}(z|\gamma_1,\dots,\gamma_M|x_1)$ as
\begin{align}
\Phi^{\rm II}_{M,1}&(z|\gamma_1,\dots,\gamma_M|x_1) \nonumber \\
&=-\sqrt{-t}z^{1/2}
\langle x_1 |\widetilde{B}(z,\{ \alpha \},\{ \gamma \})| 0^M \rangle
\langle 0^M |A(z,\{ \alpha \},\{ \gamma \})| 0^M \rangle \nonumber \\
&\quad +z^{-1/2}
\langle x_1 |\widetilde{A}(z,\{ \alpha \},\{ \gamma \})| x_1 \rangle
\langle x_1 |B(z,\{ \alpha \},\{ \gamma \})| 0^M \rangle \nonumber \\
&\quad +z^{-1/2}
\sum_{j=1}^{x_1-1}
\langle x_1 |\widetilde{A}(z,\{ \alpha \},\{ \gamma \})| j \rangle
\langle j |B(z,\{ \alpha \},\{ \gamma \})| 0^M \rangle,
\end{align}
to get
\begin{align}
\Phi^{\rm II}_{M,1}&(z|\gamma_1,\dots,\gamma_M|x_1)
=\prod_{j=x_1+1}^M (1-\gamma_j z)(1-\gamma_j z^{-1}) \nonumber \\
&\times \biggl[ -\sqrt{-t}z^{1/2}
\prod_{k=1}^{x_1-1}
\left\{(1-\alpha_k \gamma_k)z+\alpha_k\right\}
 \prod_{k=1}^{x_1}(1-\gamma_k z^{-1}) \nonumber \\
&\qquad +z^{-1/2}(t\gamma_{x_1}z+1)
\prod_{k=1}^{x_1-1}\left\{(1-\alpha_k \gamma_k)z^{-1}+\alpha_k\right\}
(1-\gamma_k z)
\nonumber \\
&\qquad+(t+1)z^{1/2}
\sum_{j=1}^{x_1-1}\prod_{k=1}^{j-1}\left\{
(1-\alpha_k \gamma_k)z^{-1}+\alpha_k\right\}(1-\gamma_k z)
\nonumber \\
&\qquad \times \prod_{k=j+1}^{x_1-1} 
\left\{(1-\alpha_k \gamma_k)z+\alpha_k\right\}
\prod_{k=j+1}^{x_1} (1-\gamma_k z^{-1})\biggr],
\end{align}
which, by using the equality \eqref{identity}
in Lemma \ref{identitylemma}, can be rewritten into the form
\eqref{typetwosimplestmatrixelement}.
\end{proof}
We call similarly calculate type II dual wavefunction $\overline{\Phi}^{\rm II}_{M,1}(z|\gamma_1,\dots,\gamma_M|\overline{x_1})$.
\begin{proposition}
The type II dual wavefunction $\overline{\Phi}^{\rm II}_{M,1}(z|\gamma_1,\dots,\gamma_M|\overline{x_1})$
is explicitly expressed as
\begin{align}
\overline{\Phi}^{\rm II}_{M,1}&(z|\gamma_1,\dots,\gamma_M|\overline{x_1})
\nonumber \\
&=\frac{t^{M}z^{1/2}(1-\sqrt{-t}z)}{t^2z^2-1}
\sum_{\tau=\pm 1} \tau
\left\{(tz)^\tau-\sqrt{-t}\right\}
\prod_{j=1}^{\overline{x_1}-1}
\left\{-\alpha_j+(1-\alpha_j \gamma_j)(tz)^\tau\right\}
\nonumber \\
&\quad \times \prod_{j=\overline{x_1}+1}^M 
\left\{1+\gamma_j (tz)^\tau\right\}
\prod_{j=1}^M \left\{1+\gamma_j (tz)^{-\tau}\right\}.
\end{align}
\end{proposition}

\subsection{Izergin-Korepin analysis}
We can use the Izergin-Korepin technique \cite{Ko,Iz,Motegi}
to extract the properties for the type II wavefunctions as well.

\begin{proposition}
\label{typetwoordinarypropertiesfordomainwallboundarypartitionfunction}
The type II wavefunctions
$\Phi^{\rm II}_{M,N}(z_1,\dots,z_N|\gamma_1,\dots,\gamma_M|x_1,\dots,x_N)$
satisfies the following properties. \\
\\
 (1) When $x_N=M$, the type II wavefunctions
$\Phi^{\rm II}_{M,N}(z_1,\dots,z_N|\gamma_1,\dots,\gamma_M|x_1,\dots,x_N)$
is a polynomial of degree $2N-1$ in $\gamma_M$.
\\
 (2) The following form
\begin{align}
\frac{\Phi^{\rm II}_{M,N}(z_1,\dots,z_N|\gamma_1,\dots,\gamma_M|x_1,\dots,x_N)}
{
\prod_{j=1}^N z_j^{j-1/2-N} (1-\sqrt{-t}z_j) \prod_{1 \le j < k \le N}
(1+tz_j z_k) (1+tz_j z_k^{-1})
},
\label{typetwosymmetrywavefunction}
\end{align}
is symmetric with respect to $z_1,\dots,z_N$,
and is invariant under the exchange $z_i \longleftrightarrow z_i^{-1}$ for
$i=1,\dots,N$.
\\
(3) The following recursive relations between the
type II wavefunctions hold if $x_N=M$:
\begin{align}
\Phi^{\rm II}_{M,N}&(z_1,\dots,z_N|\gamma_1,\dots,\gamma_M|x_1,\dots,x_N)
|_{\gamma_M=z_N}
\nonumber \\
&=z_N^{-1/2} \prod_{j=1}^{N}(t z_N z_j+1)
\prod_{j=1}^{N-1} (t+z_N z_j^{-1})
\prod_{j=1}^{M-1} \left\{(1-\alpha_j \gamma_j)z_N^{-1}+\alpha_j\right\}
\prod_{j=1}^{M-1} (1-\gamma_j z_N)
\nonumber \\
&\quad
\times \Phi^{\rm I}_{M-1,N-1}
(z_1,\dots,z_{N-1}|\gamma_1,\dots,\gamma_{M-1}|x_1,\dots,x_{N-1}). 
\label{typetwoordinaryrecursionwavefunction}
\end{align}

If $x_N \neq M$, the following factorizations hold for the type II
wavefunctions:
\begin{align}
\Phi_{M,N}^{\rm II}&(z_1,\dots,z_N|\gamma_1,\dots,\gamma_M|x_1,\dots,x_N)
 \nonumber \\
&=\prod_{j=1}^N (1-\gamma_M z_j)(1-\gamma_M z_j^{-1})
\Phi^{\rm II}_{M-1,N}(z_1,\dots,z_N|\gamma_1,\dots,\gamma_{M-1}|x_1,\dots,x_N).
\label{typetwofactorization}
\end{align}
\\
(4) The following holds for the case $N=1$, $x_1=M$:
\begin{align}
&\Phi^{\rm II}_{M,1}(z|\gamma_1,\dots,\gamma_M|M) \nonumber \\
&\quad=\frac{z^{1/2}(1-\sqrt{-t}z)}{z^2-1}
\sum_{\tau=\pm 1} \tau
(z^\tau+\sqrt{-t})
\prod_{j=1}^{M-1}\left\{\alpha_j+(1-\alpha_j \gamma_j)z^\tau\right\}
\prod_{j=1}^M (1-\gamma_j z^{-\tau})
.
\end{align}
\end{proposition}

\begin{proof}
Proposition \ref{typetwoordinarypropertiesfordomainwallboundarypartitionfunction} can be proved in the same way with
Proposition \ref{ordinarypropertiesfordomainwallboundarypartitionfunction}.

Property (1) and \eqref{typetwofactorization} in Property (3)
is the same with the corresponding ones in
Proposition \ref{ordinarypropertiesfordomainwallboundarypartitionfunction},
since the difference of the $K$-matrices one uses
for the type I and type II wavefunctions do not affect
the argument to show these properties.

Property (2) can be proved in the same way with
Brubaker-Bump-Chinta-Gunnells \cite{BBCG}, which they used the arguments by
Ivanov \cite{Iv,Ivthesis} based on the caduceus relations and the
fish relations (see Appendix \ref{app1} for details).
The denominator in \eqref{typetwosymmetrywavefunction}
is different from that in \eqref{symmetrywavefunction}.
This difference comes from the fact that for type II wavefunctions,
we use the $K$-matrix \eqref{generalizedkmatrixsecond} instead of
\eqref{generalizedkmatrix} at the boundary.

The difference of the $K$-matrices are also reflected in
the difference between
\eqref{typetwoordinaryrecursionwavefunction} in Property (3)
and Property (4) in
Proposition \ref{ordinarypropertiesfordomainwallboundarypartitionfunction}
and
Proposition
\ref{typetwoordinarypropertiesfordomainwallboundarypartitionfunction}.

\end{proof}

We also list below the properties
for the type II dual wavefunctions.

\begin{proposition} \label{typetwodualproposition}
The type II dual wavefunctions
$\overline{\Phi}^{\rm II}_{M,N}(z_1,\dots,z_N|\gamma_1,\dots,\gamma_M|\overline{x_1},\dots,\overline{x_N})$
satisfies the following properties. \\
\\
 (1) When $\overline{x_N}=M$, the type II dual wavefunctions
$\overline{\Phi}^{\rm II}_{M,N}(z_1,\dots,z_N|\gamma_1,\dots,\gamma_M|\overline{x_1},\dots,\overline{x_N})$
is a polynomial of degree $2N-1$ in $\gamma_M$.
\\
 (2) The following form
\begin{align}
\frac{\overline{\Phi}^{\rm II}_{M,N}(z_1,\dots,z_N|\gamma_1,\dots,\gamma_M|x_1,\dots,x_N)}
{
\prod_{j=1}^N z_j^{j-1/2-N} (1-\sqrt{-t}z_j) \prod_{1 \le j < k \le N}
(1+tz_j z_k) (1+tz_j z_k^{-1})
},
\label{typetwodualsymmetrywavefunction}
\end{align}
is symmetric with respect to $z_1,\dots,z_N$,
and is invariant under the exchange $z_i \longleftrightarrow z_i^{-1}$
for
$i=1,\dots,N$.
\\
(3) The following recursive relations between the
type II dual wavefunctions hold if $\overline{x_N}=M$:
\begin{align}
&\overline{\Phi}^{\rm II}_{M,N}(z_N,\dots,z_1|\gamma_1,\dots,\gamma_M|\overline{x_1},\dots,\overline{x_N})
|_{\gamma_M=-t^{-1}z_N^{-1}}
\nonumber \\
=&-\sqrt{-t}z_N^{1/2} \prod_{j=1}^{N} \Bigg(1+ \frac{1}{t z_N z_j} \Bigg)
\prod_{j=1}^{N-1} \Bigg(1+\frac{z_j}{tz_N} \Bigg)
\prod_{j=1}^{M-1} (t(1-\alpha_j \gamma_j)z_N-\alpha_j)
\prod_{j=1}^{M-1} (t+\gamma_j z_N^{-1})
\nonumber \\
&\times \overline{\Phi}^{\rm II}_{M-1,N-1}(z_{N-1},\dots,z_{1}|\gamma_1,\dots,\gamma_{M-1}|\overline{x_1},\dots,\overline{x_{N-1}})
. 
\label{typetworemark}
\end{align}

If $\overline{x_N} \neq M$, the following factorizations hold for the
type II dual wavefunctions:
\begin{align}
&\overline{\Phi}_{M,N}^{\rm II}
(z_1,\dots,z_N|\gamma_1,\dots,\gamma_M|\overline{x_1},\dots,\overline{x_N})
 \nonumber \\
=&\prod_{j=1}^N (1+t\gamma_M z_j)(t+\gamma_M z_j^{-1})
\overline{\Phi}^{\rm II}_{M-1,N}(z_1,\dots,z_N|\gamma_1,\dots,\gamma_{M-1}|\overline{x_1},\dots,\overline{x_N}).
\end{align}
\\
(4) The following holds for the case $N=1$, $\overline{x_1}=M$:
\begin{align}
&\overline{\Phi}^{\rm II}_{M,1}(z|\gamma_1,\dots,\gamma_M|M)
\nonumber \\
=&\frac{t^{M}z^{1/2}(1-\sqrt{-t}z)}{t^2z^2-1}
\sum_{\tau=\pm 1} \tau ((tz)^\tau-\sqrt{-t}) \nonumber \\
&\times \prod_{j=1}^{M-1}(-\alpha_j+(1-\alpha_j \gamma_j)(tz)^\tau)
\prod_{j=1}^M (1+\gamma_j (tz)^{-\tau}).
\end{align}
\end{proposition}

\subsection{Generalized Bump-Friedberg-Hoffstein Whittaker functions}
We first introduce the following symmetric functions
which generalizes the Whittaker functions introduced by
Bump, Friedberg and Hoffstein \cite{BFH}.

\begin{definition}
We define two generalized Bump-Friedberg-Hoffstein
Whittaker functions to be the
following determinants:
\begin{align}
o^\pm_\lambda(\{ z \}_N|\{ \alpha \}|\{ \gamma \}|t)
=\frac{H^\pm_{\lambda+\delta}(\{ z \}_N|\{ \alpha \}|\{ 
\gamma \}|t)}
{\mathrm{det}_N(z_k^{N-j+1}-z_k^{-N+j-1})}. \label{whittaker}
\end{align}
Here, $\{ z \}_N=\{z_1,\dots,z_N \}$ is a set of symmetric variables,
$\{ \alpha \}$ and $\{ \gamma \}$ are
sets of variables
$\{ \alpha \}=\{\alpha_1,\dots,\alpha_M \}$
and
$\{ \gamma \}=\{\gamma_1,\dots,\gamma_M \}$,
$\lambda$ denotes a Young diagram
$\lambda=(\lambda_1,\lambda_2,\dots,\lambda_N)$
with weakly decreasing non-negative integers
$\lambda_1 \ge \lambda_2 \ge \cdots \ge \lambda_N \ge 0$,
and $\delta=(N-1,N-2,\dots,0)$.
$H^\pm_{\mu}(\{ z \}_N|\{ \alpha \}|\{ \gamma \}|t)$
are $N \times N$ determinants
\begin{align}
H^\pm_{\mu}(\{ z \}_N|\{ \alpha \}|\{ \gamma \}|t)
=\mathrm{det}_N
(
h^\pm_{\mu_j}(z_k|\{ \alpha \}|\{ \gamma \}|t)
-
h^\pm_{\mu_j}(z_k^{-1}|\{ \alpha \}|\{ \gamma \}|t)
),
\end{align}
where
\begin{align}
h^\pm_\mu(z|\{ \alpha \}|\{ \gamma \}|t)
=(z \pm \sqrt{-t})
\prod_{j=1}^\mu(\alpha_j+(1-\alpha_j \gamma_j)z)
\prod_{j=\mu+2}^M(1-\gamma_j z)
\prod_{j=1}^M(1-\gamma_jz^{-1}).
\end{align}
\end{definition}
The generalized Whittaker functions \eqref{whittaker} are
symmetric with respect to the variables $z_1,\dots,z_N$,
and are invariant under the exchange $z_i \longleftrightarrow z_i^{-1}$
for $i=1,\dots,N$.
If one sets $\alpha_j=\gamma_j=0$ ($j=1,\dots,M$),
the generalized Whittaker functions
reduce to the Whittaker functions introduced by
Bump, Friedberg and Hoffstein \cite{BFH}.

\begin{theorem} \label{typetwomaintheoremstatement}
The type II wavefunctions
$\Phi^{\rm II}_{M,N}(z_1,\dots,z_N|\gamma_1,\dots,\gamma_M|x_1,\dots,x_N)$
are explicitly expressed
using the generalized Bump-Friedberg-Hoffstein Whittaker functions as
\begin{align}
&
\Phi^{\rm II}_{M,N}(z_1,\dots,z_N|\gamma_1,\dots,\gamma_M|x_1,\dots,x_N)
\nonumber \\
=&\prod_{j=1}^N z_j^{j-1/2-N}(1-\sqrt{-t}z_j)
\prod_{1 \le j<k \le N}(1+tz_j z_k)(1+tz_j z_k^{-1})
o^{+}_\lambda(\{ z \}_N|\{ \alpha \}|\{ \gamma \}|t),
\label{typetwomaintheorem}
\end{align}
under the relation $\lambda_j=x_{N-j+1}-N+j-1$, $j=1,\dots,N$.

The type II dual wavefunctions
$\overline{\Phi}^{\rm II}_{M,N}(z_1,\dots,z_N|\gamma_1,\dots,\gamma_M|\overline{x_1},\dots,\overline{x_N})$
are explicitly expressed
using the generalized Bump-Friedberg-Hoffstein Whittaker functions as
\begin{align}
&
\overline{\Phi}^{\rm II}_{M,N}(z_1,\dots,z_N|\gamma_1,\dots,\gamma_M|\overline{x_1},\dots,\overline{x_N})
\nonumber \\
=&t^{N(M-N)}
\prod_{j=1}^N z_j^{j-1/2-N}(1-\sqrt{-t}z_j)
\prod_{1 \le j<k \le N}(1+tz_j z_k)(1+tz_j z_k^{-1})
o^{-}_{\overline{\lambda}} ( \{ tz \}_N |\{-\alpha \}|\{-\gamma \}|t),
\label{typetwodualmaintheorem}
\end{align}
under the relation
$\overline{\lambda_j}=\overline{x_{N-j+1}}-N+j-1$, $j=1,\dots,N$,
and the symmetric variables are
$\displaystyle \{ tz \}_N=
\{ tz_1,\dots,tz_N \}$.
Moreover, the signs of the parameters of the generalized symplectic
Schur functions in the right hand side of
\eqref{dualmaintheorem} are now
inverted simultaneously: $\{-\alpha \}
=\{-\alpha_1,\dots,-\alpha_M \}$
and $\{-\gamma \}
=\{-\gamma_1,\dots,-\gamma_M \}$.
\end{theorem}
The correspondence
\eqref{typetwomaintheorem} 
is a generalization of the conjecture given
in the work by Brubaker-Bump-Chinta-Gunnells
\cite{BBCG} for the case $\alpha_j=\gamma_j=0$, $j=1,\dots,M$,
in which they conjectured the relation between the wavefunctions of
type $B$ ice models and the Whittaker functions by Bump-Friedberg-Hoffstein
\cite{BFH}.

\begin{proof}
The correspondence \eqref{typetwomaintheorem}
can be proved in the same way with proving
\eqref{maintheorem}.
We show that the right hand side of
\eqref{typetwomaintheorem} satisfies all the Properties
in Proposition \ref{typetwoordinarypropertiesfordomainwallboundarypartitionfunction}.

Let us illustrate the proof of Property (3)
for the case $x_N=M$, which is the hardest thing to check.

Again using \eqref{factorization}, 
the definition of the determinant and the correspondence
between the positions of particles $\{ x \}$ and
the Young diagrams $\{ \lambda \}$, we 
rewrite the right hand side of \eqref{typetwomaintheorem} as
\begin{align}
&F^{\rm II}_{M,N}(z_1,\dots,z_N|\gamma_1,\dots,\gamma_M|x_1,\dots,x_N)
\nonumber \\
&\quad :=
\frac{
\prod_{j=1}^N z_j^{j-1/2-N}(1-\sqrt{-t}z_j)
\prod_{1 \le j<k \le N}(1+tz_j z_k)(1+tz_j z_k^{-1})
}{
(-1)^N \prod_{j=1}^N z_j^{j-1-N}(1-z_j^2)
\prod_{1 \le j<k \le N}(1-z_j z_k)(-1+z_j z_k^{-1})
} \nonumber \\
&\qquad \times\sum_{\sigma \in S_N}
\sum_{\tau_1=\pm 1,\dots,\tau_N=\pm 1} (-1)^\sigma (-1)^{|\tau|}
\prod_{j=1}^N \left(z_{\sigma(j)}^{\tau_{\sigma(j)}}+\sqrt{-t} \right)
\prod_{j=1}^N \prod_{k=1}^{x_j-1}
\left(\alpha_k+(1-\alpha_k \gamma_k)z_{\sigma(j)}^{\tau_{\sigma(j)}}
\right) \nonumber \\
&\qquad \times
\prod_{j=1}^N \prod_{k=x_j+1}^{M}
\left(1-\gamma_k z_{\sigma(j)}^{\tau_{\sigma(j)}}\right)
\prod_{j=1}^N \prod_{k=1}^{M}
\left(1-\gamma_k z_{\sigma(j)}^{-\tau_{\sigma(j)}}\right), 
\label{typetwoclearrepresentaitonforproof}
\end{align}
where $|\tau|$ denotes the number of $\tau_j$s
satisfying $\tau_j=-1$.

When $x_N=M$,
only the summands satisfying $\sigma(N)=N$, $\tau_N=-1$ in 
\eqref{typetwoclearrepresentaitonforproof} survive
after specializing $\gamma_M$ to $\gamma_M=z_N$.
Then we find that
$F^{\rm II}_{M,N}(z_1,\dots,z_N|\gamma_1,\dots,\gamma_M|x_1,\dots,x_N)
|_{\gamma_M=z_N}$ can be rewritten as
\begin{align}
F^{\rm II}_{M,N}&(z_1,\dots,z_N|\gamma_1,\dots,\gamma_M|x_1,\dots,x_N)
|_{\gamma_M=z_N}
\nonumber \\
&=-\frac{z_N^{1/2}(1-\sqrt{-t}z_N)}{1-z_N^2}
\frac{\prod_{j=1}^{N-1}(1+tz_j z_N)(1+tz_j z_N^{-1})}
{\prod_{j=1}^{N-1}(1-z_j z_N) (-1+z_j z_N^{-1})} \nonumber \\
&\quad \times
\frac{
\prod_{j=1}^{N-1} z_j^{1/2} (1-\sqrt{-t}z_j)
\prod_{1 \le j<k \le N-1}(1+tz_j z_k)(1+tz_j z_k^{-1})
}{
(-1)^{N-1} \prod_{j=1}^{N-1} (1-z_j^2)
\prod_{1 \le j<k \le N-1}(1-z_j z_k)(-1+z_j z_k^{-1})
}
\nonumber \\
&\quad \times \sum_{\sigma \in S_{N-1}}
\sum_{\tau_1=\pm 1,\dots,\tau_{N-1}=\pm 1}
(-1)(-1)^\sigma (-1)^{|\tau|} 
%\nonumber \\
%&\quad \times 
(z_N^{-1}+\sqrt{-t})
\prod_{j=1}^{N-1} \Bigg( z_{\sigma(j)}^{\tau_{\sigma(j)}}
+\sqrt{-t} \Bigg) \nonumber \\
&\quad \times \prod_{j=1}^{N-1} \prod_{k=1}^{x_j-1}
\left\{\alpha_k+(1-\alpha_k \gamma_k)z_{\sigma(j)}^{\tau_{\sigma(j)}}
\right\}
\prod_{k=1}^{M-1}
\left\{\alpha_k+(1-\alpha_k \gamma_k)z_N^{-1}
\right\}
\nonumber \\
&\quad \times
\prod_{j=1}^{N-1} \prod_{k=x_j+1}^{M-1}
\left(1-\gamma_k z_{\sigma(j)}^{\tau_{\sigma(j)}}\right)
\prod_{j=1}^{N-1}
\left(1-z_N z_{\sigma(j)}^{\tau_{\sigma(j)}}\right)
\nonumber \\
&\quad \times
(1-z_N^2)
\prod_{k=1}^{M-1}
(1-\gamma_k z_N )
\prod_{j=1}^{N-1}
\left(1-z_N z_{\sigma(j)}^{-\tau_{\sigma(j)}}\right)
\prod_{j=1}^{N-1} \prod_{k=1}^{M-1}
\left(1-\gamma_k z_{\sigma(j)}^{-\tau_{\sigma(j)}}\right). 
\label{interwhittaker}
\end{align}
We again use the obvious identity
\begin{align}
\prod_{j=1}^{N-1}
\left(1-z_N z_{\sigma(j)}^{\tau_{\sigma(j)}}\right)
\left(1-z_N z_{\sigma(j)}^{-\tau_{\sigma(j)}}\right)
=
\prod_{j=1}^{N-1}
(1-z_N z_j)
(1-z_N z_{j}^{-1}),
\end{align}
and simplify
\eqref{interwhittaker} as
\begin{align}
&F^{\rm II}_{M,N}(z_1,\dots,z_N|\gamma_1,\dots,\gamma_M|x_1,\dots,x_N)
|_{\gamma_M=z_N}
\nonumber \\
&\quad =z_N^{-1/2} \prod_{j=1}^{N}(t z_N z_j+1)
\prod_{j=1}^{N-1} (t+z_N z_j^{-1})
\prod_{j=1}^{M-1} ((1-\alpha_j \gamma_j)z_N^{-1}+\alpha_j)
\prod_{j=1}^{M-1} (1-\gamma_j z_N)
\nonumber \\
&\qquad \times \frac{
\prod_{j=1}^{N-1} z_j^{j-1/2-(N-1)}(1-\sqrt{-t}z_j)
\prod_{1 \le j<k \le N-1}(1+tz_j z_k)(1+tz_j z_k^{-1})
}{
(-1)^{N-1} \prod_{j=1}^{N-1} z_j^{j-1-(N-1)}(1-z_j^2)
\prod_{1 \le j<k \le N-1}(1-z_j z_k)(-1+z_j z_k^{-1})
} \nonumber \\
&\qquad \times\sum_{\sigma \in S_{N-1}}
\sum_{\tau_1=\pm 1,\dots,\tau_{N-1}=\pm 1} (-1)^\sigma (-1)^{|\tau|}
\prod_{j=1}^{N-1} 
\left(z_{\sigma(j)}^{\tau_{\sigma(j)}}+\sqrt{-t} \right)\nn \\
&\qquad \times \prod_{j=1}^{N-1} \prod_{k=1}^{x_j-1}
\left\{\alpha_k+(1-\alpha_k \gamma_k)z_{\sigma(j)}^{\tau_{\sigma(j)}}
\right\}
\prod_{j=1}^{N-1} \prod_{k=x_j+1}^{M-1}
\left(1-\gamma_k z_{\sigma(j)}^{\tau_{\sigma(j)}}\right)
\prod_{j=1}^{N-1} \prod_{k=1}^{M-1}
\left(1-\gamma_k z_{\sigma(j)}^{-\tau_{\sigma(j)}}\right)
\nonumber \\
&\quad =z_N^{-1/2} \prod_{j=1}^{N}(t z_N z_j+1)
\prod_{j=1}^{N-1} (t+z_N z_j^{-1})
\prod_{j=1}^{M-1} ((1-\alpha_j \gamma_j)z_N^{-1}+\alpha_j)
\prod_{j=1}^{M-1} (1-\gamma_j z_N)
\nonumber \\
&\qquad \times F^{\rm II}_{M-1,N-1}(z_1,\dots,z_{N-1}|\gamma_1,\dots,\gamma_{M-1}|x_1,\dots,x_{N-1}),
\end{align}
hence we find that $F^{\rm II}_{M,N}(z_1,\dots,z_N|\gamma_1,\dots,\gamma_M|x_1,\dots,x_N)$
satisfies Property (3) for the case $x_N=M$
in Proposition \ref{typetwoordinarypropertiesfordomainwallboundarypartitionfunction}.

\end{proof}

\section{Dual Cauchy formula for generalized Bump-Friedberg-Hoffstein
Whittaker functions}
We can derive the dual Cauchy formula for the generalized
Whittaker functions in the same with deriving the one
for the generalized symplectic Schur functions.
We now deal the type II domain wall boundary partition functions.
\begin{align}
Z^{\rm II}_M(z_1,\dots,z_M|\gamma_1,\dots,\gamma_M)
=\Phi^{\rm II}_{M,M}(z_1,\dots,z_M|\gamma_1,\dots,\gamma_M|1,\dots,M).
\label{typetwospecialcase}
\end{align}

One again first shows the following factorization formula
for the type II domain wall boundary partition functions.

\begin{theorem} \label{typetwoinhomogeneousexpression}
The type II domain wall boundary partition functions
$Z^{\rm II}_M(z_1,\dots,z_M|\gamma_1,\dots,\gamma_M)$ have
the following factorized form:

\begin{align}
&Z^{\rm II}_M(z_1,\dots,z_M|\gamma_1,\dots,\gamma_M)
=
\prod_{j=1}^M z_j^{j-1/2-M} (1-\sqrt{-t}z_j)(1+\sqrt{-t}\gamma_j)
\nonumber \\
&\quad \times \prod_{1 \le j<k \le M}(1+tz_j z_k)(1+tz_j z_k^{-1})
\prod_{1 \le j<k \le M} \left\{1+\alpha_j(\gamma_k-\gamma_j)\right\}
\prod_{1 \le j<k \le M} (1-\gamma_j \gamma_k).
\label{typetwodeterminantinhomogeneheous}
\end{align}

\end{theorem}

\begin{proof}
Since the type II domain wall boundary partition functions
$Z^{\rm II}_M(z_1,\dots,z_M|\gamma_1,\dots,\gamma_M)$
are
special cases of the type II wavefunctions
$\Phi^{\rm II}_{M,M}(z_1,\dots,z_M|\gamma_1,\dots,\gamma_M|1,\dots,M)$
\eqref{typetwospecialcase}, it is enough to
prove the following special case of Proposition
\ref{typetwoordinarypropertiesfordomainwallboundarypartitionfunction}.
\begin{proposition}
\label{typetwospecializerginkorepin}
The type II domain wall boundary partition functions
$Z^{\rm II}_M(z_1,\dots,z_M|\gamma_1,\dots,\gamma_M)$
satisfies the following properties. \\
\\
 (1) The type II domain wall boundary partition functions
$Z^{\rm II}_M(z_1,\dots,z_M|\gamma_1,\dots,\gamma_M)$
is a polynomial of degree $2M-1$ in $\gamma_M$.
\\
 (2) The following form
\begin{align}
\frac{Z^{\rm II}_{M}(z_1,\dots,z_M|\gamma_1,\dots,\gamma_M)}
{
\prod_{j=1}^M z_j^{j-1/2-M} (1-\sqrt{-t}z_j) \prod_{1 \le j < k \le M}
(1+tz_j z_k) (1+tz_j z_k^{-1})
},
\end{align}
is symmetric with respect to $z_1,\dots,z_M$,
and is invariant under the exchange $z_i \longleftrightarrow z_i^{-1}$
for $i=1,\dots,M$.
\\
(3) The following recursive relations between the
type II domain wall boundary partition functions hold:
\begin{align}
Z^{\rm II}_{M}&(z_1,\dots,z_M|\gamma_1,\dots,\gamma_M)
|_{\gamma_M=z_M}
\nonumber \\
&=
z_M^{-1/2}
\prod_{j=1}^{M}(t z_M z_j+1)
\prod_{j=1}^{M-1} (t+z_M z_j^{-1})
\prod_{j=1}^{M-1} ((1-\alpha_j \gamma_j)z_M^{-1}+\alpha_j)
\prod_{j=1}^{M-1} (1-\gamma_j z_M)
\nonumber \\
&\quad \times Z^{\rm II}_{M-1}(z_1,\dots,z_{M-1}|\gamma_1,\dots,\gamma_{M-1}).
\end{align}
\\
(4) The following holds for the case $M=1$:
\begin{align}
Z^{\rm II}_{1}(z|\gamma_1)
=&\frac{z^{1/2}(1-\sqrt{-t}z)}{z^2-1}
\sum_{\tau=\pm 1} \tau
(z^\tau+\sqrt{-t}) (1-\gamma_1 z^{-\tau}).
\end{align}
\end{proposition}
It is easy to see that the right hand side of
\eqref{typetwodeterminantinhomogeneheous} satisfies
all the properties listed in
Proposition \ref{typetwospecializerginkorepin}.

\end{proof}

We now derive the dual Cauchy formula for
the generalized Whittaker functions.

\begin{theorem}
The following dual Cauchy formula holds
for the generalized Whittaker functions
with sets variables $\{ x \}_N=\{x_1,\dots,x_N \}$, $\{ y \}_M=\{y_1,\dots,y_M \}$,
$\{ \alpha \}=\{\alpha_1,\dots,\alpha_{N+M} \}$,
$\{ \gamma \}=\{\gamma_1,\dots,\gamma_{N+M} \}$,
\begin{align}
&\sum_{\lambda \subseteq M^N}o^+_\lambda(\{ x \}_N|\{ \alpha \}|\{
\gamma \}|t)
o^-_{\hat{\lambda}}(\{ y \}_M|\{-\alpha \}|\{-\gamma \}|t) \nonumber \\
&\quad =\prod_{j=1}^M y_j^{-N}
\prod_{j=1}^{M+N} (1+\sqrt{-t}\gamma_j)
\prod_{j=1}^N \prod_{k=1}^M (1+x_j y_k)(1+x_j^{-1} y_k)
\nonumber \\
&\qquad \times \prod_{1 \le j<k \le N+M}
\left\{1+\alpha_j(\gamma_k-\gamma_j)\right\}
\prod_{1 \le j<k \le N+M}(1-\gamma_j \gamma_k),
\label{typetwodualcauchy}
\end{align}
where
$\{ -\alpha \}=\{-\alpha_1,\dots,-\alpha_{N+M} \}$,
$\{ -\gamma \}=\{-\gamma_1,\dots,-\gamma_{N+M} \}$
and
$\hat{\lambda}
=(\hat{\lambda}_1,\dots,\hat{\lambda}_M)$ is the partition
of the Young diagram $\lambda=(\lambda_1,\dots,\lambda_N)$ given by
\begin{align}
\hat{\lambda}_i=|\{j \ | \ \lambda_j  \le M-i \}|.
\end{align}
\end{theorem}

\begin{proof}
The type II domain wall boundary partition functions
$Z^{\rm II}_M(z_1,\dots,z_M|\gamma_1,\dots,\gamma_M)$
have the factorized form
\eqref{typetwodeterminantinhomogeneheous}
in Theorem \ref{typetwoinhomogeneousexpression} on one hand.

On the other hand, one can evaluate
the domain wall boundary partition functions
by inserting the completeness relation and
using the correspondences between
the wavefunctions, the dual wavefunctions
and the generalized Whittaker functions
\eqref{typetwomaintheorem} and \eqref{typetwodualmaintheorem} as
\begin{align}
Z^{\rm II}_M&(z_1,\dots,z_M|\gamma_1,\dots,\gamma_M) \nonumber \\
&=\sum_{x \sqcup \overline{x}=\{1,2,\dots,M \}} 
\overline{\Phi}_{M,M-N}^{\rm I}
(z_1,\dots,z_{M-N}|\gamma_1,\dots,\gamma_M|\overline{x_1},\dots,
\overline{x_{M-N}}) \nonumber \\
&\quad \times \Phi^{\rm I}_{M,N}
(z_{M-N+1},\dots,z_M|\gamma_1,\dots,\gamma_M|x_1,\dots,x_N),
\nonumber \\
&=\sum_{\lambda \subseteq (M-N)^N}
t^{N(M-N)}
\prod_{j=1}^{M-N} z_j^{j-1/2-M+N}(1-\sqrt{-t}z_j)
\prod_{1 \le j<k \le M-N}(1+tz_j z_k)(1+tz_j z_k^{-1}) \nonumber \\
&\quad \times
o^-_{\overline{\lambda}}(tz_1,\dots,tz_{M-N}|\{-\alpha \}|\{-\gamma \}|t) \nonumber \\
&\quad \times
\prod_{j=M-N+1}^M z_j^{j-1/2-M}(1-\sqrt{-t}z_j)
\prod_{M-N+1 \le j<k \le M}(1+tz_j z_k)(1+tz_j z_k^{-1})
\nonumber \\
&\quad \times
o^+_\lambda(z_{M-N+1},\dots,z_M|
\{ \alpha \}|\{ \gamma \}|t).
\label{typetwocomparisontwo}
\end{align}
Comparing the two ways of evaluations
\eqref{typetwocomparisontwo} and
\eqref{typetwodeterminantinhomogeneheous}
and cancelling common factors leads to
\begin{align}
&\prod_{j=1}^{M-N} (tz_j)^{-N}
\prod_{j=1}^M (1+\sqrt{-t}\gamma_j)
\prod_{\substack{1 \le j \le M-N \\ M-N+1 \le k \le M}}
(1+tz_j z_k)(1+tz_j z_k^{-1}) \nonumber \\
&\qquad \times \prod_{1 \le j<k \le M}(1+\alpha_j(\gamma_k-\gamma_j))
\prod_{1 \le j<k \le M}(1-\gamma_j \gamma_k)
\nonumber \\
&\quad =\sum_{\lambda \subseteq (M-N)^N}
o^-_{\overline{\lambda}}(tz_1,\dots,tz_{M-N}|\{-\alpha \}|\{-\gamma \}|t)
o^+_\lambda(z_{M-N+1},\dots,z_M|
\{ \alpha \}|\{ \gamma \}|t),
\end{align}
which, after some reparametrization, can be rewritten in
the form \eqref{typetwodualcauchy}.
\end{proof}

\section{Conclusion}
In this paper, we introduced and studied in detail
generalizations of the
free-fermionic six-vertex model under two types of reflecting boundary
which were recently introduced by
Ivanov and Brubaker-Bump-Chinta-Gunnells.
We derived the explicit forms of the wavefunctions and the dual wavefunctions
by using the Izergin-Korepin technique, a technique
which belongs to a class of the quantum inverse scattering method.
For the case which one uses the $K$-matrix by Ivanov \cite{Iv,Ivthesis}
at the boundary, we showed the wavefunctions can be expressed as a product of factors and the generalized symplectic Schur functions,
which generalizes the correspondence by Ivanov.
For the case which one uses the $K$-matrix by Brubaker-Bump-Chinta-Gunnells
\cite{BBCG},
we proved the wavefunctions are expressed as a product of factors
and generalizations of the Whittaker functions introduced by
Bump-Friedberg-Hoffstein \cite{BFH}. The correspondence reduces
to the conjecture by Brubaker-Bump-Chinta-Gunnells \cite{BBCG}
when all the factorial parameters are set to zero.

We also derived the factorized form of the
domain wall boundary partition functions for both models.
As a consequence, we derived the dual Cauchy formulas
for the generalized type $B$ Whittaker functions
and the symplectic Schur functions.

It is interesting to apply the analysis given in this paper
to other boundary conditions.
Formulating analogues of the wavefunctions
by looking at the graphical representations of
the partition functions introduced and studied
by Kuperberg \cite{Ku} or Brubaker-Schultz \cite{BS},
and applying the Izergin-Korepin analysis
to find the explicit forms, and deriving algebraic identities
is an interesting problem to be studied.
It should make the connections between
number theory and integrable lattice models
more fruitful.

\section*{Acknowledgments}
This work was partially supported by grant-in-Aid
for and Scientific Research (C) No. 18K03205 and No. 16K05468.
The authors thank the referee for helpful comments and suggestions
to improve the paper.

\begin{appendix}

\section{Some calculations for Property (2) in Proposition 
\ref{ordinarypropertiesfordomainwallboundarypartitionfunction} and 
\ref{typetwoordinarypropertiesfordomainwallboundarypartitionfunction}} 
\label{app1}
Here we present some calculations for 
the proof of Property (2) in Proposition 
\ref{ordinarypropertiesfordomainwallboundarypartitionfunction}
and 
\ref{typetwoordinarypropertiesfordomainwallboundarypartitionfunction}. 
First let us show the 
commutation relation of the $B$-operators
\begin{equation}
\mathcal{B}(z_i)\mathcal{B}(z_j)=\frac{z_j+t z_i}{z_i+tz_j}
\mathcal{B}(z_j)\mathcal{B}(z_i),
\end{equation}
where this relation holds for both type I and II $B$-operators.
This commutation relation can be proven by use of the graphical representation
in the following manner.
\begin{equation}
\includegraphics[width=0.99\textwidth]{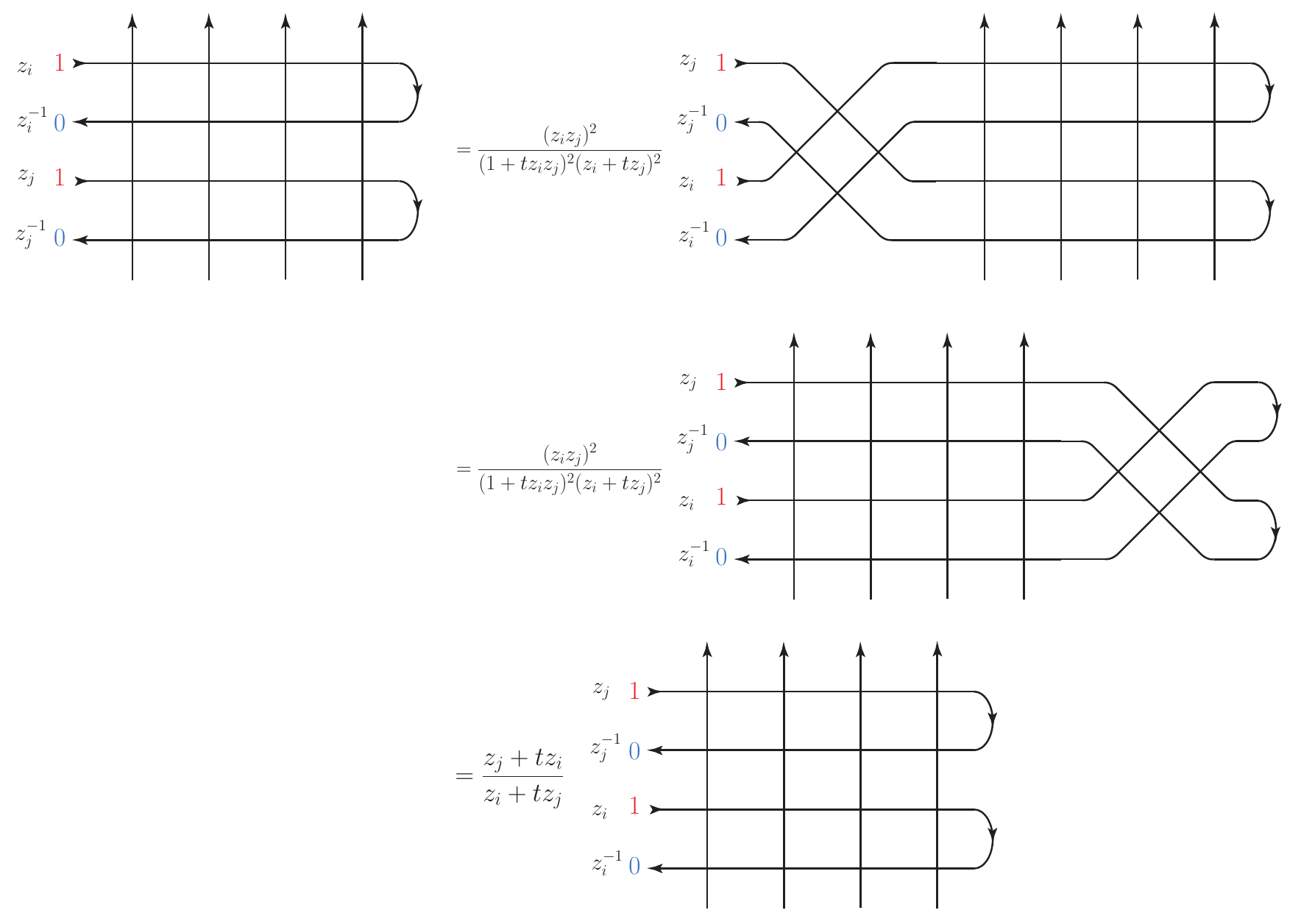}
\end{equation}
Let us explain in order. For the first equality, we have used the 
elements of the $R$-matrices shown in Figure \ref{pic-vertex}.
In going from the first to the second line, we have applied 
the Yang-Baxter relation \eqref{YBR} to commute the $B$-operators. Finally
the following ``caduceus relation" \cite{Iv,Ivthesis}
\begin{equation}
\includegraphics[width=0.75\textwidth]{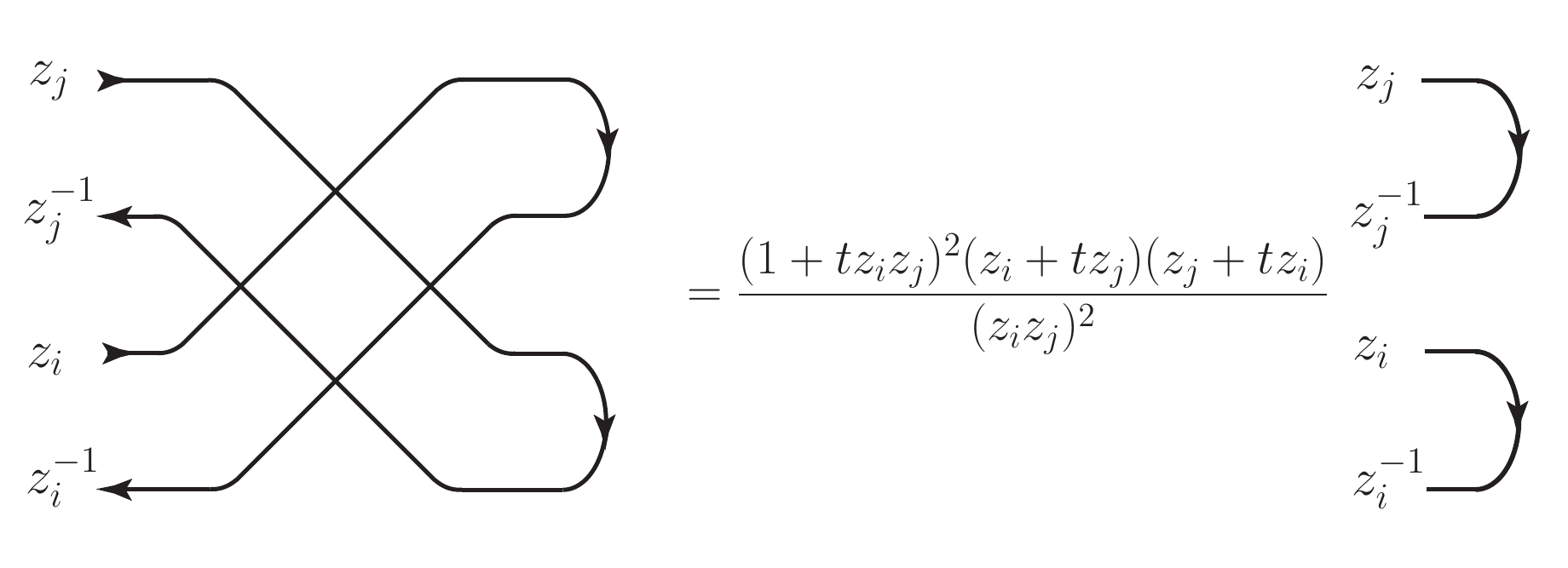}
\end{equation}
have been applied to the last line. The caduceus relation can be 
easily proved by the reflection equation
\eqref{RE}:
\begin{equation}
\includegraphics[width=0.8\textwidth]{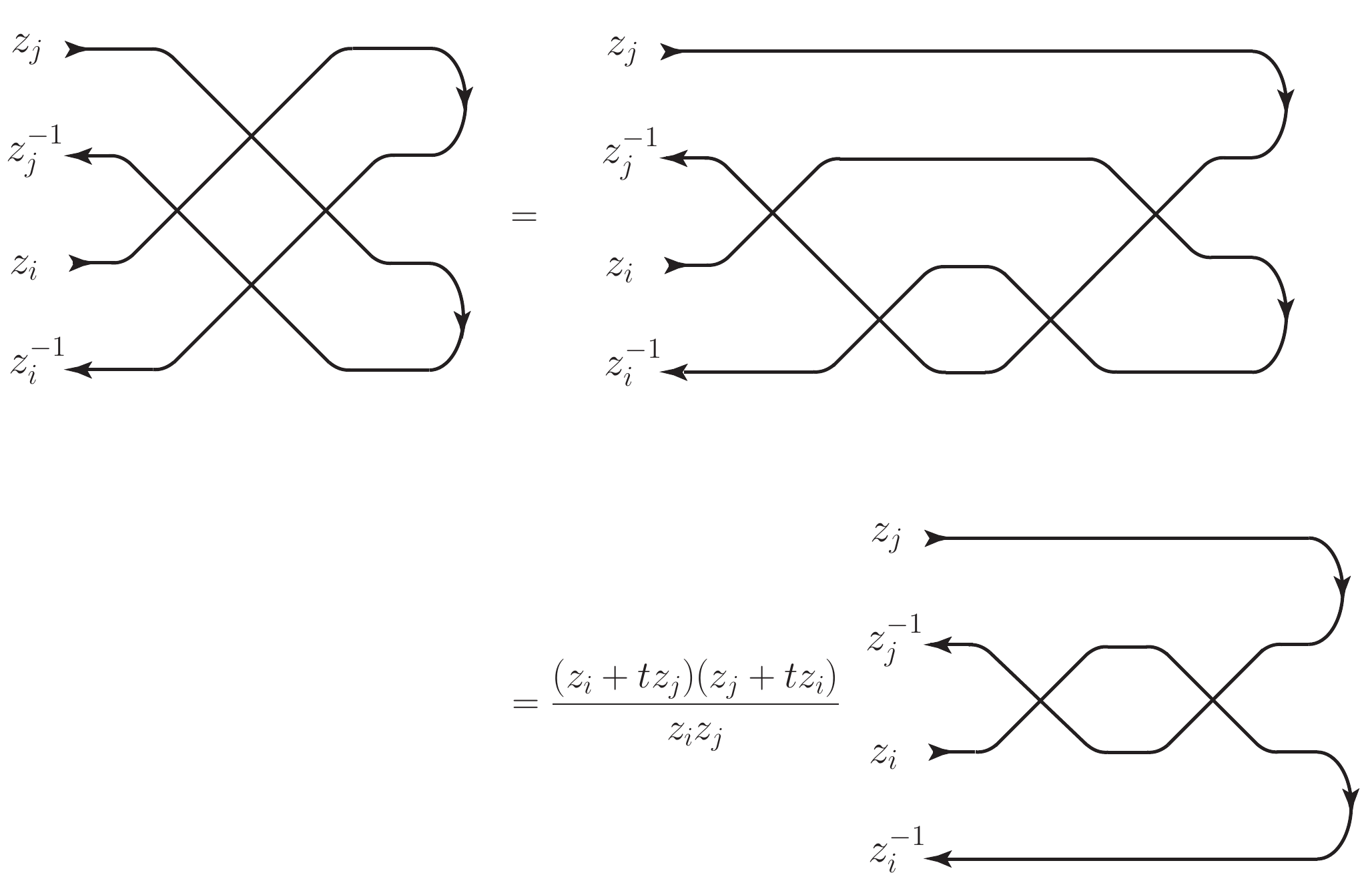}
\end{equation}
Here the reflection equation has been applied in the first equality.
In going from the first to the second equality and the second to the third 
equality, we have applied the following unitarity relations.
\begin{equation}
\includegraphics[width=0.98\textwidth]{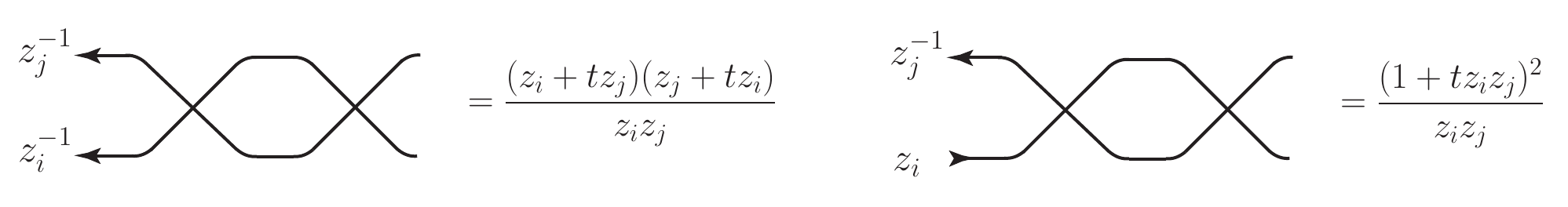}
\end{equation}

Thus the quantity 
\begin{equation}
\frac{\mathcal{B}(z_N)\cdots \mathcal{B}(z_1)}{\prod_{1\le j<k \le N}(z_k+t z_j)}
\label{comm-B}
\end{equation}
are symmetric with respect to the variables $z_1,\dots,z_N$, and hence
\eqref{symmetrywavefunction} and \eqref{typetwosymmetrywavefunction} are
 symmetric with respect to  $z_1,\dots,z_N$.

Next along the procedure in \cite{Iv,Ivthesis}, we show the quantities
\eqref{symmetrywavefunction} and \eqref{dualsymmetrywavefunction}
are invariant under the exchange $z_i\longleftrightarrow z_i^{-1}$. Due to the symmetry of 
\eqref{comm-B} with respect to $z_1,\dots,z_N$, it is sufficient to consider the case
$z_N\longleftrightarrow z_N^{-1}$ in Figure \ref{picturewavefunction}. In fact, 
the type $\Gamma$ $L$-operator on the last row in Figure 
\ref{picturewavefunction} can be replaced by the second type (type $\Delta$) $L$-operator 
$\widetilde{L}$ as follows. Because the elements of the $L$ appearing on the
last row are only the following three types:
\[
\includegraphics[width=0.4\textwidth]{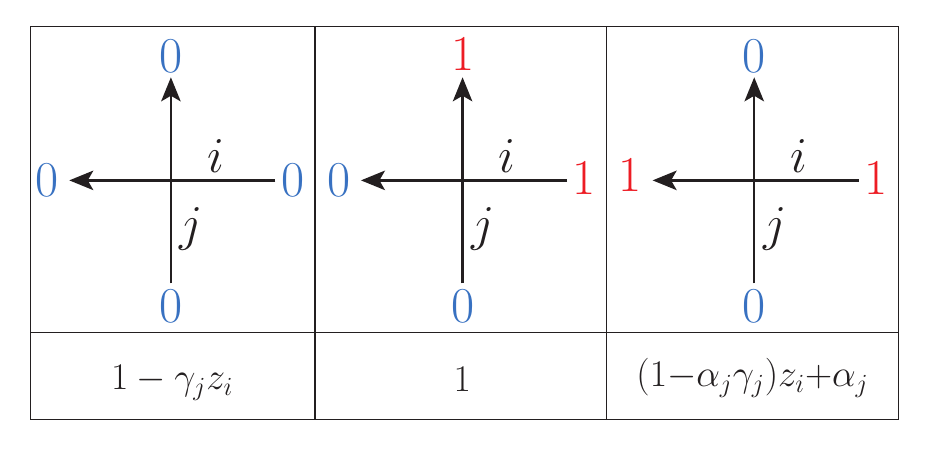}
\]
which are equivalent to the elements of $\widetilde{L}$:
\[
\includegraphics[width=0.4\textwidth]{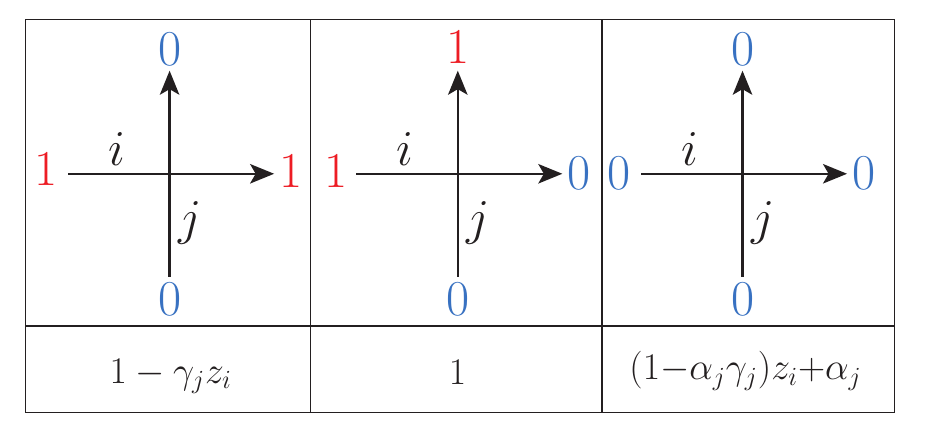}
\]
Thus we replace $\widetilde{L}$ to $L$ with $0\leftrightarrow 1$ without 
changing the wavefunctions.
\begin{equation}
\includegraphics[width=0.4\textwidth]{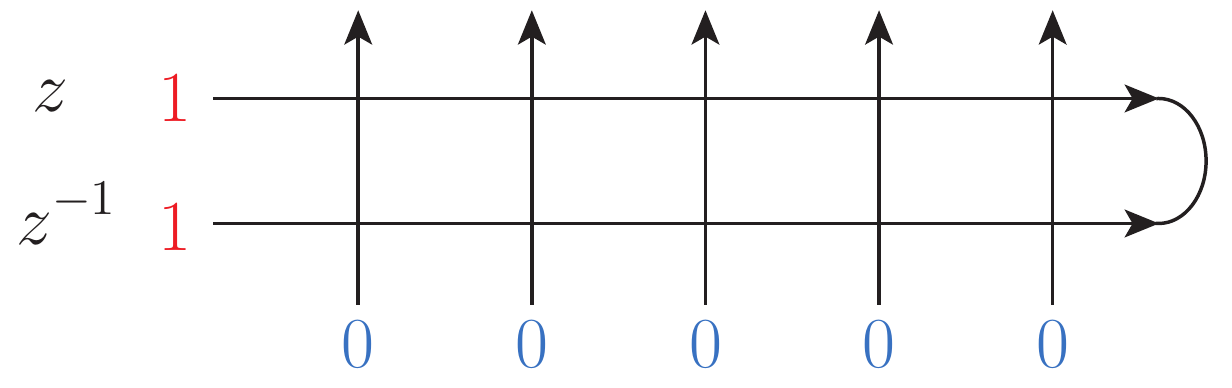}
\end{equation}
where the type I and type II $K$-matrices are, respectively, modified as
\[
\includegraphics[width=0.65\textwidth]{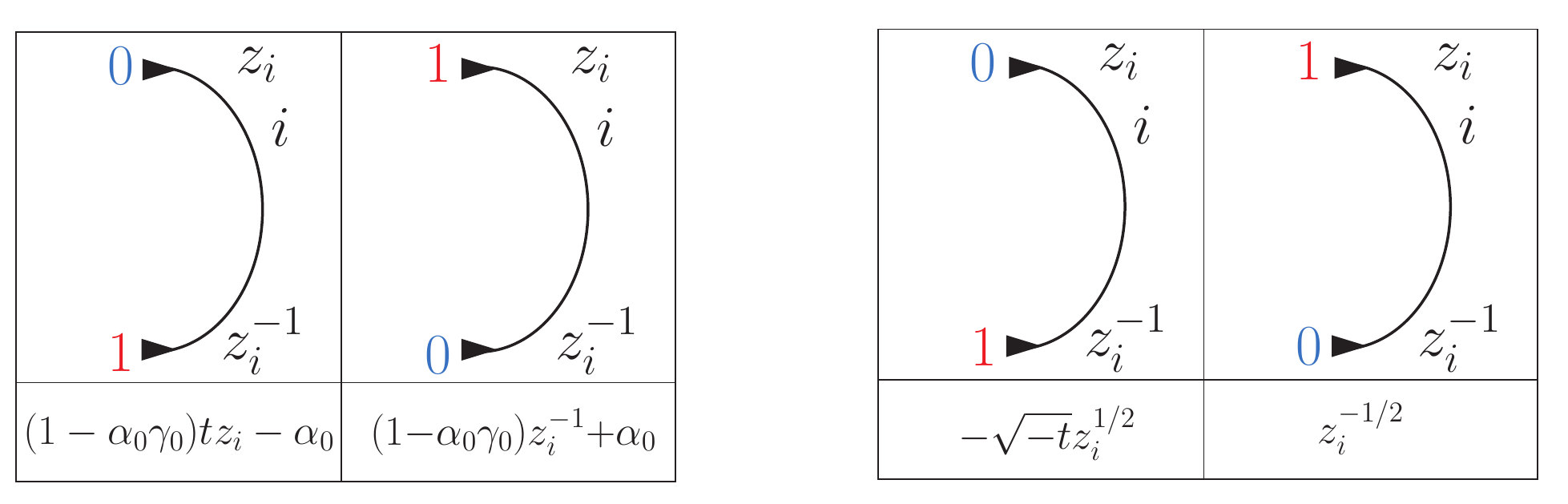}
\]
Multiplying the type $\Delta\Delta$ $R$-matrix  and using the Yang-Baxter
relation \eqref{YBR}, we obtain
\begin{equation}
\includegraphics[width=0.98\textwidth]{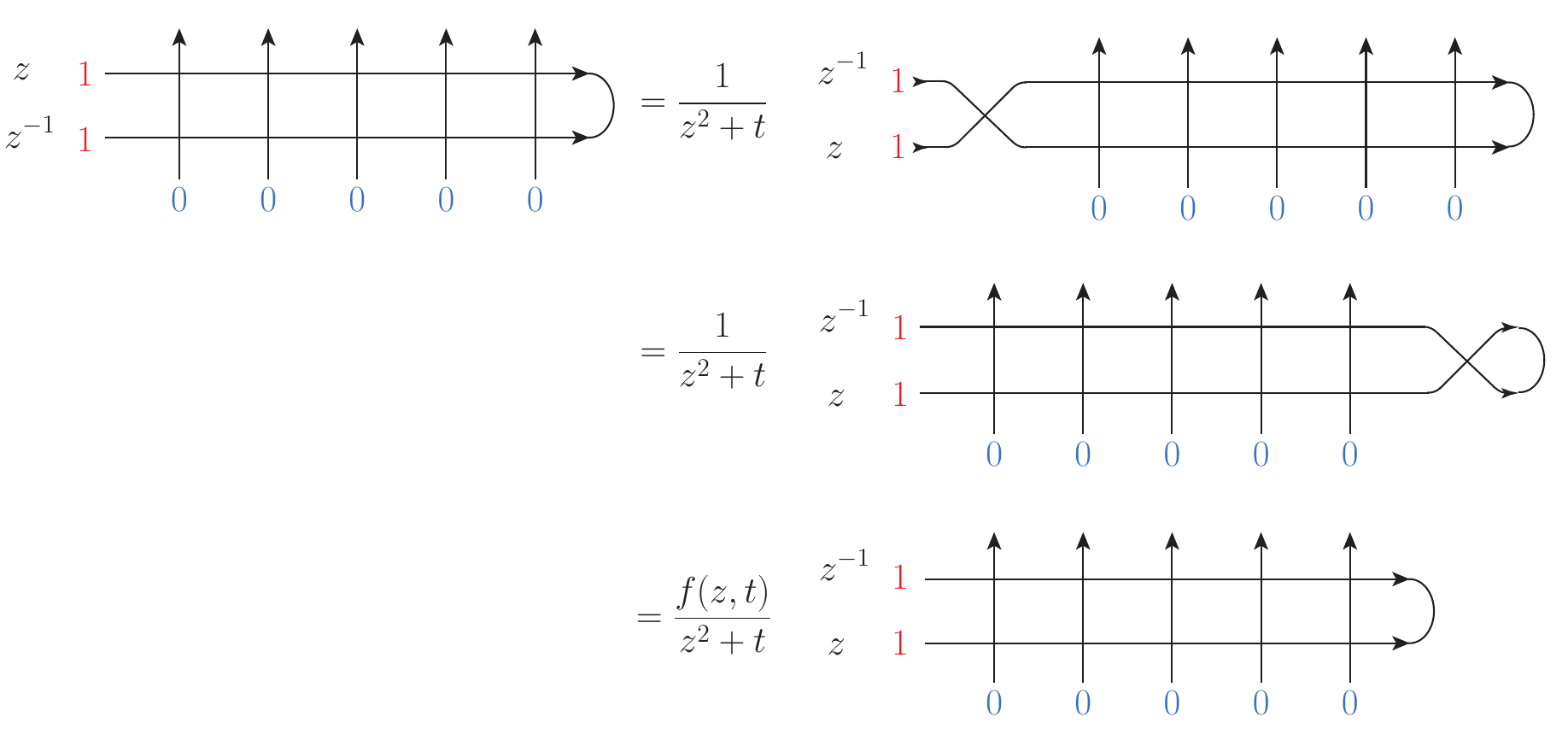}
\end{equation}
where we have used the elements of the $R$-matrix (see Figure \ref{pic-vertex})
and the following relation called ``fish relation" \cite{Iv,Ivthesis}:
\begin{equation}
\includegraphics[width=0.45\textwidth]{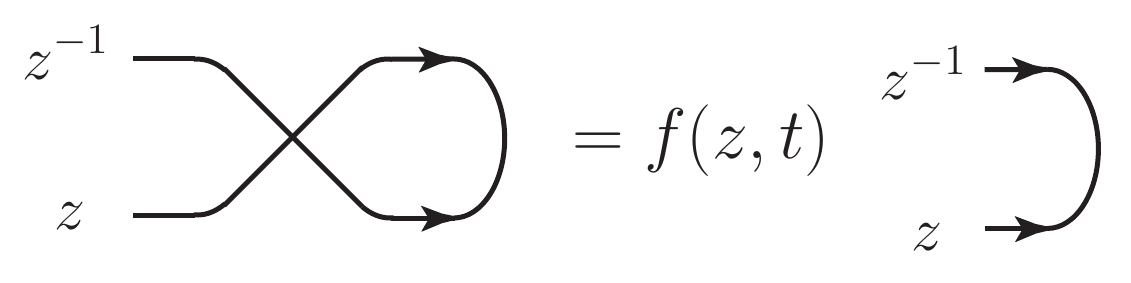}
\end{equation}
The overall factor $f(z,t)$ is given by
\begin{equation}
f(z,t)=\begin{cases}
1+tz^2 & \text{ for type I}\\
(\sqrt{-t}+z)(1-\sqrt{-t}z)  & \text{ for type II}
\end{cases}.
\end{equation}
Namely, the following expression holds:
\begin{equation}
\includegraphics[width=0.95\textwidth]{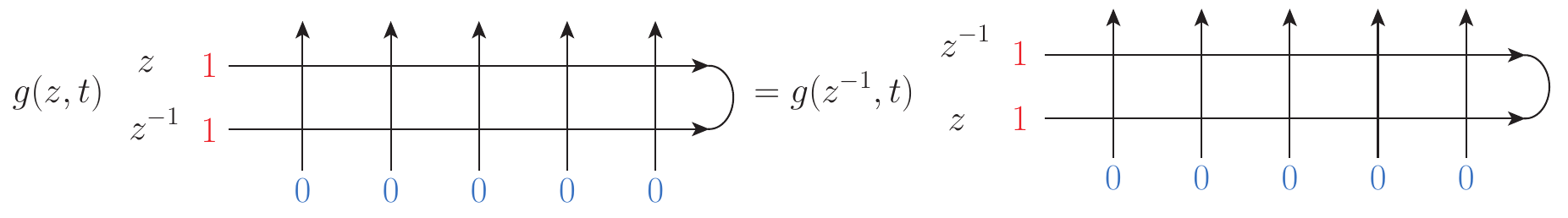}
\end{equation}
with
\begin{equation}
g(z,t)=\begin{cases}
\dfrac{z}{1+tz^2} & \text{ for type I}\\[5mm]
\dfrac{z^{1/2}}{1-\sqrt{-t}z}  & \text{ for type II}
\end{cases}.
\end{equation}
Therefore, the quantities \eqref{symmetrywavefunction} and \eqref{typetwosymmetrywavefunction} are
invariant under $z_i\longleftrightarrow z^{-1}_i$.

\section{Some calculations for the dual wavefunctions}

Let us show some computations
about the proof of the dual wavefunctions.
To prove
\eqref{dualmaintheorem}
in Theorem \ref{maintheoremstatement},
one must show that the following functions
$\overline{F}^{\rm I}_{M,N}(z_N,\dots,z_1|\gamma_1,\dots,\gamma_M|\overline{x_1},\dots,\overline{x_N})$ satisfy \eqref{remark}
in Proposition \ref{dualproposition}.

\begin{align}
\overline{F}^{\rm I}_{M,N}&(z_N,\dots,z_1|\gamma_1,\dots,\gamma_M|\overline{x_1},\dots,\overline{x_N})
\nonumber \\
&:=t^{N(M-N)}
\prod_{j=1}^N z_j^{j-1-N}(1+tz_j^2)
\prod_{1 \le j<k \le N}(1+tz_j z_k)(1+tz_j z_k^{-1})
\nonumber \\
&\quad 
\times sp_{\overline{\lambda}} 
( \{ tz \}_N |\{-\overline{\alpha} \}|\{-\overline{\gamma} \}) 
\biggl|_{z_i \to z_{N+1-i}} \nonumber \\
&=
\frac{t^{N(M-N)}
\prod_{j=1}^N z_{N+1-j}^{j-1-N}(1+tz_{N+1-j}^2)
\prod_{1 \le k<j \le N}(1+tz_j z_k)(1+tz_j z_k^{-1})
}{
(-1)^N \prod_{j=1}^N (tz_j)^{j-1-N}(1-t^2z_j^2)
\prod_{1 \le j<k \le N}(1-t^2 z_j z_k)(-1+z_j z_k^{-1})
} \nonumber \\
&\quad \times\sum_{\sigma \in S_N}
\sum_{\tau_1=\pm 1,\dots,\tau_N=\pm 1} (-1)^\sigma (-1)^{|\tau|}
\prod_{j=1}^N \prod_{k=0}^{\overline{x_j}-1}
\left\{-\alpha_k+(1-\alpha_k \gamma_k)(tz_{\sigma(j)})^{\tau_{\sigma(j)}}
\right\} \nonumber \\
&\quad \times
\prod_{j=1}^N \prod_{k=\overline{x_j}+1}^{M}
\left\{1+\gamma_k (tz_{\sigma(j)})^{\tau_{\sigma(j)}}\right\}
\prod_{j=1}^N \prod_{k=1}^{M}
\left\{1+\gamma_k (tz_{\sigma(j)})^{-\tau_{\sigma(j)}}\right\}. \label{appendixequation}
\end{align}
To check this, one first rewrites \eqref{appendixequation} as

\begin{align}
\overline{F}^{\rm I}_{M,N}&
(z_N,\dots,z_1|\gamma_1,\dots,\gamma_M|\overline{x_1},\dots,\overline{x_N})
\nonumber \\
&=
\frac{t^{N(M-N)}
\prod_{j=1}^N (1+tz_{j}^2)
\prod_{1 \le j<k \le N}(1+tz_j z_k)(1+tz_j^{-1} z_k)
}{
t^{-N(N+1)/2} (-1)^N \prod_{j=1}^N (1-t^2z_j^2)
\prod_{1 \le j<k \le N}(1-t^2 z_j z_k)(1-z_j^{-1} z_k)
} \nonumber \\
&\quad \times \sum_{\sigma \in S_N}
\sum_{\tau_1=\pm 1,\dots,\tau_N=\pm 1} (-1)^\sigma (-1)^{|\tau|}
\prod_{j=1}^N \prod_{k=0}^{\overline{x_j}-1}
\left(-\alpha_k+(1-\alpha_k \gamma_k)(tz_{\sigma(j)})^{\tau_{\sigma(j)}}
\right) \nonumber \\
&\quad \times
\prod_{j=1}^N \prod_{k=\overline{x_j}+1}^{M}
\left\{1+\gamma_k (tz_{\sigma(j)})^{\tau_{\sigma(j)}}\right\}
\prod_{j=1}^N \prod_{k=1}^{M}
\left\{1+\gamma_k (tz_{\sigma(j)})^{-\tau_{\sigma(j)}}\right\}. 
\label{appendixequationtwo}
\end{align}

Specializing $\gamma_M=-t^{-1} z_N^{-1}$ and $\overline{x_N}=M$,
only the summands satisfying $\sigma(N)=N$, $\tau_N=+1$ in 
\eqref{appendixequationtwo} survive,
from which we find that
$\overline{F}^{\rm I}_{M,N}
(z_1,\dots,z_N|\gamma_1,\dots,\gamma_M|\overline{x_1},\dots,\overline{x_N})
|_{\gamma_M=-t^{-1} z_N^{-1}}$ can be rewritten as
\begin{align}
&\overline{F}^{\rm I}_{M,N}
(z_1,\dots,z_N|\gamma_1,\dots,\gamma_M|\overline{x_1},\dots,\overline{x_N})
|_{\gamma_M=-t^{-1} z_N^{-1}}
\nonumber \\
&\,\, =-t^M \frac{1+tz_N^2}{1-t^2 z_N^2}
\frac{\prod_{j=1}^{N-1}(1+tz_j z_N)(1+tz_j^{-1} z_N)}
{\prod_{j=1}^{N-1}(1-t^2 z_j z_N) (1-z_j^{-1} z_N)} \nonumber \\
&\quad \times
\frac{t^{(N-1)(M-1-(N-1))}
\prod_{j=1}^{N-1} (1+tz_j^2)
\prod_{1 \le j<k \le N-1}(1+tz_j z_k)(1+tz_j^{-1} z_k)
}{t^{-(N-1)N/2}
(-1)^{N-1} \prod_{j=1}^{N-1} (1-t^2 z_j^2)
\prod_{1 \le j<k \le N-1}(1-t^2 z_j z_k)(1-z_j^{-1} z_k)
}
\nonumber \\
&\quad \times \sum_{\sigma \in S_{N-1}}
\sum_{\tau_1=\pm 1,\dots,\tau_{N-1}=\pm 1}
(-1)^\sigma (-1)^{|\tau|} \nonumber \\
&\quad \times \prod_{j=1}^{N-1} \prod_{k=0}^{\overline{x_j}-1}
\left\{-\alpha_k+(1-\alpha_k \gamma_k)(tz_{\sigma(j)})^{\tau_{\sigma(j)}}
\right\}
\prod_{k=0}^{M-1}
\left\{-\alpha_k+t(1-\alpha_k \gamma_k)z_N
\right\}
\nonumber \\
&\quad \times
\prod_{j=1}^{N-1} \prod_{k=\overline{x_j}+1}^{M-1}
\left\{1+\gamma_k (tz_{\sigma(j)})^{\tau_{\sigma(j)}}\right\}
\prod_{j=1}^{N-1}
\left\{1-\frac{(tz_{\sigma(j)})^{\tau_{\sigma(j)}}}{tz_N} \right\}
\nonumber \\
&\quad \times
\left(1-\frac{1}{t^2 z_N^2} \right)
\prod_{k=1}^{M-1}
\left(1+\frac{\gamma_k}{t z_N} \right)
\prod_{j=1}^{N-1}
\left\{1-\frac{(tz_{\sigma(j)})^{-\tau_{\sigma(j)}}}{tz_N} \right\}
\prod_{j=1}^{N-1} \prod_{k=1}^{M-1}
\left\{1+\gamma_k (tz_{\sigma(j)})^{-\tau_{\sigma(j)}}\right\}. 
\label{appendixequationthree}
\end{align}
Using the identity
\begin{align}
\prod_{j=1}^{N-1}
\left\{1-\frac{(tz_{\sigma(j)})^{\tau_{\sigma(j)}}}{tz_N} \right\}
\left\{1-\frac{(tz_{\sigma(j)})^{-\tau_{\sigma(j)}}}{tz_N} \right\}
=
\prod_{j=1}^{N-1}
\left(1-\frac{z_j}{z_N} \right)
\left(1-\frac{1}{t^2 z_N z_{j}} \right),
\end{align}
to simplify \eqref{appendixequationthree}
leads to
\begin{align}
\overline{F}^{\rm I}_{M,N}&
(z_1,\dots,z_N|\gamma_1,\dots,\gamma_M|\overline{x_1},\dots,\overline{x_N})
|_{\gamma_M=-t^{-1} z_N^{-1}}
\nonumber \\
&=
\prod_{j=1}^{N} \left(1+ \frac{1}{t z_N z_j} \right)
\prod_{j=1}^{N-1} \left(1+\frac{z_j}{tz_N} \right)
\prod_{j=0}^{M-1}\left\{t(1-\alpha_j \gamma_j)z_N-\alpha_j\right\}
\prod_{j=1}^{M-1} (t+\gamma_j z_N^{-1})
\nonumber \\
&\quad \times
\frac{t^{(N-1)(M-1-(N-1))}
\prod_{j=1}^{N-1} (1+tz_{j}^2)
\prod_{1 \le j<k \le N-1}(1+tz_j z_k)(1+tz_j^{-1} z_k)
}{
t^{-(N-1)N/2} (-1)^{N-1} \prod_{j=1}^{N-1} (1-t^2z_j^2)
\prod_{1 \le j<k \le N-1}(1-t^2 z_j z_k)(1-z_j^{-1} z_k)
} \nonumber \\
&\quad \times\sum_{\sigma \in S_{N-1}}
\sum_{\tau_1=\pm 1,\dots,\tau_{N-1}=\pm 1} (-1)^\sigma (-1)^{|\tau|}
\prod_{j=1}^{N-1} \prod_{k=0}^{\overline{x_j}-1}
\left\{-\alpha_k+(1-\alpha_k \gamma_k)(tz_{\sigma(j)})^{\tau_{\sigma(j)}}
\right\} \nonumber \\
&\quad \times
\prod_{j=1}^{N-1} \prod_{k=\overline{x_j}+1}^{M-1}
\left\{1+\gamma_k (tz_{\sigma(j)})^{\tau_{\sigma(j)}}\right\}
\prod_{j=1}^{N-1} \prod_{k=1}^{M-1}
\left\{1+\gamma_k (tz_{\sigma(j)})^{-\tau_{\sigma(j)}}\right\}
\nonumber \\
&=\prod_{j=1}^{N} \Bigg(1+ \frac{1}{t z_N z_j} \Bigg)
\prod_{j=1}^{N-1} \Bigg(1+\frac{z_j}{tz_N} \Bigg)
\prod_{j=0}^{M-1} \left\{t(1-\alpha_j \gamma_j)z_N-\alpha_j\right\}
\prod_{j=1}^{M-1} (t+\gamma_j z_N^{-1})
\nonumber \\
&\quad
\times \overline{F}^{\rm I}_{M-1,N-1}(z_{N-1},\dots,z_{1}|\gamma_1,\dots,\gamma_{M-1}|\overline{x_1},\dots,\overline{x_{N-1}}),
\end{align}
and we have shown
$\overline{F}^{\rm I}_{M,N}(z_N,\dots,z_1|\gamma_1,\dots,\gamma_M|\overline{x_1},\dots,\overline{x_N})$ satisfy \eqref{remark}
in Proposition \ref{dualproposition}.

Similarly, to prove
\eqref{typetwodualmaintheorem}
in Theorem \ref{typetwomaintheoremstatement},
one has to check that the following functions
$\overline{F}^{\rm II}_{M,N}(z_N,\dots,z_1|\gamma_1,\dots,\gamma_M|\overline{x_1},\dots,\overline{x_N})$ satisfy \eqref{typetworemark}
in Proposition \ref{typetwodualproposition}.

\begin{align}
\overline{F}^{\rm II}_{M,N}&(z_N,\dots,z_1|\gamma_1,\dots,\gamma_M|\overline{x_1},\dots,\overline{x_N})
\nonumber \\
&:=t^{N(M-N)}
\prod_{j=1}^N z_j^{j-1/2-N}(1-\sqrt{-t}z_j)
\prod_{1 \le j<k \le N}(1+tz_j z_k)(1+tz_j z_k^{-1})
\nonumber \\
&\quad \times o^{-}_{\overline{\lambda}} ( \{ tz \}_N |\{-\alpha \}|\{-\gamma \}|t)
\biggl|_{z_i \to z_{N+1-i}} \nonumber \\
&=
\frac{t^{N(M-N)}
\prod_{j=1}^N z_{N+1-j}^{j-1/2-N}(1-\sqrt{-t}z_{N+1-j})
\prod_{1 \le k<j \le N}(1+tz_j z_k)(1+tz_j z_k^{-1})
}{
(-1)^N \prod_{j=1}^N (tz_j)^{j-1-N}(1-t^2z_j^2)
\prod_{1 \le j<k \le N}(1-t^2 z_j z_k)(-1+z_j z_k^{-1})
} \nonumber \\
&\quad \times\sum_{\sigma \in S_N}
\sum_{\tau_1=\pm 1,\dots,\tau_N=\pm 1} (-1)^\sigma (-1)^{|\tau|}
\prod_{j=1}^N \left\{ (tz_{\sigma(j)})^{\tau_{\sigma(j)}}-\sqrt{-t}\right\}
\nonumber \\
&\quad \times \prod_{j=1}^N \prod_{k=1}^{\overline{x_j}-1}
\left\{-\alpha_k+(1-\alpha_k \gamma_k)(tz_{\sigma(j)})^{\tau_{\sigma(j)}}
\right\} \nonumber \\
&\quad \times
\prod_{j=1}^N \prod_{k=\overline{x_j}+1}^{M}
\left\{1+\gamma_k (tz_{\sigma(j)})^{\tau_{\sigma(j)}}\right\}
\prod_{j=1}^N \prod_{k=1}^{M}
\left\{1+\gamma_k (tz_{\sigma(j)})^{-\tau_{\sigma(j)}}\right\}. 
\label{typetwoappendixequation}
\end{align}

We again rewrite \eqref{typetwoappendixequation}
as
\begin{align}
\overline{F}^{\rm II}_{M,N}&
(z_N,\dots,z_1|\gamma_1,\dots,\gamma_M|\overline{x_1},\dots,\overline{x_N})
\nonumber \\
&=
\frac{t^{N(M-N)}
\prod_{j=1}^N z_j^{1/2}(1-\sqrt{-t}z_{j})
\prod_{1 \le j<k \le N}(1+tz_j z_k)(1+tz_j^{-1} z_k)
}{
t^{-N(N+1)/2} (-1)^N \prod_{j=1}^N (1-t^2z_j^2)
\prod_{1 \le j<k \le N}(1-t^2 z_j z_k)(1-z_j^{-1} z_k)
} \nonumber \\
&\quad \times\sum_{\sigma \in S_N}
\sum_{\tau_1=\pm 1,\dots,\tau_N=\pm 1} (-1)^\sigma (-1)^{|\tau|}
\prod_{j=1}^N \left\{ 
(tz_{\sigma(j)})^{\tau_{\sigma(j)}}-\sqrt{-t} \right\}
\nonumber \\
&\quad \times\prod_{j=1}^N \prod_{k=1}^{\overline{x_j}-1}
\left\{-\alpha_k+(1-\alpha_k \gamma_k)(tz_{\sigma(j)})^{\tau_{\sigma(j)}}
\right\} \nonumber \\
&\quad \times
\prod_{j=1}^N \prod_{k=\overline{x_j}+1}^{M}
\left\{1+\gamma_k (tz_{\sigma(j)})^{\tau_{\sigma(j)}}\right\}
\prod_{j=1}^N \prod_{k=1}^{M}
\left\{1+\gamma_k (tz_{\sigma(j)})^{-\tau_{\sigma(j)}}\right\}. \label{typetwoappendixequationtwo}
\end{align}

Only the summands satisfying $\sigma(N)=N$, $\tau_N=+1$ in 
\eqref{typetwoappendixequationtwo} survive
after the substitution $\gamma_M=-t^{-1} z_N^{-1}$ and 
$\overline{x_N}=M$
from which we finds

\begin{align}
&\overline{F}^{\rm II}_{M,N}
(z_1,\dots,z_N|\gamma_1,\dots,\gamma_M|\overline{x_1},\dots,\overline{x_N})
|_{\gamma_M=-t^{^1} z_N^{^1}}
\nonumber \\
&\,\,=-t^M \frac{z_N^{1/2} (1-\sqrt{-t}z_N)}{1-t^2 z_N^2}
\frac{\prod_{j=1}^{N-1}(1+tz_j z_N)(1+tz_j^{-1} z_N)}
{\prod_{j=1}^{N-1}(1-t^2 z_j z_N) (1-z_j^{-1} z_N)} \nonumber \\
&\quad \times
\frac{t^{(N-1)(M-1-(N-1))}
\prod_{j=1}^{N-1} z_j^{1/2} (1-\sqrt{-t}z_j)
\prod_{1 \le j<k \le N-1}(1+tz_j z_k)(1+tz_j^{-1} z_k)
}{t^{-(N-1)N/2}
(-1)^{N-1} \prod_{j=1}^{N-1} (1-t^2 z_j^2)
\prod_{1 \le j<k \le N-1}(1-t^2 z_j z_k)(1-z_j^{-1} z_k)
}
\nonumber \\
&\quad \times \sum_{\sigma \in S_{N-1}}
\sum_{\tau_1=\pm 1,\dots,\tau_{N-1}=\pm 1}
(-1)^\sigma (-1)^{|\tau|}
(tz_N-\sqrt{-t})
\prod_{j=1}^{N-1} \left\{
(tz_{\sigma(j)})^{\tau_{\sigma(j)}}-\sqrt{-t} \right\}
\nonumber \\
&\quad \times \prod_{j=1}^{N-1} \prod_{k=1}^{\overline{x_j}-1}
\left\{-\alpha_k+(1-\alpha_k \gamma_k)(tz_{\sigma(j)})^{\tau_{\sigma(j)}}
\right\}
\prod_{k=1}^{M-1}
\left\{-\alpha_k+t(1-\alpha_k \gamma_k)z_N
\right\}
\nonumber \\
&\quad \times
\prod_{j=1}^{N-1} \prod_{k=\overline{x_j}+1}^{M-1}
\left\{1+\gamma_k (tz_{\sigma(j)})^{\tau_{\sigma(j)}}\right\}
\prod_{j=1}^{N-1}
\left\{1-\frac{(tz_{\sigma(j)})^{\tau_{\sigma(j)}}}{tz_N} \right\}
\nonumber \\
&\quad \times
\left(1-\frac{1}{t^2 z_N^2} \right)
\prod_{k=1}^{M-1}
\left(1+\frac{\gamma_k}{t z_N} \right)
\prod_{j=1}^{N-1}
\left\{1-\frac{(tz_{\sigma(j)})^{-\tau_{\sigma(j)}}}{tz_N} \right\}
\prod_{j=1}^{N-1} \prod_{k=1}^{M-1}
\left\{1+\gamma_k (tz_{\sigma(j)})^{-\tau_{\sigma(j)}}\right\}.
\label{typetwoappendixequationthree}
\end{align}
We again use the identity
\begin{align}
\prod_{j=1}^{N-1}
\left\{1-\frac{(tz_{\sigma(j)})^{\tau_{\sigma(j)}}}{tz_N} \right\}
\left\{1-\frac{(tz_{\sigma(j)})^{-\tau_{\sigma(j)}}}{tz_N} \right\}
=
\prod_{j=1}^{N-1}
\left(1-\frac{z_j}{z_N} \right)
\left(1-\frac{1}{t^2 z_N z_{j}} \right),
\end{align}
to simplify \eqref{typetwoappendixequationthree} as
\begin{align}
&\overline{F}^{\rm II}_{M,N}(z_1,\dots,z_N|\gamma_1,\dots,\gamma_M|\overline{x_1},\dots,\overline{x_N})
|_{\gamma_M=-t^{-1} z_N^{-1}}
\nonumber \\
&\,\,=
-\sqrt{-t}z_N^{1/2} \prod_{j=1}^{N} \left(1+ \frac{1}{t z_N z_j} \right)
\prod_{j=1}^{N-1} \left(1+\frac{z_j}{tz_N} \right)
\prod_{j=1}^{M-1} \left\{t(1-\alpha_j \gamma_j)z_N-\alpha_j\right\}
\prod_{j=1}^{M-1} (t+\gamma_j z_N^{-1})
\nonumber \\
&\quad \times
\frac{t^{(N-1)(M-1-(N-1))}
\prod_{j=1}^{N-1} z_j^{1/2}(1-\sqrt{-t}z_{j})
\prod_{1 \le j<k \le N-1}(1+tz_j z_k)(1+tz_j^{-1} z_k)
}{
t^{-(N-1)N/2} (-1)^{N-1} \prod_{j=1}^{N-1} (1-t^2z_j^2)
\prod_{1 \le j<k \le N-1}(1-t^2 z_j z_k)(1-z_j^{-1} z_k)
} \nonumber \\
&\quad \times\sum_{\sigma \in S_{N-1}}
\sum_{\tau_1=\pm 1,\dots,\tau_{N-1}=\pm 1} (-1)^\sigma (-1)^{|\tau|}
\prod_{j=1}^{N-1} \left\{ 
(tz_{\sigma(j)})^{\tau_{\sigma(j)}}-\sqrt{-t} \right\}
\nonumber \\
&\quad \times\prod_{j=1}^{N-1} \prod_{k=0}^{\overline{x_j}-1}
\left\{-\alpha_k+(1-\alpha_k \gamma_k)(tz_{\sigma(j)})^{\tau_{\sigma(j)}}
\right\} \nonumber \\
&\quad \times
\prod_{j=1}^{N-1} \prod_{k=\overline{x_j}+1}^{M-1}
\left\{1+\gamma_k (tz_{\sigma(j)})^{\tau_{\sigma(j)}}\right\}
\prod_{j=1}^{N-1} \prod_{k=1}^{M-1}
\left\{1+\gamma_k (tz_{\sigma(j)})^{-\tau_{\sigma(j)}}\right\} \nonumber
\end{align}
\begin{align}
&\,\,=-\sqrt{-t}z_N^{1/2} \prod_{j=1}^{N} \left(1+ \frac{1}{t z_N z_j} \right)
\prod_{j=1}^{N-1} \left(1+\frac{z_j}{tz_N} \right)
\prod_{j=1}^{M-1} \left\{t(1-\alpha_j \gamma_j)z_N-\alpha_j\right\}
\prod_{j=1}^{M-1} (t+\gamma_j z_N^{-1})
\nonumber \\
&\quad \times \overline{F}^{\rm II}_{M-1,N-1}
(z_{N-1},\dots,z_{1}|\gamma_1,\dots,\gamma_{M-1}|\overline{x_1},
\dots,\overline{x_{N-1}}),
\end{align}
and we have shown
$\overline{F}^{\rm II}_{M,N}(z_N,\dots,z_1|\gamma_1,\dots,\gamma_M|\overline{x_1},\dots,\overline{x_N})$ satisfy \eqref{typetworemark}
in Proposition \ref{typetwodualproposition}.

\end{appendix}

\end{document}